\documentclass[%
  a4paper,
  UKenglish,
  cleveref,
  autoref,
  thm-restate,
  numberwithinsect,
  pdfa,
]{lipics-v2021}

\pdfoutput=1

\usepackage{booktabs}   

\usepackage{mathtools}
\usepackage{stmaryrd}
\usepackage{tikz-cd}
\usepackage{bussproofs}  \EnableBpAbbreviations
\usepackage{bm}

\usepackage{macros}

\usepackage[%
  full,		
]{optional}

\opt{full}{\hideLIPIcs}

\newcommand\colin[1]{\opt{draft}{{\color{red}\sffamily C: #1}}}

\usepackage{comment}
\opt{fullproof}{\newcommand\FLAGPROOF{}}
\ifdefined\FLAGPROOF
  
\else
  \excludecomment{fullproof}
\fi

\opt{draft}{\newcommand\FLAGDRAFTPROOF{}}
\ifdefined\FLAGDRAFTPROOF
  
\else
  \excludecomment{draftproof}
\fi

\opt{full}{\newcommand\FLAGFULL{}}
\ifdefined\FLAGFULL
  
\else
  \excludecomment{full}
\fi

\opt{draft}{\newcommand\FLAGDRAFT{}}
\ifdefined\FLAGDRAFT
  
\else
  \excludecomment{draft}
\fi

\ifdefined\FLAGDRAFT
  \newenvironment{colinpar}{\begin{color}{red}\sffamily C:}{\end{color}}
\else
  \excludecomment{colinpar}
\fi

\ifdefined\FLAGDRAFT
  
\else
  \excludecomment{adampar}
\fi

\title{A Complete Finitary Refinement Type System for Scott-Open Properties} 

\author
  {Colin Riba}
  {ENS de Lyon, CNRS, Université Claude Bernard Lyon 1, LIP, UMR 5668, 69342,
  Lyon cedex 07, France
  \and \url{https://perso.ens-lyon.fr/colin.riba/}}
  {colin.riba@ens-lyon.fr}{https://orcid.org/0009-0005-8858-7532}{}

\author
  {Adam Donadille}
  {ENS de Lyon, CNRS, Université Claude Bernard Lyon 1, LIP, UMR 5668, 69342,
  Lyon cedex 07, France}{}{}{}

\authorrunning{C. Riba and A. Donadille}

\Copyright{Colin Riba and Adam Donadille}

\ccsdesc[500]{Theory of computation~Denotational semantics}
\ccsdesc[500]{Theory of computation~Logic and verification}

\keywords{Domain Theory, Temporal Logic, Refinement Types.}

\opt{short}{\relatedversion{}}
\opt{short}{\relatedversiondetails
  [cite=rd26full]
  {Full Version}
  {https://arxiv.org/abs/2601.23082}} 

\funding{This work was partially supported by the French ANR PRCI project \emph{ATQUA}, ANR-25-CE48-1428-01.}

\nolinenumbers

\EventEditors{Frank Pfenning}
\EventNoEds{1}
\EventLongTitle{11th International Conference on Formal Structures for Computation and Deduction (FSCD 2026)}
\EventShortTitle{FSCD 2026}
\EventAcronym{FSCD}
\EventYear{2026}
\EventDate{July 20--23, 2026}
\EventLocation{Lisbon, Portugal}
\EventLogo{}
\SeriesVolume{378}
\ArticleNo{15}

\begin{document}

\maketitle

\begin{abstract}
We are interested in proving input-output properties of functions
that handle infinite data such as streams or non-wellfounded trees.
We provide a finitary refinement type system which is (sound and)
complete for Scott-open properties defined in a fixpoint-like logic.
Working on top of Abramsky's Domain Theory in Logical Form,
we build from the well-known fact that the Scott domains interpreting
recursive types are spectral spaces. 
The usual symmetry between Scott-open and compact-saturated sets 
is reflected in logical polarities:
positive formulae allow for least fixpoints and define Scott-open sets,
while negative formulae allow for greatest fixpoints and define
compact-saturated sets.
A realizability implication with the expected (contra)variance on polarities
allows for non-trivial input-output properties to be formulated as 
positive formulae on function types.
\end{abstract}

\section{Introduction}
\label{sec:intro}

We are interested in input-output
specifications of functions that handle infinite 
data, such as streams or non-wellfounded trees.
Consider the function $\cnt \colon \Stream \Bool \to \Stream \Nat$
defined as $\cnt = \cntrec\ \term 0$,
where
\[
\begin{array}{r @{~} l @{~~} l}
  \cntrec
& :
& \Nat
  ~\longto~
  \Stream \Bool
  ~\longto~
  \Stream \Nat
\\

  \cntrec\ n\ (b \Colon x)
& =
& \term{if}~ b 
  ~\term{then}~ (\cntrec\ (\term{1+}n)\ x)
  ~\term{else}~ n \Colon (\cntrec\ n\ x)
\end{array}
\]

\noindent
Given a stream of Booleans $x$, the stream of natural
numbers $(\cnt\ x)$ retains,
for each occurrence of $\false$ in $x$,
the number of $\true$'s seen so far in $x$.
It is clear that
if there are infinitely many $\true$'s and infinitely many $\false$'s in $x$,
then $(\cnt\ x)$ contains arbitrary large numbers.

Working with a simply-typed $\lambda$-calculus with sums and recursive types
(a.k.a.\ $\FPC$)~\cite{pierce02book},
we are interested in specifications 
interpreted in a denotational semantics:
since a stream (as opposed to e.g.\ an integer)
is an inherently infinite object, a specification for (e.g.) $\cnt$
should hold for any stream whatsoever, and not only
for those definable in a given syntax.

This leads us to consider properties on infinite datatypes in Scott domains.
Logics on top of domains are known since quite a long time.
%
%
We resort to Abramsky's paradigm of
``Domain Theory in Logical Form'' \cite{abramsky87lics,abramsky91apal},
which allows one to systematically generate a logic from a 
domain representing a type.
These logics are obtained by Stone duality,
the core of a
rich interplay between domain theory, logic and topology
(cf e.g.~\cite{johnstone82book,vickers89book,zhang91book,ac98book,cz00csl,vickers07chapter,goubault13book,gg24book}).

In a previous work~\cite{rk25wollic} (based on~\cite{bk03ic}),
we considered a refinement type system
which is sound and complete for saturated (i.e.\ upward-closed)
properties on Scott domains.
This strictly extends~\cite{abramsky87lics,abramsky91apal}.
Indeed,
saturated properties cover least and greatest fixpoints of
negation-free formulae in modal $\mu$-calculi on the corresponding types.
On streams, this includes usual modalities of Linear Temporal Logic ($\LTL$),
such as $\Diam$ (\emph{eventually}, i.e.\ ``there is a position such that'')
and $\Box$ (\emph{always}, i.e.\ ``for all positions'')~\cite{rs24jfla}.
Logics like $\LTL$, $\CTL$ (Computation Tree Logic) or
modal $\mu$-calculi
are widely used to formulate properties on infinite objects
(see e.g.~\cite{bk08book,hr07chapter,bs07chapter}).

On the other hand, only compact-open properties are available
in~\cite{abramsky87lics,abramsky91apal}.
But for each compact-open set $K$ of streams, there is some $n \in \NN$
such that membership to $K$ is completely determined by prefixes of length $n$.
This excludes proper uses of $\Diam$ and $\Box$.
However, the systems in~\cite{bk03ic,rk25wollic} are \emph{infinitary},
in the sense that some rules have infinitely many premises.

In this paper, we devise a \emph{finitary} refinement type system
which is sound and complete for a family of Scott-open properties
allowing for specifications of infinitary behaviours.
This extends~\cite{abramsky87lics,abramsky91apal}.
Types are refined by a (negation-free) fixpoint-like logic which
relies on the well-known fact that the Scott domains interpreting
recursive types are \emph{spectral spaces}. 

In a spectral space
(and more generally in a \emph{stably compact space}~\cite{lawson11mscs}),
there is a symmetry between the open sets and the compact-saturated ones.
We reflect this symmetry in the logic:
we have \emph{positive} formulae
which are closed under least fixpoints and define Scott-open sets,
and \emph{negative} formulae which are closed under greatest fixpoints and define
compact-saturated sets.
Input-output specifications of functions are written using a realizability implication
$\realto$
with the expected (contra)variance w.r.t.\ polarities.

On streams, positive formulae are closed under
$\Diam$, while negative formulae are closed under $\Box$.
Given a positive formula $\varphi$ and a negative formula $\psi$,
we have a positive formula $\Box \psi \realto \Diam \varphi$
which selects those functions that take a stream
whose elements all satisfy $\psi$ to a stream with at least one element
satisfying $\varphi$.
Consider now the composite
\begin{equation*}
\begin{array}{l @{~} l @{~~} l}
  \cnt \comp \bft
& :
& \Tree \Bool
  \longarrow
  \Stream \Nat
\end{array}
\end{equation*}

\noindent
where $\bft \colon \Tree \PTbis \arrow \Stream \PTbis$
is a breadth-first tree traversal.
Negative formulae on trees are closed under
the usual $\CTL$-like modalities $\forall\Box$
(\emph{invariantly}, i.e.\ ``for all nodes'')
and $\exists\Box$
(\emph{potentially always},
i.e.\ ``there is an infinite path on which for all nodes'').
Thanks to the realizability implication, 
for each $n \in \NN$ there is a positive
formula which expresses the following specification for $\cnt \comp \bft$:
Given a totally defined tree of Booleans,
if the children of the root have an infinite path with only $\true$'s
and an infinite path with only $\false$'s,
then $\cnt \comp \bft$ produces a stream which contains a totally
defined number $\geq n$.

We prove a \emph{Positive Completeness} result: our type system
derives all sound positive specifications of $\lambda$-terms
(thus including the above positive specifications for $\cnt \comp \bft$).

Positive completeness implies that
checking positive formulae on $\lambda$-terms is semi-decidable.
This is in line with the intuition that 
membership to a Scott-open set is akin to a reachability problem:
evaluate the (semantics of the) $\lambda$-term until knowing that the open set is met.
Our refinement type system essentially amounts to a
syntactic approach to this reachability problem.
We also note that checking positive formulae on $\lambda$-terms
is \emph{not} decidable.

Besides, an adaptation of Higher-Order Model Checking
to input-output specifications for a specific format
of higher-order tree transducers (weaker than our $\FPC$)
has been shown to be undecidable in~\cite{ktu10popl}.

At the technical level, 
our logic allows for prenex normal forms.
To this end, least and greatest fixpoints are decomposed
as finite iterations indexed by variables which are
existentially and universally quantified, respectively.
This builds on ideas in~\cite{jr21esop}
(which is based on guarded recursion and states no completeness result).

\subparagraph{Organisation of the paper.}
Our simply-typed $\lambda$-calculus is discussed in~\S\ref{sec:pure}.
The logic is introduced in~\S\ref{sec:log},
while refinement types are presented in~\S\ref{sec:reft}.
Section~\ref{sec:sem} covers denotational semantics,
soundness and positive completeness,
as well as some necessary material on spectral spaces.
Finally, the conclusion (\S\ref{sec:conc}) discusses some further works
and compare with more related works.
\opt{short}{Proofs are available in the Appendices of the full version~\cite{rd26full}.}%
\opt{full}{Proofs are available in the Appendices.}

\begin{colinpar}
NOTES/TODO:
\begin{itemize}
\item NOTE: Ideas from~\cite{jr21esop}
regarding continuity properties of formulae (as well as iteration terms).

\item NOTE: Maybe cite~\cite{smyth83icalp} as a historical reference?.
\end{itemize}
\end{colinpar}

\section{The Pure System}
\label{sec:pure}

The \emph{pure types}
(notation $\PT,\PTbis,\dots$) are
the closed types over the grammar
\[
\begin{array}{r @{~~} c @{~~} l}
    \PT
&   \bnf
&   \Unit
\gs \PT \times \PT
\gs \PT \arrow \PT
\gs \TV
\gs \rec \TV.\PT 
\gs \PT + \PT
\end{array}
\]

\noindent
where $\TV$ ranges over an infinite supply of \emph{type variables},
and where $\rec\TV.\PT$ binds $\TV$ in $\PT$.
 
We consider $\lambda$-terms from the grammar
\[
\begin{array}{r @{~~} r @{~~} l}
    M,N 
&   \bnf
&   \pair{}
\gs \pair{M,N}
\gs \pi_1(M)
\gs \pi_2(M)
\gs x
\gs \lambda x.M
\gs M N
\gs \fix x.M
\\

&   \mid
&   \fold(M)
\gs \unfold(M)
\gs \inj_1(M)
\gs \inj_2(M)
\gs \cse\ M\ \copair{x_1 \mapsto N_1 \mid x_2 \mapsto N_2}
\end{array}
\]

\noindent
The term formers $\fold, \unfold, \pi_1, \pi_2, \inj_1, \inj_2$
are often written in curried form,
so that e.g.\ $(\fold M)$ stands for $\fold(M)$.
We write $M \comp N$ for $\lambda x. M(N\, x)$.

Terms are typed with judgements of the form
$\Env \thesis M \colon \PT$, where the \emph{typing context} $\Env$ is a list
$x_1\colon \PTbis_1, \dots, x_n\colon \PTbis_n$ with $x_i \neq x_j$
if $i \neq j$.
The typing rules are those in Figure~\ref{fig:pure},
together with the \emph{basic rules} in Figure~\ref{fig:reft:basic}
\emph{restricted to pure types}
(i.e.\ with $\Env$ as above and $\RT$, $\RTbis$ pure types).
Of course, each type $\PT$ is inhabited
by the ``diverging'' term $\Omega_\PT \deq \fix x.x \colon \PT$.

\begin{figure}[t!]
\centering
\(
\begin{array}{c}

\dfrac{}
  {\Env \thesis \pair{} \colon \Unit}

\quad~~

\dfrac{\Env,x\colon \PT \thesis M \colon \PT}
  {\Env \thesis \fix x.M \colon \PT}

\quad~~

\dfrac{\Env \thesis M \colon \PT[\rec\TV.\PT/\TV]}
  {\Env \thesis \fold(M) \colon \rec\TV.\PT}

\quad~~

\dfrac{\Env \thesis M \colon \rec\TV.\PT}
  {\Env \thesis \unfold(M) \colon \PT[\rec\TV.\PT/\TV]}

\\\\

\dfrac{\Env \thesis M \colon \PT_i}
  {\Env \thesis \inj_i(M) \colon \PT_1 + \PT_2}

\qquad

\dfrac{\Env \thesis M \colon \PT_1 + \PT_2
  \qquad \Env, x_1 \colon \PT_1 \thesis N_1 \colon \PTbis
  \qquad \Env, x_2 \colon \PT_2 \thesis N_2 \colon \PTbis}
  {\Env \thesis \cse\ M\ \copair{x_1 \mapsto N_1 \mid x_2 \mapsto N_2} \colon \PTbis}

\end{array}
\)
\caption{Typing rules for pure types, where $i \in \{1,2\}$.%
\label{fig:pure}}
\end{figure}

\begin{example}
\label{ex:pure}
The type $\Bool \deq \Unit + \Unit$ represents Booleans, 
with $\true \deq \inj_1(\pair{})$
and $\false \deq \inj_2(\pair{})$.
The type of natural numbers
is $\Nat \deq \rec \TV.\, \Unit + \TV$.
For each $n \in \NN$, we have a term $\term n \colon \Nat$
defined as $\term 0 \deq \fold(\inj_1 \pair{})$
and $\term{n+1} \deq \term{(1+)} \term n$,
where $\term{(1+)} \colon \Nat \to \Nat$
is the successor function $\lambda x. \fold(\inj_2 x)$.
We can represent ``$+\infty$'' as
$\term{inf} \deq \fix x.\, \term{1+}x \colon \Nat$.

The type of streams over $\PTbis$ is
$\Stream\PTbis \deq \rec\TV.\, \PTbis \times \TV$.
It is equipped with the constructor
\(
  \Cons
  \deq
  \lambda h .\lambda t.\fold \pair{h,t}
  \colon
  \PTbis \arrow \Stream\PTbis \arrow \Stream \PTbis
\).
We use the infix notation $(M \Colon N)$ for $(\Cons M\, N)$.
The usual \emph{head} and \emph{tail} functions
are
$\hd \deq \lambda s.\, \pi_1 (\unfold s) \colon \Stream\PTbis \arrow \PTbis$
and
$\tl \deq \lambda s.\, \pi_2 (\unfold s) \colon \Stream\PTbis \arrow \Stream\PTbis$.

The type of binary trees over $\PTbis$ is
$\Tree\PTbis \deq \rec\TV.\, \PTbis \times (\TV \times \TV)$.
The constructor
$\Node \colon \PTbis \arrow \Tree\PTbis \arrow \Tree\PTbis \arrow \Tree\PTbis$
and the destructors
$\lbl \colon \Tree\PTbis \arrow \PTbis$
and
$\lft, \rght \colon \Tree\PTbis \arrow \Tree\PTbis$
are defined similarly as $\Cons$, $\hd$, $\tl$ on streams.
\lipicsEnd
\end{example}

\begin{table}[t!]
\caption{Functions on Streams and Trees.%
\label{tab:ex}}
\centering
\(
\begin{array}{c}
\toprule

\begin{array}{r @{~} l @{~~} l}
  \bm{\cntrec}
& :
& \Nat
  \longarrow
  \Stream \Bool
  \longarrow
  \Stream \Nat
\\

& \deq
& \fix g. \lambda n. \lambda x.~
  \term{if}~ (\hd x)
  ~\term{then}~ (g\ (\term{1+} n)\ (\tl x))
  ~\term{else}~ (n \Colon (g\ n\ (\tl x)))
\end{array}

\\\\

\begin{array}{r l l l}
  \bm{\cnt}
& \deq
& (\cntrec\ \term 0)
  \colon
& \Stream\Bool
  \longarrow
  \Stream\Nat
\end{array}

\\\midrule

\begin{array}{c !{\qquad} c}

\begin{array}[t]{l @{~} l @{~~} l}

  \bm{\extract}
& :
& \Cont \PTbis \longarrow \PTbis
\\
& \deq 
& \fix e.\lambda c. \unfold\, c\ e
\end{array}

&

\begin{array}[t]{l @{~} l @{~~} l}
  \bm{\Over}
& :
& \Cont \PTbis
\\
& \deq
& \fix c. \fold(\lambda k.\ k\ c)
\end{array}

\end{array}

\\\\

\begin{array}{r @{~} l @{~~} l}
  \bm{\bftrec}
& :
& \Tree\PTbis
  \longarrow
  \Cont (\Stream\PTbis)
  \longarrow
  \Cont (\Stream\PTbis)
\\

& \deq
& \fix g.\lambda t.\lambda c.
  \fold \left(
    \lambda k.~
    (\lbl t) \Colon
    \left( \unfold\, c \ \big( k \comp (g (\lft t)) \comp (g (\rght t)) \big) \right)
  \right)
\end{array}

\\\\ 

\begin{array}{r l l l}
  \bm{\bft}
& \deq
& \lambda t.~ \extract (\bftrec\ t\ \Over)
  \colon
& \Tree \PTbis
  \longarrow
  \Stream \PTbis

\end{array}

\\\bottomrule
\end{array}
\)
\end{table}

\begin{example}
\label{ex:pure:fun}
Table~\ref{tab:ex} defines functions on streams and trees
mentioned in~\S\ref{sec:intro}.
The stream function $\cnt$ was discussed in \S\ref{sec:intro}.
The notation
$\term{if}~ M ~\term{then}~ N_{1} ~\term{else}~ N_{2}$
stands for the term
$\cse\ M\ \copair{\_ \mapsto N_1 \mid \_ \mapsto N_2}$.
On trees, the function $\bft$ implements Martin Hofmann's breadth-first traversal
(see e.g.~\cite{bms19types,jr21esop,rk25wollic,kw26popl}).
It uses the recursive type
$\Cont \PTbis \deq \rec\TV.\, (\TV \arrow \PTbis) \arrow \PTbis$.
Note that the recursive type corresponding to $\Cont \PTbis$
in~\cite{bms19types,jr21esop,kw26popl}
has an additional constructor
(in place of our fixpoint $\Over \colon \Cont\PTbis$).
See Example~\ref{ex:scott:stream-tree} (\S\ref{sec:sem:pure}).
\lipicsEnd
\end{example}

\section{The Logic}
\label{sec:log}

The main ingredient of this paper is the logic 
we use to annotate pure types when forming refinement types.
This logic is in essence \emph{manysorted}:
for each pure type $\PT$ we have formulae $\varphi$ \emph{of type $\PT$},
notation $\varphi \in \Lang^\sign(\PT)$,
where $\sign$ is $+$ for positive formulae and $-$ for a negative ones.

Positive and negative formulae are in essence closed under
least and greatest fixpoints, respectively,
and thus may contain fixpoint variables $\FP$.
But least fixpoints $(\mu \FP)\varphi$
and greatest fixpoints $(\nu \FP)\psi$ are actually not directly available.
They are defined 
as $(\exists \itvar)(\finmu^\itvar \FP)\varphi$
and $(\forall \itvarbis)(\finnu^\itvarbis \FP)\psi$,
where $(\finmu^\itvar \FP)\varphi$ and $(\finnu^\itvarbis \FP)\psi$
represent bounded iterations
$\varphi^{\itvar}(\False)$ and $\psi^{\itvarbis}(\True)$,
respectively.
This decomposition allows for a notion of \emph{prenex} form
(Proposition~\ref{prop:ded:prenex}, \S\ref{sec:deduction}).
A fixpoint logic with iteration variables was already
considered in~\cite{sd03fossacs}.

Formulae are presented in~\S\ref{sec:formulae},
while deduction is discussed in~\S\ref{sec:deduction}.

\begin{colinpar}
NOTES/TODO:
\begin{itemize}
\item IDEA: \emph{Bounded iteration}
intermediate between finitary modal language $\Lang^\omega$
and languages $\Lang^+$, $\Lang^-$
with quantification.
\end{itemize}
\end{colinpar}

\subsection{Formulae}
\label{sec:formulae}

\begin{figure}[t!]
\begin{subfigure}{\textwidth}
\centering
\(
\begin{array}{c}

\dfrac{}
  {\True \in \Lang^\sign(\FPEnv; \PT)}

\qquad

\dfrac{\varphi,\psi \in \Lang^\sign(\FPEnv; \PT)}
  {\varphi \land \psi \in \Lang^\sign(\FPEnv; \PT)}

\qquad

\dfrac{}
  {\False \in \Lang^\sign(\FPEnv; \PT)}

\qquad

\dfrac{\varphi,\psi \in \Lang^\sign(\FPEnv; \PT)}
  {\varphi \lor \psi \in \Lang^\sign(\FPEnv; \PT)}

\end{array}
\)
\caption{Propositional connectives.%
\label{fig:form:prop}}
\end{subfigure}

\begin{subfigure}{\textwidth}
\centering
\(
\begin{array}{c}

\dfrac{}
  {\form{\pair{}} \in \Lang^\sign(\FPEnv; \Unit)}

\qquad

\dfrac{\varphi \in \Lang^\sign(\FPEnv; \PT_1)}
  {\form{\pi_1}\varphi \in \Lang^\sign(\FPEnv; \PT_1 \times \PT_2)}

\qquad

\dfrac{\varphi \in \Lang^\sign(\FPEnv; \PT_2)}
  {\form{\pi_2}\varphi \in \Lang^\sign(\FPEnv; \PT_1 \times \PT_2)}

\\\\

\dfrac{\varphi \in \Lang^\sign(\FPEnv; \PT[\rec \TV.\PT/\TV])}
  {\form\fold \varphi \in \Lang^\sign(\FPEnv; \rec \TV.\PT)}

\qquad

\dfrac{\varphi \in \Lang^\sign(\FPEnv; \PT_1)}
  {\form{\inj_1}\varphi \in \Lang^\sign(\FPEnv; \PT_1 + \PT_2)}

\qquad

\dfrac{\varphi \in \Lang^\sign(\FPEnv; \PT_2)}
  {\form{\inj_2}\varphi \in \Lang^\sign(\FPEnv; \PT_1 + \PT_2)}

\end{array}
\)
\caption{Modalities.%
\label{fig:form:mod}}
\end{subfigure}

\begin{subfigure}{\textwidth}
\centering
\(
\begin{array}{c}

\dfrac{(\FP \colon \PT) \in \FPEnv}
  {\FP \in \Lang^\sign(\FPEnv; \PT)}  

\quad~~

\dfrac{\varphi \in \Lang^\sign(\FPEnv;\PT)
  \quad~
  \FP \notin \dom(\FPEnv)}
  {\varphi \in \Lang^\sign(\FPEnv, \FP \colon \PTbis; \PT)}

\quad~~

\dfrac{\varphi \in \Lang^\sign(\FPEnv, \FP\colon\PT; \PT)}
  {(\finmu^t \FP)\varphi \in \Lang^\sign(\FPEnv; \PT)}

\quad~~

\dfrac{\varphi \in \Lang^\sign(\FPEnv, \FP\colon \PT; \PT)}
  {(\finnu^t \FP)\varphi \in \Lang^\sign(\FPEnv; \PT)}
\end{array}
\)
\caption{Bounded iteration, where $t$ is an iteration term.%
\label{fig:form:iter}}
\end{subfigure}

\begin{subfigure}{\textwidth}
\centering
\(
\begin{array}{c}

\dfrac{\varphi \in \Lang^+(\FPEnv;\PT)
  \quad~
  \itvar \Pos \varphi}
  {(\exists \itvar)\varphi \in \Lang^+(\FPEnv;\PT)}

\qquad

\dfrac{\varphi \in \Lang^-(\FPEnv;\PT)
  \quad~
  \itvar \Neg \varphi}
  {(\forall \itvar)\varphi \in \Lang^-(\FPEnv;\PT)}

\qquad

\dfrac{\psi \in \Lang^{-\sign}(;\PTbis)
  \quad~
  \varphi \in \Lang^\sign(;\PT)}
  {\psi \realto \varphi \in \Lang^\sign(;\PTbis \arrow \PT)}

\end{array}
\)
\caption{Quantification
and realizability implication,
with $-\sign$ defined as $-\pm = \pm$, $-+ = -$ and $-- = +$.%
\label{fig:form:pol}}
\end{subfigure}

\caption{Formation rules for $\Lang^\sign$
(the predicates $\Pos$ and $\Neg$ are defined in Figure~\ref{fig:posneg}).%
\label{fig:form}}
\end{figure}

\begin{figure}[t!]
\centering
\(
\begin{array}{c}

\dfrac{\itvar \notin \FV(\varphi)}
  {\itvar \Pos \varphi}

\qquad

\dfrac{\itvar \Pos \varphi, \psi}
  {\itvar \Pos \varphi \star \psi}

\qquad

\dfrac{\itvar \Pos \varphi}
  {\itvar \Pos \form\modgen \varphi}

\qquad

\dfrac{\itvar \Neg \psi
  \qquad
  \itvar \Pos \varphi}
  {\itvar \Pos \psi \realto \varphi}

\qquad

\dfrac{\itvar \Pos \varphi}
  {\itvar \Pos (\exists \itvarbis)\varphi}

\\\\

\dfrac{\itvar \notin \FV(\varphi)}
  {\itvar \Neg \varphi}

\qquad

\dfrac{\itvar \Neg \varphi, \psi}
  {\itvar \Neg \varphi \star \psi}

\qquad

\dfrac{\itvar \Neg \varphi}
  {\itvar \Neg \form\modgen \varphi}

\qquad

\dfrac{\itvar \Pos \psi
  \qquad
  \itvar \Neg \varphi}
  {\itvar \Neg \psi \realto \varphi}

\qquad

\dfrac{\itvar \Neg \varphi}
  {\itvar \Neg (\forall \itvarbis)\varphi}

\\\\

\dfrac{\itvar \Pos \varphi}
  {\itvar \Pos (\finmu^t \FP)\varphi}

\qquad

\dfrac{\itvar \Pos \varphi
  \quad~~
  \itvar \notin \FV(t)}
  {\itvar \Pos (\finnu^t \FP)\varphi}

\qquad

\dfrac{\itvar \Neg \varphi}
  {\itvar \Neg (\finnu^t \FP)\varphi}

\qquad

\dfrac{\itvar \Neg \varphi
  \quad~~
  \itvar \notin \FV(t)}
  {\itvar \Neg (\finmu^t \FP)\varphi}

\end{array}
\)
\caption{The predicates $\Pos$ and $\Neg$,
where $\modgen \in \{\pi_1,\pi_2, \inj_1, \inj_2, \fold\}$,
and where $\star$ is $\land$ or $\lor$.%
\label{fig:posneg}}
\end{figure}

We assume given infinitely many \emph{iteration variables}
$\itvar, \itvarbis, \text{etc}$.
\emph{Iteration terms},
intended to range over natural numbers,
are given by the grammar
\[
\begin{array}{r @{~~} c @{~~} l}
     t
&    \bnf
&    \itvar
\gss 0
\gss t+1
\end{array}
\]

We also assume given \emph{fixpoint variables}
$\FP, \FPbis, \text{etc}$.
\emph{Fixpoint contexts} $\FPEnv$ are
lists of the form
$\FP_1\colon \PTbis_1,\dots,\FP_n \colon \PTbis_n$
where $\FP_i \neq \FP_j$ whenever $i \neq j$.
We let $\dom(\FPEnv)$ be $\{\FP_1,\dots,\FP_n\}$.
Figure~\ref{fig:form} gathers formation rules for sets $\Lang^\sign(\FPEnv;\PT)$
indexed by fixpoint contexts and types.
We write $\Lang^\sign(\PT)$ or $\Lang^\sign(;\PT)$ 
for $\Lang^\sign(\FPEnv;\PT)$
when the fixpoint context $\FPEnv$ is empty.

\begin{definition}[Formulae]
\label{def:form}
The sets $\Lang^\pm(\FPEnv;\PT)$, $\Lang^+(\FPEnv;\PT)$
and $\Lang^-(\FPEnv;\PT)$
of respectively \emph{neutral} ($\pm$),
\emph{positive} ($+$) and \emph{negative} ($-$)
\emph{formulae}
are defined by the rules in Figure~\ref{fig:form} with $\sign \in \{\pm,+,-\}$,
and using the predicates $\Pos$ and $\Neg$ defined in Figure~\ref{fig:posneg}.

$\Lang^\omega(\PT)$ is the fragment of $\Lang^\pm(\PT)$
\emph{without} the rules 
in Figure~\ref{fig:form:iter}
(but with $\realto$ from Figure~\ref{fig:form:pol}).
$\Lang^\land(\PT)$ is the fragment of $\Lang^\omega(\PT)$
without binary disjunctions $\lor$ (but with $\False$).
\end{definition}

Note that $\Lang^\pm$, $\Lang^\omega$ and $\Lang^\land$ are stable under
the realizability implication $\realto$.
In fact,
$\Lang^\pm(\FPEnv;\PT) = \Lang^+(\FPEnv;\PT) \cap \Lang^-(\FPEnv;\PT)$
and $\Lang^\omega(\PT) \sle \Lang^\pm(\PT)$
are defined by the rules
in Figure~\ref{fig:form:prop}--\ref{fig:form:iter}
and Figure~\ref{fig:form:prop}--\ref{fig:form:mod} respectively,
together with
\[
\begin{array}{c !{\qquad\qquad} c}

\dfrac{\psi \in \Lang^\pm(;\PTbis)
  \qquad
  \varphi \in \Lang^\pm(;\PT)}
  {\psi \realto \varphi \in \Lang^\pm(;\PTbis \arrow \PT)}

&

\dfrac{\psi \in \Lang^\omega(\PTbis)
  \qquad
  \varphi \in \Lang^\omega(\PT)}
  {\psi \realto \varphi \in \Lang^\omega(\PTbis \arrow \PT)}

\end{array}
\]

\noindent
On the other hand, $\realto$ is contravariant in its first argument
w.r.t.\ $+$ and $-$:
\begin{equation}
\label{eq:form:realto}
\begin{array}{c c}

\dfrac{\psi \in \Lang^-(;\PTbis)
  \qquad
  \varphi \in \Lang^+(;\PT)}
  {\psi \realto \varphi \in \Lang^+(;\PTbis \arrow \PT)}

\qquad\qquad

\dfrac{\psi \in \Lang^+(;\PTbis)
  \qquad
  \varphi \in \Lang^-(;\PT)}
  {\psi \realto \varphi \in \Lang^-(;\PTbis \arrow \PT)}

\end{array}
\end{equation}

\noindent
In any case, 
$\varphi \realto \psi$
contains no free fixpoint variables.
%
Also, formulae $\varphi \in \Lang^{\omega}(\PT)$
have no fixpoint variables nor iteration variables at all.
When possible, we write $\Lang$ for $\Lang^\sign$.

Note that in Figure~\ref{fig:posneg}, we could have had a rule stating
that $\itvar \Pos (\forall \itvarbis)\varphi$ if $\itvar \Pos \varphi$.
But this would have been useless for the left rule in Figure~\ref{fig:form:pol},
since a positive formula has no positive occurrence of a $(\forall \itvarbis)$.
A dual remark applies to $\itvar \Neg (\exists \itvarbis)\varphi$.


The semantics of formulae is discussed in~\S\ref{sec:sem:log}.
Their intended meaning is as follows.
The formula $\psi \realto \varphi \in \Lang(\PTbis \arrow \PT)$
is intended to select those $M \colon \PTbis \arrow \PT$
such that $\varphi$ holds on $M N \colon \PT$ whenever $\psi$ holds on $N \colon \PTbis$.
The formula $\form\fold \varphi$ holds on $M$ if, and only if,
$\varphi$ holds on $\unfold(M)$.
For $i = 1,2$,
the formula $\form{\pi_i} \varphi$
selects those $M \colon \PT_1 \times \PT_2$
such that $\varphi$ holds on $\pi_i(M)$,
while $\form{\inj_i} \varphi$ selects those
$\inj_j(M) \colon \PT_1 + \PT_2$ such that $j = i$ and
$\varphi$ holds on $M$.
Hence
$\form{\inj_1}\True \land \form{\inj_2}\True$ is a contradiction
and
$\form{\inj_1}\True \lor \form{\inj_2}\True$ fails on $\Omega_{\PT_1 + \PT_2}$.



\begin{example}
\label{ex:form:mod}
We derive some composite modalities in $\Lang^\omega$ on the datatypes of
Example~\ref{ex:pure}.
On $\Bool$, we have formulae $\form\true \deq \form{\inj_1}\form{\pair{}}$
and $\form\false \deq \form{\inj_2}\form{\pair{}}$ which hold
on $\true, \false \colon \Bool$, respectively.
The formula $\form\true \lor \form\false$,
which expresses totality at type $\Bool$,
fails on $\Omega_\Bool$.

On $\Nat$, we have $\form{\term 0} \deq \form\fold \form{\inj_1} \form{\pair{}}$
and a composite modality
$\form{\term{1+}}\varphi \deq \form\fold \form{\inj_2}\varphi$.
By taking $\form{\term{n+1}} \deq \form{\term{1+}}\form{\term n}$,
we have for each $n \in \NN$ a formula $\form{\term{n}} \in \Lang(\Nat)$
which holds on $\term n \colon \Nat$.
We also have a composite modality $\form{\geq \term{n}} \varphi$,
defined as $\form{\geq \term{0}}\varphi \deq \varphi$
and $\form{\geq \term{n+1}} \deq \form{\term{1+}}\form{\geq \term{n}} \varphi$,
such that $\form{\geq \term{n}}\True$
holds on $\term{inf} \colon \Nat$ and on $\term m \colon \Nat$ for all $m \geq n$.

On streams $\Stream\PTbis$, 
the composite modalities $\form\hd$ and $\form\tl$
are defined as $\form\hd \psi \deq \form\fold \form{\pi_1} \psi$
and $\form\tl \varphi \deq \form\fold \form{\pi_2} \varphi$.
Given $\psi \in \Lang(\PTbis)$ and $\varphi \in \Lang(\Stream\PTbis)$,
the formulae $\form\hd\psi \in \Lang(\Stream\PTbis)$
and $\form\tl\varphi \in \Lang(\Stream\PTbis)$
select those streams $M$ such that $\psi$ holds on $(\hd M)$
and such that $\varphi$ holds on $(\tl M)$, respectively.
We write $\Next\varphi$ for $\form\tl\varphi$.

Similarly, on trees $\Tree\PTbis$
one can define $\form\lbl$, $\form\lft$ and $\form\rght$
such that
$\form\lbl\psi \in \Lang(\Tree\PTbis)$
given $\psi \in \Lang(\PTbis)$,
and
$\form\lft\varphi, \form\rght\varphi \in \Lang(\Tree\PTbis)$
given $\varphi \in \Lang(\Tree\PTbis)$.
We also have composite modalities
$\exists\Next \varphi \deq \form\lft \varphi \lor \form\rght\varphi$
and
$\forall\Next \varphi \deq \form\lft \varphi \land \form\rght\varphi$
which allow for taking, respectively,
disjunctions and conjunctions over the children of a node.
\lipicsEnd
\end{example}

We now discuss some uses of fixpoint and iteration variables in
$\Lang^\pm$, $\Lang^+$ and $\Lang^-$.
Let $\varphi \in \Lang(\FPEnv, \FP\colon\PT;\PT)$.
Given $\psi \in \Lang(\FPEnv;\PT)$ and $n \in \NN$,
we define $(\FP.\varphi)^n(\psi) \in \Lang(\FPEnv;\PT)$ as
\begin{equation}
\label{eq:form:syntfun}
\begin{array}{l l l !{\qquad\text{and}\qquad} l l l}
  (\FP.\varphi)^0(\psi)
& \deq
& \psi

& (\FP.\varphi)^{n+1}(\psi)
& \deq
& \varphi[(\FP.\varphi)^n(\psi) / \FP]
\end{array}
\end{equation}

If $t$ is a closed iteration term representing $n \in \NN$
(i.e.\ $t$ is the $n$th successor of $0$),
then the formulae
$(\finmu^t \FP)\varphi$ and $(\finnu^t \FP)\varphi$
should be read as 
$(\FP.\varphi)^n(\False)$
and
$(\FP.\varphi)^n(\True)$,
respectively
(cf Equation~\eqref{eq:ded:syntfun} in \S\ref{sec:deduction} below).
But in general, the iteration term $t$
may contain free iteration variables.
In fact, we shall see in 
Example~\ref{ex:sem:log:fixpoints} (\S\ref{sec:sem:log})
that the least fixpoint of $\varphi \in \Lang^+(\FP\colon\PT;\PT)$
and the greatest fixpoint of $\psi \in \Lang^-(\FP\colon\PT; \PT)$
can be represented respectively as
\begin{equation}
\label{eq:form:fp}
\begin{array}{l !{\qquad\text{and}\qquad} l}
  (\exists \itvar)(\finmu^\itvar \FP)\varphi \in \Lang^+(\PT)
& (\forall \itvarbis)(\finnu^\itvarbis \FP)\psi \in \Lang^-(\PT)
\end{array}
\end{equation}

\begin{example}
\label{ex:form:fix}
The representation of fixpoints in Equation~\eqref{eq:form:fp}
allows for formulae expressing possibly unbounded behaviours.
For instance, totality at type $\Nat$ is expressed by the \emph{positive} formula
\(
  \form\tot
  \deq
  (\exists \itvar)(\finmu^\itvar \FP)(\form{\term 0} \lor \form{\term{1+}}\FP)
\),
intended to hold
on all $\term n \colon \Nat$ but not on $\term{inf} \colon \Nat$.

On streams $\Stream \PTbis$, the usual $\LTL$-like
\emph{eventually}
modality $\Diam$
is available on \emph{positive} formulae $\varphi \in \Lang^+(\Stream\PTbis)$,
namely
$\Diam \varphi \deq (\exists \itvar)(\finmu^\itvar \FP)(\varphi \lor \Next\FP)$.
The formula $\Diam \varphi$ is intended to hold on those $M \colon \Stream\PTbis$
such that $\varphi$ holds on $(\tl^n M)$ for some $n \in \NN$.

Regarding greatest fixpoints, the $\CTL$-like modalities
$\forall\Box$ and $\exists\Box$ on $\Tree \PTbis$ 
are available for \emph{negative} formulae: given $\psi \in \Lang^-(\Tree\PTbis)$,
we let
\(
  \forall \Box \psi
  \deq
  (\forall \itvarbis)(\finnu^\itvarbis \FP)(\psi \land \forall\Next \FP)
\)
and
\(
  \exists \Box \psi
  \deq
  (\forall \itvarbis)(\finnu^\itvarbis \FP)(\psi \land \exists\Next \FP)
\).
The intended meaning of $\forall\Box\form\lbl\psi \in \Lang^-(\Tree\PTbis)$
is to select those trees of $\PTbis$'s whose node labels all satisfy
$\psi \in \Lang^-(\PTbis)$,
while $\exists\Box\form\lbl\psi \in \Lang^-(\Tree\PTbis)$
asks $\psi$ to hold on all labels in some infinite path.

We thus obtain for each $n \in \NN$
the following positive formula 
\begin{equation}
\label{eq:form:spec}
\begin{array}{l l l}
  \forall\Box\form\lbl(\form\true \lor \form\false)
  ~\land~
  \exists\Next\exists\Box\form\lbl\form\true
  ~\land~
  \exists\Next\exists\Box\form\lbl\form\false
& \realto
& \Diam\form\hd \form{\geq \term n} \form\tot
\end{array}
\end{equation}

\noindent
in $\Lang^+(\Tree\Bool \arrow \Stream\Nat)$,
which expresses the specification for $\cnt \comp \bft$
devised in~\S\ref{sec:intro}:
Given a totally defined tree of Booleans,
if the children of the root have an infinite path with only $\true$'s
and an infinite path with only $\false$'s,
then $\cnt \comp \bft$ produces a stream which contains a totally
defined number $\geq n$.
Note that
fixpoints of the same polarity can be nested,
e.g.\
$\Diam\form\hd\form{\geq \term{n}}\form\tot \in \Lang^+(\Stream \Nat)$
with $\form\tot \in \Lang^+(\Nat)$.
\lipicsEnd
\end{example}

\begin{remark}[Non-Examples]
\label{rem:form:fix}
Given a \emph{negative} $\psi \in \Lang^-(\Stream\PTbis)$,
we have
the usual \emph{always} modality
\(
  \Box \psi
  \deq
  (\forall \itvarbis)(\finnu^\itvarbis \FP)(\psi \land \Next\FP)
  \in \Lang^-(\Stream\PTbis)
\)
which holds on
$M$ if $\psi$ holds on $(\tl^n M)$ for all $n \in \NN$.
A composite $\Box\Diam\form\hd\varphi$ would require a stream
to have infinitely many elements satisfying $\varphi$.
But $\Box\Diam\form\hd\varphi$ is not a formula since
$\Diam\form\hd\varphi$ is not negative.

Also, in order to 
universally quantify over $n \in \NN$
in the r.-h.s.\ of~\eqref{eq:form:spec},
one could think of
\(
  (\forall \itvarbis)\Diam\form\hd (\finmu^\itvarbis \FP)
  (\form\tot \lor \form{\term{1+}}\FP)
\).
This is not a formula
since $\psi$ has to be negative in $(\forall \itvarbis)\psi$.
\lipicsEnd
\end{remark}

\subsection{Deduction}
\label{sec:deduction}

\begin{figure}[p!]
\begin{subfigure}{\textwidth}
\centering
\(
\begin{array}{c}

\dfrac{}
  {\varphi \,\thesis\, \varphi}

\qquad

\dfrac{\psi \,\thesis\, \theta
  \qquad
  \theta \,\thesis\, \varphi}
  {\psi \,\thesis\, \varphi}

\qquad

\ax{D}
\dfrac{}
  {\psi \land (\varphi \lor \theta)
  \,\thesis\,
  (\psi \land \varphi) \lor (\psi \land \theta)}

\\\\

\dfrac{\psi \,\thesis\, \varphi_1
  \qquad
  \psi \,\thesis\, \varphi_2}
  {\psi \,\thesis\, \varphi_1 \land \varphi_2}

\qquad

\dfrac{\psi_1 \,\thesis\, \varphi}
  {\psi_1 \land \psi_2 \,\thesis\, \varphi}

\qquad

\dfrac{\psi_1 \,\thesis\, \varphi}
  {\psi_1 \land \psi_2 \,\thesis\, \varphi}

\qquad

\dfrac{}
  {\psi \,\thesis\, \True}

\\\\

\dfrac{\psi \,\thesis\, \varphi_1}
  {\psi \,\thesis\, \varphi_1 \lor \varphi_2}

\qquad

\dfrac{\psi \,\thesis\, \varphi_2}
  {\psi \,\thesis\, \varphi_1 \lor \varphi_2}

\qquad

\dfrac{\psi_1 \,\thesis\, \varphi
  \qquad
  \psi_2 \,\thesis\, \varphi}
  {\psi_1 \lor \psi_2 \,\thesis\, \varphi}

\qquad

\dfrac{}
  {\False \,\thesis\, \varphi}

\end{array}
\)
\caption{Propositional rules.%
\label{fig:ded:prop}}
\end{subfigure}

\begin{subfigure}{\textwidth}
\centering
\(
\begin{array}{c}

\ax{I}
\dfrac{}
  {\form{\inj_1}\varphi \land \form{\inj_2}\psi \,\thesis\, \False}

\qquad

\dfrac{\psi \,\thesis\, \varphi}
  {\form\modgen \psi \,\thesis\, \form\modgen \varphi}

\qquad

\dfrac{\psi' \,\thesis\, \psi
  \qquad
  \varphi \,\thesis\, \varphi'}
  {\psi \realto \varphi
  \,\thesis\,
  \psi' \realto \varphi'}

\\\\

\dfrac{}
  {\form\modgen \varphi_1 \land \form\modgen \varphi_2
  \,\thesis\,
  \form\modgen(\varphi_1 \land \varphi_2)}

\qquad

\dfrac{\modgen \notin \{\inj_1,\inj_2\}}
  {\True \,\thesis\, \form{\modgen} \True}

\qquad

\dfrac{}
  {(\forall \itvar)\form{\modgen}\varphi
  \,\thesis\,
  \form{\modgen} (\forall \itvar)\varphi}

\\\\\

\dfrac{}
  {\form\modgen(\varphi_1 \lor \varphi_2)
  \,\thesis\,
  \form\modgen\varphi_1 \lor \form\modgen\varphi_2}

\qquad

\dfrac{}
  {\form\modgen \False \,\thesis\, \False}

\qquad

\dfrac{}
  {\form{\modgen} (\exists \itvar)\varphi
  \,\thesis\,
  (\exists \itvar)\form{\modgen}\varphi}

\end{array}
\)
\caption{Modal rules,
where $\modgen \in \{\pi_1, \pi_2, \inj_1, \inj_2, \fold\}$.%
\label{fig:ded:mod}}
\end{subfigure}

\begin{subfigure}{\textwidth}
\centering
\(
\begin{array}{c}

\dfrac{\psi \,\thesis\, \varphi[t / \itvar]}
  {\psi \,\thesis\ (\exists \itvar)\varphi}

\qquad

\ax{\exists L}
\dfrac{\theta \land \varphi \,\thesis\, \psi
  \qquad
  \itvar \notin \FV(\psi, \theta)}
  {\theta \land (\exists \itvar)\varphi \,\thesis\, \psi}

\qquad

\ax{\exists M}
\dfrac{}
  {(\exists \itvar)(\exists \itvarbis)\varphi
  \,\thesis\,
  (\exists \itvar)\varphi[\itvar / \itvarbis]}

\\\\

\dfrac{\varphi[t / \itvar] \,\thesis\, \psi}
  {(\forall \itvar)\varphi \,\thesis\, \psi}

\qquad

\ax{\forall R}
\dfrac{\psi \,\thesis\, \varphi \lor \theta
  \qquad
  \itvar \notin \FV(\psi, \theta)}
  {\psi \,\thesis\, (\forall \itvar)\varphi \lor \theta}

\qquad

\ax{\forall M}
\dfrac{}
  {(\forall \itvar)\varphi[\itvar / \itvarbis]
  \,\thesis\,
  (\forall \itvar)(\forall \itvarbis)\varphi}

\\\\

\ax{\fingen/\exists}
\dfrac{\itvar \notin \FV(t)}
  {(\fingen^t \FP)(\exists \itvar)\varphi
  \,\thesis\,
  (\exists \itvar)(\fingen^t \FP)\varphi}

\qquad

\ax{\forall/\fingen}
\dfrac{\itvar \notin \FV(t)}
  {(\forall \itvar)(\fingen^t \FP)\varphi
  \,\thesis\,
  (\fingen^t \FP)(\forall \itvar)\varphi}

\end{array}
\)
\caption{Quantifier rules,
where $\fingen \in \{\finmu, \finnu\}$.%
\label{fig:ded:quant}}
\end{subfigure}

\begin{subfigure}{\textwidth}
\centering
\(
\begin{array}{c}

\dfrac{}
  {(\finmu^0 \FP)\varphi \,\thesis\, \False}

\quad~~

\dfrac{(\finmu^t \FP)\varphi \,\thesis\, \psi}
  {(\finmu^{t+1} \FP)\varphi \,\thesis\, \varphi[\psi / \FP]}

\quad~~

\dfrac{}
  {\varphi[(\finmu^t \FP)\varphi / \FP] \,\thesis\, (\finmu^{t+1} \FP)\varphi}

\quad~~

\dfrac{\psi \,\thesis\, \varphi}
  {(\finmu^t \FP)\psi \,\thesis\, (\finmu^t \FP)\varphi}

\\\\

\dfrac{}
  {\True \,\thesis\, (\finnu^0 \FP)\varphi}

\quad~~

\dfrac{\psi \,\thesis\, (\finnu^t \FP)\varphi}
  {\varphi[\psi / \FP] \,\thesis\, (\finnu^{t+1} \FP)\varphi}

\quad~~

\dfrac{}
  {(\finnu^{t+1} \FP)\varphi \,\thesis\, \varphi[(\finnu^t \FP)\varphi / \FP]}

\quad~~

\dfrac{\psi \,\thesis\, \varphi}
  {(\finnu^t \FP)\psi \,\thesis\, (\finnu^t \FP)\varphi}

\end{array}
\)
\caption{Bounded iteration.%
\label{fig:ded:iter}}
\end{subfigure}

\begin{subfigure}{\textwidth}
\centering
\(
\begin{array}{c}

\ax{WF}
\dfrac{\itvar \notin \FV(\varphi)}
  {(\forall \itvar)\psi \realto \varphi
  \,\thesis\,
  (\exists \itvar)(\psi \realto \varphi)}

\qquad

\ax{CC}
\dfrac{\psi \in \Lang^\pm
  \qquad
  \itvar \notin \FV(\psi)}
  {\psi \realto (\exists \itvar)\varphi
  \,\thesis\,
  (\exists \itvar)(\psi \realto \varphi)}

\\\\

\ax{\exists/{\realto}}
\dfrac{\itvar \notin \FV(\varphi)}
  {(\forall \itvar)(\psi \realto \varphi)
  \,\thesis\,
  (\exists \itvar)\psi \realto \varphi}

\qquad

\ax{{\realto}/\forall}
\dfrac{\itvar \notin \FV(\psi)}
  {(\forall \itvar)(\psi \realto \varphi)
  \,\thesis\,
  \psi \realto (\forall \itvar)\varphi}

\\\\

\ax{{\realto}/\land}
\dfrac{}
  {(\psi \realto \varphi_1) \land (\psi \realto \varphi_2)
  \,\thesis\,
  \psi \realto (\varphi_1 \land \varphi_2)}

\qquad

\dfrac{}
  {\True \,\thesis\, (\psi \realto \True)}

\\\\

\ax{\lor/{\realto}}
\dfrac{}
  {(\psi_1 \realto \varphi) \land (\psi_1 \realto \varphi)
  \,\thesis\,
  (\psi_1 \lor \psi_2) \realto \varphi}

\qquad

\dfrac{}
  {\True \,\thesis\, (\False \realto \varphi)}

\\\\

\ax{{\realto}/\lor}
\dfrac{\delta \in \Lang^\land}
  {\delta \realto (\varphi_1 \lor \varphi_2)
  \,\thesis\,
  (\delta \realto \varphi_1) \lor (\delta \realto \varphi_2)}

\qquad

\ax{C}
\dfrac{\C(\delta)
  \qquad
  \delta \in \Lang^\land}
  {(\delta \realto \False) \,\thesis\, \False}

\end{array}
\)
\caption{Rules for the realizability implication $\realto$
(the predicate $\C$ is defined in Figure~\ref{fig:consist}).%
\label{fig:ded:realto}}
\end{subfigure}

\caption{Deduction rules.%
\label{fig:ded}}
\end{figure}

\begin{figure}[t!]
\centering
\(
\begin{array}{c}

\dfrac
  {}
  {\C(\form{\pair{}})}

\qquad

\dfrac{\C(\varphi)}
  {\C(\form\fold \varphi)}

\qquad

\dfrac{\C(\varphi) 
  \qquad
  \C(\psi)}
  {\C(\form{\pi_1}\varphi \land \form{\pi_2}\psi)}

\qquad

\dfrac{\C(\varphi)}
  {\C(\form{\inj_1}\varphi)}

\qquad

\dfrac{\C(\varphi)}
  {\C(\form{\inj_2}\varphi)}

\\\\

\dfrac{}{\C(\True)}

\qquad

\dfrac{\C(\psi)
  \quad~~
  \psi \,\thesis\, \varphi
  \quad~~
  \varphi \in \Lang^\land}
  {\C(\varphi)}

\qquad

\dfrac{\begin{array}{l}
  \text{$I$ finite and $\forall i \in I$,}~
  \C(\psi_i) 
  ~\text{and}~
  \C(\varphi_i) ;
  \\
  \text{$\forall J \sle I$,}~
  \bigwedge_{j \in J} \psi_j \thesis \False
  ~~\text{or}~~
  \C\left( \bigwedge_{j \in J} \varphi_j \right)
  \end{array}}
  {\C\left( \bigwedge_{i \in I}(\psi_i \realto \varphi_i) \right)}

\end{array}
\)
\caption{The consistency predicate $\C$.%
\label{fig:consist}}
\end{figure}

Deduction
will enter the refinement type system via subtyping
(Figure~\ref{fig:reft:subtyping} in~\S\ref{sec:reft}).

\begin{definition}[Deduction]
\label{def:ded}
A \emph{sequent} has the form $\psi \thesis \varphi$
where $\varphi,\psi \in \Lang^\sign(\FPEnv; \PT)$.
The \emph{derivable} sequents are defined by 
the rules in Figure~\ref{fig:ded},
using the \emph{consistency predicate} $\C$
defined in Figure~\ref{fig:consist}.
Write $\psi \thesisiff \varphi$ when the sequents
$\psi \thesis \varphi$ and $\varphi \thesis \psi$
are both derivable.
\end{definition}

We discuss some rules in Figure~\ref{fig:ded},
heading to basic properties of the deduction system.
We begin with a notion of normal form for $\Lang^\omega$.
Let $\Lang^{\lor\land}(\PT)$ be the closure of
$\Lang^\land(\PT)$ under disjunctions
(but under no other rules in Figure~\ref{fig:form},
so e.g.\ $\form{\pi_1}(\varphi \lor \psi)$ is \emph{not} in
$\Lang^{\lor\land}$).

As usual,
the converse of the distributive rule $\ax{D}$
in Figure~\ref{fig:ded:prop} is derivable,
and so is the dual law
\(
  \psi \lor (\varphi \land \theta)
  \,\thesisiff\,
  (\psi \lor \varphi) \land (\psi \lor \theta)
\).
Moreover, given $\modgen \in \{\pi_i,\inj_i,\fold\}$
we can derive
\(
  \form{\modgen}(\varphi \lor \psi)
  \thesisiff
  \form{\modgen}\varphi \lor \form{\modgen}\psi
\)
(and similarly for $\land$).
It follows that for each $\varphi \in \Lang^{\lor\land}$,
there is some $\psi \in \Lang^{\lor\land}$
such that $\form{\modgen}\varphi \thesisiff \psi$.
Using the rules $\ax{\lor/{\realto}}$
and $\ax{{\realto}/\lor}$ in Figure~\ref{fig:ded:realto},
we can similarly obtain that for each
$\varphi_1, \varphi_2 \in \Lang^{\lor\land}$
there is some $\psi \in \Lang^{\lor\land}$
such that
$(\varphi_1 \realto \varphi_2) \thesisiff \psi$
(but beware that the rule $\ax{{\realto}/\lor}$ requires
$\delta$ to be in $\Lang^\land$).
Hence:

\begin{restatable}{lemma}{LemOmegaLorLand}
\label{lem:ded:omegalorland}
For each $\varphi \in \Lang^\omega(\PT)$,
there is some $\psi \in \Lang^{\lor\land}(\PT)$
such that $\varphi \thesisiff \psi$.
\end{restatable}

\begin{remark} 
\label{rem:log:ded:cons}
The rule $\ax{C}$ in Figure~\ref{fig:ded:realto}
is a kind of a zero-ary case of $\ax{{\realto}/\lor}$.
Since $\True \thesis (\False \realto \False)$,
the sequent $(\psi \realto \False) \thesis \False$
is derivable only when $\psi$ is consistent.
This is enforced by the consistency predicate $\C(\psi)$,
see Theorem~\ref{thm:sem:log:sound}(\ref{item:sem:log:sound:cons})
in~\S\ref{sec:sem:log}.
Our predicate $\C$ is akin
to~\cite{abramsky87lics,rk25wollic},
but 
differs from the \emph{coprimeness} predicates
of~\cite{abramsky91apal,bk03ic,ac98book}.
Note that the clauses defining $\C$ are positive
(compare with \cite[Figure 3]{bk03ic} and \cite[Figure 10.3]{ac98book}).

At sum types, we have inconsistent formulae which do not
mention $\False$:
The rule $\ax{I}$ in Figure~\ref{fig:ded:mod}
means that it is inconsistent to assert both $\form{\inj_1}\True$
and $\form{\inj_2}\True$. 
Hence, our sum types are \emph{disjoint}, and we do not need the termination
predicates of~\cite{abramsky87lics,abramsky91apal}.

Besides, note that the sequent
$\True \thesis \form{\inj_1}\True \lor \form{\inj_2} \True$
is \emph{not} derivable.
\lipicsEnd
\end{remark}

If $t$ is a closed iteration term,
then using the rules in Figure~\ref{fig:ded:iter}
we can derive
\begin{equation}
\label{eq:ded:syntfun}
\begin{array}{r c l !{\qquad\text{and}\qquad} r c l}
  (\finmu^t \FP)\varphi
& \thesisiff
& (\FP. \varphi)^n(\False)

& (\finnu^t \FP)\varphi
& \thesisiff
& (\FP. \varphi)^n(\True)
\end{array}
\end{equation}

\noindent
where $t$ is the $n$th successor of $0$
(cf Equation~\eqref{eq:form:syntfun} in~\S\ref{sec:formulae}).
This yields:

\begin{restatable}{lemma}{LemNeutralOmega}
\label{lem:ded:omega}
For each \emph{closed} $\varphi \in \Lang^\pm(\PT)$,
there is some $\psi \in \Lang^\omega(\PT)$ such that
$\varphi \thesisiff \psi$.
\end{restatable}

The main fact of this~\S\ref{sec:deduction}
is that quantifiers can be assumed to be in prenex position.
Quantifiers of the same type can be merged using
the rules $\ax{\exists M}$ and $\ax{\forall M}$ in Figure~\ref{fig:ded:quant},
themselves permitted by
the polarity constraints in Figure~\ref{fig:form:pol}.
Formally, the \emph{positive}, resp.\ \emph{negative}, \emph{prenex forms}
are the formulae
$(\exists \itvar)\psi$, resp.\ $(\forall \itvar)\psi$,
with $\psi \in \Lang^\pm$ neutral.

The rules $\ax{WF}$, $\ax{CC}$ in Figure~\ref{fig:ded:realto},
which apply to positive formulae,
are crucial to obtain prenex forms.
%
Note also that the \emph{Frobenius} rules $\ax{\exists L}$ and $\ax{\forall R}$
in Figure~\ref{fig:ded:quant}
make it possible to derive the following
when $\itvar \notin \FV(\psi)$:
\[
\begin{array}{l l l !{\qquad\text{and}\qquad} l l l}
  \psi \land (\exists \itvar)\varphi
& \thesisiff
& (\exists \itvar)(\psi \land \varphi)

& (\forall \itvar)(\varphi \lor \psi)
& \thesisiff
& (\forall \itvar)\varphi \lor \psi
\end{array}
\]

\begin{proposition}[restate = PropFormPrenex, name = ]
\label{prop:ded:prenex}
For each $\varphi \in \Lang^\sign(\FPEnv; \PT)$,
there is a prenex $\psi \in \Lang^\sign(\FPEnv; \PT)$
such that $\varphi \thesisiff \psi$.
\end{proposition}

\section{The Refinement Type System}
\label{sec:reft}

We consider \emph{neutral}, \emph{positive} and \emph{negative}
\emph{refinement types} (or \emph{types}),
notation $\RT^\sign, \RTbis^\sign,\text{etc.}$
with $\sign =\pm,+,-$ respectively.
They are given by the grammar
\[
\begin{array}{r @{~~} c @{~~} l}
    \RT^\sign, \RTbis^\sign
&   \bnf
&   \PT
\gs \reft{\PT \mid \varphi}
\gs \RT^\sign \times \RTbis^\sign
\gs \RTbis^{-\sign} \arrow \RT^\sign
\end{array}
\]

\noindent
where $\PT$ is a pure type and where
$\varphi \in \Lang^\sign(\PT)$ is \emph{closed}.
When possible, we write $\RT$ for $\RT^\sign$.

\begin{example}
\label{ex:reft:spec}
For each $n \in \NN$,
the specification for $\cnt \comp \bft$
is given either by the positive 
refinement type
$\reft{\Tree\Bool \arrow \Stream\Nat \mid \Phi^+}$,
where $\Phi^+$ is the formula in Equation~\eqref{eq:form:spec},
Example~\ref{ex:form:fix},
or by the positive refinement type
$S^+$ defined as
\begin{multline*}
  \big\{ \Tree\Bool ~\big|~
  \forall\Box\form\lbl(\form\true \lor \form\false)
  ~\land~
  \exists\Next\exists\Box\form\lbl\form\true
  ~\land~
  \exists\Next\exists\Box\form\lbl\form\false \big\}
  ~~\longarrow
\\
  \big\{ \Stream\Nat ~\big|~ \Diam\form\hd \form{\geq \term n} \form\tot \big\}
\end{multline*}
\lipicsEnd
\end{example}

We shall consider typing judgements of the form
$\Env \thesis M \colon \RT$,
where $\Env$ is allowed to mention refinement types.
A judgement
$M \colon \reft{\PT \mid \varphi}$
is intended to mean that $M$ is of pure type $\PT$ and satisfies $\varphi$.
\emph{Positive} typing judgements are of the
form $\Env^- \thesis M \colon \RT^+$,
where only negative types are declared in $\Env^-$.

\begin{figure}[t!]
\centering
\(
\begin{array}{c}

\dfrac{}
  {\RT \subtype \RT}

\qquad

\dfrac{\RT \subtype \RTbis
  \qquad
  \RTbis \subtype \RTter}
  {\RT \subtype \RTter}

\qquad

\dfrac{}
  {\RT \subtype \UPT\RT}

\qquad

\dfrac{}
  {\PT \subtype \reft{\PT \mid \True}}

\qquad

\dfrac{\psi \thesis \varphi}
  {\reft{\PT \mid \psi} \subtype \reft{\PT \mid \varphi}}

\\\\

\dfrac{\RT \subtype \RT'
  \qquad
  \RTbis \subtype \RTbis'}
  {\RT \times \RTbis \subtype \RT' \times \RTbis'}

\qquad

\dfrac{}{
  \reft{\PT \mid \varphi}
  \times
  \reft{\PTbis \mid \psi}
  \eqtype
  \reft{\PT \times \PTbis \mid \form{\pi_1}\varphi \land \form{\pi_2}\psi}}

\\\\

\dfrac{\RTbis' \subtype \RTbis
  \qquad
  \RT \subtype \RT'}
  {\RTbis \arrow \RT \subtype \RTbis' \arrow \RT'}

\qquad

\dfrac{}
  {\reft{\PTbis \mid \psi} \arrow \reft{\PT \mid \varphi}
  \eqtype
  \reft{\PTbis \arrow \PT \mid \psi \realto \varphi}}

%
%
%

\end{array}
\)

\caption{Subtyping.%
\label{fig:reft:subtyping}}
\end{figure}

\begin{figure}[t!]
\begin{subfigure}{\textwidth}
\centering
\(
\begin{array}{c}

\dfrac{(x \colon \RT) \in \Env}
  {\Env \thesis x \colon \RT}

\qquad

\dfrac{\Env,\, x \colon \RTbis \thesis M \colon \RT}
  {\Env \thesis \lambda x.M \colon \RTbis \arrow \RT}

\qquad

\dfrac{\Env \thesis M \colon \RTbis \arrow \RT
  \qquad
  \Env \thesis N \colon \RTbis}
  {\Env \thesis M N \colon \RT}

\\\\

\dfrac{\Env \thesis M \colon \RT
  \qquad
  \Env \thesis N \colon \RTbis}
  {\Env \thesis \pair{M,N} \colon \RT \times \RTbis}

\qquad

\dfrac{\Env \thesis M \colon \RT \times \RTbis}
  {\Env \thesis \pi_1(M) \colon \RT}

\qquad

\dfrac{\Env \thesis M \colon \RT \times \RTbis}
  {\Env \thesis \pi_2(M) \colon \RTbis}

\end{array}
\)
\caption{Basic rules.%
\label{fig:reft:basic}}
\end{subfigure}

\begin{subfigure}{\textwidth}
\centering
\(
\begin{array}{c}

\dfrac{
  \Env \thesis M \colon \reft{\PT \mid \varphi_1}
  \qquad
  \Env \thesis M \colon \reft{\PT \mid \varphi_2}}
  {\Env \thesis M \colon \reft{\PT \mid \varphi_1 \land \varphi_2}}

\qquad

\dfrac{\UPT\Env,\, x\colon \PTbis,\, \UPT{\Env'} \thesis M \colon \UPT\RT}
  {\Env,\, x \colon \reft{\PTbis \mid \False},\, \Env' \thesis M \colon \RT}

\\\\

\dfrac{\Env,\, x\colon \reft{\PTbis \mid \psi_1},\, \Env' \thesis M \colon \RT
  \qquad
  \Env,\, x\colon \reft{\PTbis \mid \psi_2},\, \Env' \thesis M \colon \RT}
  {\Env,\, x \colon \reft{\PTbis \mid \psi_1 \lor \psi_2},\, \Env' \thesis M \colon \RT}

\end{array}
\)
\caption{Propositional connectives.%
\label{fig:reft:log}}
\end{subfigure}

\begin{subfigure}{\textwidth}
\centering
\(
\begin{array}{c}

\dfrac{
  \Env' \subtype \Env
  \qquad 
  \RT \subtype \RT'
  \qquad
  \Env \thesis M \colon \RT}
  {\Env' \thesis M \colon \RT'}

\end{array}
\)
\caption{Subtyping (cf Figure~\ref{fig:reft:subtyping}).%
\label{fig:reft:reftsub}}
\end{subfigure}

\begin{subfigure}{\textwidth}
\centering
\(
\begin{array}{c}

\dfrac{\Env \thesis \fix x.M \colon \reft{\PT \mid \psi}
  \qquad
  \Env, x\colon \reft{\PT \mid \psi} \thesis M \colon \reft{\PT \mid \varphi}}
  {\Env \thesis \fix x.M \colon \reft{\PT \mid \varphi}}

\end{array}
\)
\caption{Fixpoints.%
\label{fig:reft:fix}}
\end{subfigure}

\begin{subfigure}{\textwidth}
\centering
\(
\begin{array}{c}

\dfrac{\Env \thesis M \colon \reft{\PT_1 \times \PT_2 \mid \form{\pi_i} \varphi}}
  {\Env \thesis \pi_i(M) \colon \reft{\PT_i \mid \varphi}}

\qquad

\dfrac{\Env \thesis M_i \colon \reft{\PT_i \mid \varphi}
  \qquad
  \Env \thesis M_{3-i} \colon \PT_{3-i}}
  {\Env \thesis \pair{M_1,M_2} \colon \reft{\PT_1 \times \PT_2 \mid \form{\pi_i} \varphi}}

\\\\

\dfrac{}
  {\Env \thesis \pair{} \colon \reft{\Unit \mid \form{\pair{}}}}

\qquad

\dfrac{\Env \thesis M \colon \reft{\PT_i \mid \varphi}}
  {\Env \thesis \inj_i(M) \colon \reft{\PT_1 + \PT_2 \mid \form{\inj_i} \varphi}}

\\\\

\dfrac{
  \Env \thesis M \colon \reft{\PT_1 + \PT_2 \mid \form{\inj_i}\varphi}
  \qquad
  \Env, x_i \colon \reft{\PT_i \mid \varphi} \thesis N_i \colon \RT
  \qquad
  \UPT\Env, x_{3-i} \colon \PT_{3-i} \thesis N_{3-i} \colon \UPT\RT}
  {\Env \thesis \cse\ M\ \copair{x_1 \mapsto N_1 \mid x_2 \mapsto N_2} \colon \RT}

\\\\

\dfrac{\Env \thesis M \colon \reft{\PT[\rec\TV.\PT/\TV] \mid \varphi}}
  {\Env \thesis \fold(M) \colon \reft{\rec\TV.\PT \mid \form\fold \varphi}}

\qquad

\dfrac{\Env \thesis M \colon \reft{\rec\TV.\PT \mid \form\fold \varphi}}
  {\Env \thesis \unfold(M) \colon \reft{\PT[\rec\TV.\PT/\TV] \mid \varphi}}

\end{array}
\)
\caption{Modalities, where $i \in \{1,2\}$.%
\label{fig:reft:mod}}
\end{subfigure}

\caption{Typing rules.%
\label{fig:reft:reftyping}}
\end{figure}

We derive typing judgements $\Env \thesis M \colon \RT$
using the rules in Figure~\ref{fig:reft:reftyping}
augmented with those in Figure~\ref{fig:pure} (\S\ref{sec:pure}).
Deduction on formulae (\S\ref{sec:deduction})
enters the type system in Figure~\ref{fig:reft:reftsub}
via the subtyping relation $\RTbis \subtype \RT$
defined in Figure~\ref{fig:reft:subtyping},
where $\RTbis \eqtype \RT$
stands for the conjunction of $\RTbis \subtype \RT$ and $\RT \subtype \RTbis$.
Note that for each $\PT$ we have $\PT \eqtype \reft{\PT \mid \True}$.
Subtyping is extended to typing contexts:
given $\Env = x_1 \colon \RTbis_1,\dots,x_n \colon\RTbis_n$
and $\Env' = x_1 \colon \RTbis'_1,\dots,x_n \colon \RTbis'_n$,
we let $\Env \subtype \Env'$ when $\RTbis_i \subtype \RTbis'_i$
for all $i =1,\dots,n$.
Note that if $\Env \thesis M \colon \RT$ is derivable
then so is $\UPT\Env \thesis M \colon \UPT\RT$,
where the \emph{underlying pure type} $\UPT\RT$ of $\RT$
is defined by induction from $\UPT\PT \deq \PT$
and $\UPT{\reft{\PT \mid \varphi}} \deq \PT$,
and where $\UPT\Env$ is the extension of $\UPT{\pl}$ to $\Env$.

The rules in Figures~\ref{fig:reft:subtyping} and~\ref{fig:reft:reftyping}
are adaptations of those in~\cite{abramsky91apal,bk03ic,jr21esop,rk25wollic}.
In particular, the crucial rule for $\fix x.M$ in Figure~\ref{fig:reft:fix}
comes from \cite{abramsky91apal}.

\begin{example}
The bottom $\eqtype$-rule in Figure~\ref{fig:reft:subtyping} yields
$S^+ \eqtype \reft{\Tree\Bool \arrow \Stream\Nat \mid \Phi^+}$
in Example~\ref{ex:reft:spec},
as well as the following derived typing rules.
\[
\begin{array}{c}

\dfrac{\Env,x \colon \reft{\PTbis \mid \psi} \thesis M \colon \reft{\PT \mid \varphi}}
  {\Env \thesis \lambda x.M \colon \reft{\PTbis \arrow \PT \mid \psi \realto \varphi}}

\qquad

\dfrac{\Env \thesis M \colon \reft{\PTbis \arrow \PT \mid \psi \realto \varphi}
  \qquad
  \Env \thesis N \colon \reft{\PTbis \mid \psi}}
  {\Env \thesis M N \colon \reft{\PT \mid \varphi}}

\end{array}
\]
\end{example}

\begin{lemma}[restate = LemReft, name = ]
\label{lem:reft}
For each type $\RT^\sign$, there is a $\varphi \in \Lang^\sign(\UPT\RT)$
such that $\RT^\sign \eqtype \reft{\UPT{\RT^\sign} \mid \varphi}$.
\end{lemma}

\begin{remark}
\label{rem:reft:plus}
Refinement types are \emph{not} closed under $+$
since in view of Remark~\ref{rem:log:ded:cons},
the equivalence
\(
  \reft{\PT \mid \varphi}
  +
  \reft{\PTbis \mid \psi}
  \eqtype
  \reft{\PT + \PTbis \mid \form{\inj_1}\varphi \lor \form{\inj_2}\psi}
\)
would be wrong,
thus preventing from Lemma~\ref{lem:reft} to hold.
\lipicsEnd
\end{remark}

Our main result (Theorem~\ref{thm:compl})
is that the \emph{positive} fragment of this type system
is (sound and) complete w.r.t.\ the usual Scott semantics:
given $\thesis M \colon \PT$
and a \emph{positive} $\varphi \in \Lang^+(\PT)$,
\[
\begin{array}{c !{\quad}c!{\quad} c}
  \thesis M \colon \reft{\PT \mid \varphi}
& \text{if, and only if,}
& \text{$\varphi$ holds on $\I M$ in the Scott semantics.}
\end{array}
\]

\noindent
In particular, the judgement $\thesis \cnt \comp \bft \colon S^+$
is derivable.

%

\section{Semantics}
\label{sec:sem}

We interpret pure types as Scott domains and $\lambda$-terms as Scott-continuous
functions (\S\ref{sec:sem:pure}).
Using the well-known fact that Scott domains are spectral spaces
(\S\ref{sec:spectral}),
this yields the semantics of our logic (\S\ref{sec:sem:log})
and of our refinement type system (\S\ref{sec:sem:reft}).
Our main result is the Positive Completeness Theorem~\ref{thm:compl}
in \S\ref{sec:sem:reft}.

We mostly use the terminology in~\cite[\S 1]{ac98book}.
A \emph{dcpo} is a poset with suprema of all directed subsets
(a subset $D$ of a poset is \emph{directed} if all (possibly empty) finite $F \sle D$
have an upper bound in $D$; directed sets are non-empty).
A \emph{cpo} is a dcpo with a least element (often denoted $\bot$).
A function between dcpos is \emph{Scott-continuous}
if it preserves the order (i.e.\ is monotone) as well as directed suprema.

\begin{definition}[Scott Domain]
\label{def:scott}
A \emph{Scott domain} is a bounded-complete algebraic cpo.
$\Scott$ is the category of Scott domains and Scott-continuous functions.
\end{definition}

Recall that a cpo $X$ is bounded-complete if
any two $x,y \in X$ have a sup (or \emph{least} upper bound)
$x \vee y \in X$ whenever they have an upper bound.

An element $x$ of a dcpo $X$ is \emph{finite}
if for all directed $D \sle X$ such that $x \leq \bigvee D$,
we have $x \leq d$ for some $d \in D$
(finite elements are called \emph{compact} in~\cite{ac98book}).
Note that $\bot$ is always finite,
and that if $d,d' \in X$ are finite,
then $d \vee d'$ is finite whenever it exists.
A dcpo $X$ is \emph{algebraic} if for each $x \in X$,
the set $\{ d\in X \mid \text{$d$ finite and $\leq x$} \}$
is directed and has sup $x$.

The category $\Scott$ is Cartesian-closed
(see e.g.\ \cite[Corollary 4.1.6]{aj95chapter}).

\subsection{Semantics of the Pure System}
\label{sec:sem:pure}

Typed terms $\Env \thesis M : \PT$ of the pure system (\S\ref{sec:pure})
are interpreted as morphisms $\I M \colon \I\Env \to \I\PT$ in $\Scott$,
where $\I{\Env} = \prod_{i=1}^n\I{\PTbis_i}$ when
$\Env = x_1:\PTbis_1,\dots,x_n:\PTbis_n$.%
\opt{full}{ See Appendix~\ref{sec:proof:sem:pure} for details.}%
\opt{short}{ See \cite[\S C]{rd26full} for details.}%

Product types $\PT \times \PTbis$ are interpreted using
the Cartesian product of $\Scott$, i.e.\ the Cartesian product
of sets equipped with component-wise order.
Arrow types $\PTbis \to \PT$ are interpreted using the
closed structure of $\Scott$,
given by equipping each homset $\Scott(X,Y)$ with the pointwise order.

The interpretation of $\Unit$ is $\{\bot, \top\}$
with $\I{\pair{}} \deq \top$.
We let $\I{\PT + \PTbis}$ be the \emph{weak} coproduct
$(\I\PT \amalg \I\PTbis)_\bot$,
i.e.,
the disjoint union of $\I\PT$ and $\I\PTbis$ augmented with a new least element $\bot$
($\Scott$ does not have finite coproducts~\cite{hp90tcs}).
We refer to~\cite{ac98book,aj95chapter,streicher06book}
for the interpretation of recursive types $\rec\TV.\PT$.
Term-level fixpoints $\fix x.M$ are interpreted
using the usual fixpoint combinators $\term Y \colon (X \to X) \to X$
taking $f \colon X \to X$ to $\term Y(f) \deq \bigvee_{n \in \NN} f^n(\bot)$.

\begin{example}
\label{ex:scott:stream-tree}
The domain $\I{\Stream\PTbis}$ of streams is $\I\PTbis^\NN$
equipped with the pointwise order,
while the domain $\I{\Tree\PTbis}$ of trees is $\I\PTbis^{\{0,1\}^*}$.
The finite elements are those $z \in \I\PTbis^I$ such that $z(p)$ is finite in
$\I\PTbis$ for all $p \in I$,
and $z(p) \neq \bot$ for at most finitely many $p \in I$.

Given $x \in \I{\Stream\PTbis}$
we have $\I\hd(x) = x(0)$,
while $\I\tl(x)$ is the stream taking $n \in \NN$ to $x(n+1) \in \I\PTbis$.
Moreover, $x = \I\Cons(\I\hd(x),\I\tl(x))$.
%
Similarly, if $y \in \I{\Tree\PTbis}$
then $\I\lbl(y) = y(\es)$ is the root label of $y$,
while $\I\lft(y)$ and $\I\rght(y)$ are the left- and right-subtrees
of $y$, respectively.

One can then check that $\I\bft \colon \I{\Tree \PTbis} \to \I{\Stream \PTbis}$
indeed computes a breadth-first tree traversal%
\opt{full}{ (see Appendix~\ref{sec:proof:sem:stream-tree:bft}).}%
\opt{short}{ (see~\cite[\S C.3.1]{rd26full}).}
In particular, $\I\cnt \comp \I\bft$
satisfies its expected specification
(cf Equation~\eqref{eq:form:spec} in Example~\ref{ex:form:fix}).
\lipicsEnd
\end{example}

\subsection{Scott Domains as Spectral Spaces}
\label{sec:spectral}

The semantics of formulae and refinement types involves some topology.

Let $(X,\leq)$ be a dcpo.
A set $U \sle X$ is \emph{Scott-open}, notation $U \in \Open(X)$,
if $U$ is \emph{saturated}
(if $x \in U$ and $x \leq y$ in $X$, then $y \in U$),
and if moreover $U$ is inaccessible by directed sups,
in the sense that if $\bigvee D \in U$
with $D \sle X$ directed, then $D \cap U \neq \emptyset$.
This equips $X$ with a topology, called the \emph{Scott topology}.
A function between dcpos is Scott-continuous
precisely when it is continuous for the Scott topology.
See e.g.~\cite[\S 1.2]{ac98book} or~\cite[\S 2.3]{aj95chapter}.
The following is well-known
(see e.g.~\cite[Theorem 7.2.29]{aj95chapter}
or~\cite[Theorem 7.47]{gg24book}).

\begin{restatable}{theorem}{ThmScottSpectral}
\label{thm:spectral}
The Scott topology of a Scott domain is spectral.
\end{restatable}

There are different equivalent definitions of spectral spaces,
see e.g.~\cite[1.1.5]{dst19book},
\cite[\S 6.1]{gg24book} or~\cite[9.5.1]{goubault13book}.
For now, a \emph{spectral} space is $T_0$, well-filtered,
and
the compact-opens are stable under finite intersections and form a basis
of the topology.%
\opt{full}{ See Appendix~\ref{sec:proof:spectral} for details.}%
\opt{short}{ See \cite[\S D]{rd26full} for details.}%

Scott topologies are always $T_0$ since $x \leq y$ in a dcpo $X$
if, and only if, $x \in U$ implies $y \in U$ for each Scott-open $U$.
%
Given a finite $d \in X$, it is easily seen that 
$\up d = \{ x \in X \mid d \leq x\}$ is
\emph{compact-open} (i.e.\ compact and Scott-open).
In fact,
if $X$ is algebraic then the Scott-opens are exactly the unions of $\up d$'s,
with $d$ finite in $X$.
Each cpo $X$ is trivially compact since $\up \bot = X$.
More importantly,
if $X$ is a Scott domain then $\up d \cap \up d'$
is compact for all finite $d,d' \in X$
(by bounded-completeness,
if $\up d \cap \up d'$ is non-empty, then $d \vee d'$
is defined, finite and such that $\up (d \vee d') = \up d \cap \up d'$).
It follows that the set $\K\Open(X)$ of compact-open subsets of $X$
is stable under finite intersections and forms
a basis of the Scott topology.

Proposition~\ref{prop:sem:wf} below states that Scott domains are \emph{well-filtered}.
This will yield the topological classification of 
polarised formulae in Proposition~\ref{prop:sem:log:degroot} (\S\ref{sec:sem:log}).
We let $\K\Sat(X)$ be the set of \emph{compact-saturated}
(i.e.\ compact and saturated) subsets of $X$.
Note that a set $\SP \sle X$ is saturated if, and only if,
$\SP = \bigcap \{U \in \Open(X) \mid \SP \sle U \}$.
Proposition~\ref{prop:sem:wf}
implies that $Q \in \K\Sat(X)$ if, and only if,
$Q = \bigcap\{K \in \K\Open(X) \mid Q \sle K \}$,
and thus that $\K\Sat(X)$ is stable under all intersections.
A subset $\mathcal{Q}$ of a poset $P$ is \emph{codirected}
if $\mathcal{Q}$ is directed in $P^\op$.

\begin{restatable}[Well-Filteredness]{proposition}{PropWellFilt}
\label{prop:sem:wf}
Each Scott domain $X$ is \emph{well-filtered}:
given a codirected
set $\mathcal{Q} \sle \K\Sat(X)$,
if $\bigcap \mathcal{Q} \sle U$ with $U$ Scott-open,
then $Q \sle U$ for some $Q \in \mathcal{Q}$.
\end{restatable}

\begin{corollary}
\label{cor:sem:ksatbigcap}
If $X$ is a Scott-domain, then $\K\Sat(X)$ is stable under all intersections.
\end{corollary}

\begin{remark}[De Groot Duality]
\label{rem:sem:degroot}
In a spectral space $X$, the compact-saturated sets
are exactly the intersections of compact-opens,
while the opens are exactly the unions of compact-opens.
This symmetry between $\K\Sat(X)$ and $\Open(X)$
is formalised by \emph{de Groot} duality.
A nice panorama from the more general point
of view of \emph{stably compact spaces}
is given in~\cite{lawson11mscs}.%
\footnote{In a stably compact space $X$, $\K\Open(X)$ may not be a basis.
A stably compact algebraic dcpo is spectral.}
\lipicsEnd
\end{remark}

\begin{colinpar}
TODO/NOTES
\begin{itemize}
\item
NOTE: We avoided the terminology \emph{coherent},
since it is too overloaded.
\end{itemize}
\end{colinpar}

\subsection{Semantics of Formulae}
\label{sec:sem:log}

\begin{figure}[t!]
\centering
\(
\begin{array}{r @{~}c@{~} l r @{~}c@{~} l r @{~}c@{~} l}

  \I{\varphi \land \psi}\val
& \deq
& \I\varphi\val \cap \I\psi\val

& \I{\FP}\val
& \deq
& \val(\FP)

& \I{(\finmu^t \FP)\varphi}\val
& \deq
& \I{(\FP.\varphi)^{\I{t}\val}(\False)}\val
\\

  \I{\varphi \lor \psi}\val
& \deq
& \I\varphi\val \cup \I\psi\val

& \I{\form{\pair{}}}\val
& \deq
& \{\top\}

& \I{(\finnu^t \FP)\varphi}\val
& \deq
& \I{(\FP.\varphi)^{\I{t}\val}(\True)}\val
\\

  \I\True\val
& \deq
& \I\PT

& \I{\form{\modgen}\varphi}\val
& \deq
& \I{\form{\modgen}}(\I\varphi\val)

& \I{(\exists \itvar)\varphi}\val
& \deq
& \bigcup_{n \in \NN} \I{\varphi}\val[n / \itvar]
\\

  \I\False\val
& \deq
& \emptyset

& \I{\psi \realto \varphi}\val
& \deq
& \I\psi\val \realto \I\varphi\val

& \I{(\forall \itvar)\varphi}\val
& \deq
& \bigcap_{n \in \NN} \I{\varphi}\val[n / \itvar]

\end{array}
\)
\caption{Interpretation $\I\varphi\val \in \Po(\I\PT)$
of $\varphi \in \Lang(\FPEnv;\PT)$,
where $\modgen \in \{\pi_1,\pi_2,\inj_1,\inj_2,\fold\}$.%
\label{fig:sem:form}}
\end{figure}

\begin{figure}[t!]
\centering
\(
\begin{array}{r c l c l}

  \SP \in \Po(\I{\PT_i})
& \longmapsto
& \I{\form{\pi_i}}(\SP)
& \deq
& \left\{
  x \in \I{\PT_1 \times \PT_2}
  \mid
  \I{\pi_i}(x) \in \SP
  \right\}
\\

  \SP \in \Po(\I{\PT_i})
& \longmapsto
& \I{\form{\inj_i}}(\SP)
& \deq
& \left\{
  \I{\inj_i}(x) \in \I{\PT_1 + \PT_2}
  \mid
  x \in \SP
  \right\}
\\

  \SP \in \Po(\I{\PT[\rec\TV.\PT/\TV]})
& \longmapsto
& \I{\form\fold}(\SP)
& \deq
& \left\{
  x \in \I{\rec\TV.\PT}
  \mid
  \I\unfold(x) \in \SP
  \right\}
\\

  \SP \in \Po(\I\PTbis)
  \,,\,
  \SPbis \in \Po(\I\PT)
& \longmapsto
& (\SP \realto \SPbis)
& \deq
& \left\{
  f \in \I{\PTbis \arrow \PT}
  \mid
  \forall x \in \SP,~ f(x) \in \SPbis
  \right\}

\end{array}
\)
\caption{Semantic modalities, where $i = 1,2$.%
\label{fig:sem:mod}}
\end{figure}

A \emph{valuation} of a fixpoint context $\FPEnv$ is a function $\val$ taking
fixpoint variables
$\FP$ with $(\FP \colon \PT) \in \FPEnv$ to sets $\val(\FP) \in \Po(\I\PT)$,
and iteration variables $\itvar$ to natural numbers $\val(\itvar) \in \NN$.
A valuation $\val$ thus interprets iteration terms $t$
as natural numbers $\I{t}\val \in \NN$.
Figure~\ref{fig:sem:form}
defines the interpretation $\I\varphi\val \in \Po(\I\PT)$
of a formula $\varphi \in \Lang(\FPEnv; \PT)$
under a valuation $\val$ of $\FPEnv$,
where $(\FP.\varphi)^n(\False)$
and
$(\FP.\varphi)^n(\True)$
are defined in Equation~\eqref{eq:form:syntfun}, \S\ref{sec:formulae},
and where the semantic modalities $\I{\form\modgen}$ and
$(\pl) \realto (\pl)$ are defined in Figure~\ref{fig:sem:mod}.%
\opt{full}{ See Appendix~\ref{sec:proof:sem:log} for details.}%
\opt{short}{ See \cite[\S E]{rd26full} for details.}%

In view of Figure~\ref{fig:sem:mod},
it is clear that $\I{\inj_1}(\SP_1) \cap \I{\inj_2}(\SP_2)$
is always empty.
This yields the soundness of the rule $\ax{I}$ in Figure~\ref{fig:ded:mod}
(cf Remark~\ref{rem:log:ded:cons}).
The semantic realizability implication $(\pl) \realto (\pl)$ 
is clearly contravariant in its first argument and covariant
in its second argument.
Besides, the polarity constraints in Equation~\eqref{eq:form:realto}
(\S\ref{sec:formulae}),
as well as the rule $\ax{WF}$ in Figure~\ref{fig:ded:realto},
come from the following crucial consequence of well-filteredness
(Proposition~\ref{prop:sem:wf}).

\begin{restatable}{lemma}{LemWfArrow}
\label{lem:sem:log:realto}
Given $V \in \Open(\I\PTbis)$
and a codirected $\mathcal{Q} \sle \K\Sat(\I\PT)$,
if $f \in \bigcap\mathcal{Q} \realto V$,
then there is some $Q \in \mathcal{Q}$
such that $f \in Q \realto V$ already.
Moreover, given $Q \in \K\Sat(\I\PT)$ we have
$Q \realto V \in \Open(\I{\PT \arrow \PTbis})$
and
$V \realto Q \in \K\Sat(\I{\PTbis \arrow \PT})$.
\end{restatable}

\begin{proposition}[restate = PropFormdeGroot, name = ]
\label{prop:sem:log:degroot}
If $\varphi \in \Lang^+(\PT)$
then
$\I\varphi\val \in \Open(\I\PT)$.
If $\varphi \in \Lang^-(\PT)$
then
$\I\varphi\val \in \K\Sat(\I\PT)$.
In particular,
if $\varphi \in \Lang^\pm(\PT)$
then
$\I\varphi\val \in \K\Open(\I\PT)$.
\end{proposition}

Our polarised logical languages $\Lang^+$ and $\Lang^-$
are motivated by the topological classification in
Proposition~\ref{prop:sem:log:degroot}.
In view of Figure~\ref{fig:form:pol},
the following Lemma~\ref{lem:sem:log:posneg} implies
that the unions and intersections interpreting
$(\exists \itvar)\varphi$ and $(\forall \itvar)\varphi$
are directed and codirected, respectively.
It also yields the soundness of the 
rules $\ax{\exists M}$ and $\ax{\forall M}$
in Figure~\ref{fig:ded:quant}.

\begin{restatable}{lemma}{LemPosNeg}
\label{lem:sem:log:posneg}
If $\itvar \Pos \varphi$
(resp.\ $\itvar \Neg \varphi$)
then the function
$n \mapsto \I\varphi\val[n/\itvar]$
is monotone (resp.\ antimonotone),
that is, $n \leq m$ (resp.\ $m \leq n$)
implies
$\I\varphi\val[n / \itvar] \sle \I\varphi\val[m / \itvar]$.
\end{restatable}

Formulae satisfy an important Scott-(co)continuity property
w.r.t.\ their fixpoint variables.
Given $\varphi \in \Lang(\FPEnv;\PT)$ and a valuation $\val$ of $\FPEnv$,
for each $(\FP \colon \PTbis) \in \FPEnv$
we have a function
\[
  \I{\FP.\varphi}\val \colon
  \Po(\I{\PTbis})
  \longto
  \Po(\I\PT)
  ,\quad
  \SP
  \longmapsto
  \I{\varphi}\val[\SP/\FP]
\]

\noindent
Note that in the case of
$\varphi \in \Lang(\FPEnv,\FP \colon \PT;\PT)$,
for all $\psi \in \Lang(\FPEnv;\PT)$ and all $n \in \NN$,
we have
$(\I{\FP.\varphi}\val)^n(\I\psi\val) = \I{(\FP.\varphi)^n(\psi)}\val$,
where
$(\FP.\varphi)^n(\psi)$ is the formula in Equation~\eqref{eq:form:syntfun}.

\begin{restatable}{proposition}{PropFormScott}
\label{prop:sem:log:scott}
Let $\val$ be a valuation of $\FPEnv$, and let $(\FP \colon \PTbis) \in \FPEnv$.
\begin{enumerate}[(1)]
\item
\label{item:sem:log:scott:mon}
If $\varphi \in \Lang^\sign(\FPEnv;\PT)$,
then $\I{\FP.\varphi}\val$ is monotone.

\item
\label{item:sem:log:scott:pos}
If $\varphi \in \Lang^+(\FPEnv;\PT)$,
then $\I{\FP.\varphi}\val$ is Scott-continuous
(i.e.\ preserves directed unions).

\item
\label{item:sem:log:scott:neg}
If $\varphi \in \Lang^-(\FPEnv;\PT)$,
then $\I{\FP.\varphi}\val$ is Scott-cocontinuous
(i.e.\ preserves codirected intersections).
\end{enumerate}
\end{restatable}

\begin{example}
\label{ex:sem:log:fixpoints}
Let $\varphi \in \Lang^+(\FPEnv,\FP\colon\PT;\PT)$
and $\psi \in \Lang^-(\FPEnv,\FP\colon\PT;\PT)$.
It follows from Proposition~\ref{prop:sem:log:scott}
that the least fixpoint of $\I{\FP.\varphi}\val$
and the greatest fixpoint of $\I{\FP.\psi}\val$
are
\[
\begin{array}{l !{\quad\text{and}\quad} l}
  \I{(\exists \itvar)(\finmu^\itvar \FP)\varphi}\val
  =
  \bigcup_{n \in \NN} (\I{\FP.\varphi}\val)^n(\False)

& \I{(\forall \itvarbis)(\finnu^\itvarbis \FP)\psi}\val
  =
  \bigcap_{n \in \NN} (\I{\FP.\psi}\val)^n(\True)
\end{array}
\]

On streams, we get that
$\Diam\varphi$ defines an \emph{open} set,
while
$\Box\psi$ defines a \emph{compact-saturated} set 
(with $\varphi$ positive and $\psi$ negative).
For instance, $\I{\Box\form\hd\form\true} = \{\I\true^\omega\}$
is compact-saturated in
$\I{\Stream\Bool} = \I\Bool^\omega$.
This contrasts with the case of $\omega$-words,
for which $\{\true^\omega\}$ is a \emph{closed} subset
of $\{\true, \false\}^\omega$ for the product topology.%
\footnote{See~\cite{rs24jfla} for a comparison of the Scott topology on streams
with the product topology on $\omega$-words.}
Besides, $\omega$-words over a finite alphabet form a Stone (Boolean) space,
i.e.\ a spectral space in which the compact-saturated sets coincide with 
the closed ones.
But in the case of the Scott domain of streams,
proper compact-saturated sets (such as $\{\I\true^\omega\}$)
are never closed.
We think this makes the case of de Groot duality (Remark~\ref{rem:sem:degroot})
for temporal properties on streams.

In the context of Example~\ref{ex:form:fix},
the significance of well-filteredness (Proposition~\ref{prop:sem:wf})
in Lemma~\ref{lem:sem:log:realto}
is that a Scott-continuous function
$f \colon \I{\Tree\Bool \arrow \Stream\Nat}$,
say $\I{\cnt \comp \bft}$,
satisfies the formula in Equation~\eqref{eq:form:spec}
\emph{if, and only if,} there is some $m \in \NN$
such that for each totally defined $t \in \I{\Tree\Bool}$,
the length-$m$ prefix of $f(t)$ contains a totally defined
number $\geq n$
provided that in the depth-$m$ truncation of $t$,
the children of the root have a maximal path consisting of $\true$'s
and one consisting of $\false$'s.
%
%
Hence, only a depth-$m$ exploration of $f(t)$ is needed to establish the property,
where $m \in \NN$ depends on $f$ but \emph{not} on $t$.
\lipicsEnd
\end{example}

\colin{BEWARE: One may define $P \sle X$ to be a safety property
if there is a set of finite elements $P_{\text{bad}} \sle X$
such that $P$ is the set of all $x \in X$ with no finite $d \leq x$ in $P_{\text{bad}}$.
But in this case all safety properties are Scott-closed.}

Proposition~\ref{prop:sem:log:scott} yields 
the rules $\ax{\fingen/\exists}$ and $\ax{\forall/\fingen}$
in Figure~\ref{fig:ded:quant}.
The following Proposition~\ref{prop:sem:char:fin}
is crucial for the rules
$\ax{{\realto}/\lor}$ and $\ax{CC}$ in Figure~\ref{fig:ded:realto},
as well as for completeness.

\begin{restatable}{proposition}{PropTopCharFin}
\label{prop:sem:char:fin}
Given $\delta \in \Lang^\land(\PT)$,
if $\I\delta \neq \emptyset$ then
$\I\delta = \up d$ for some finite $d \in \I\PT$.
Conversely, if $d \in \I\PT$ is finite, then $\up d = \I\delta$ for some
$\delta \in \Lang^{\land}(\PT)$.
\end{restatable}

\begin{theorem}[restate = ThmLogSound, name = Soundness of Deduction]
\label{thm:sem:log:sound}
\hfill
\begin{enumerate}[(1)]
\item
\label{item:sem:log:sound:cons}
For all $\delta \in \Lang^\land(\PT)$,
if $\C(\delta)$ is derivable, then $\I\delta \neq \emptyset$.

\item
\label{item:sem:log:sound:ded}
For all $\varphi, \psi \in \Lang^\sign(\FPEnv;\PT)$,
if $\psi \thesis \varphi$ is derivable,
then $\I\psi\val \sle \I\varphi\val$.
\end{enumerate}
\end{theorem}

The following is the key logical fact toward completeness.
It also entails the completeness of the neutral fragment
(which we actually do not need for our main Theorem~\ref{thm:compl}).

\begin{restatable}[Completeness of $\Lang^\land$]{theorem}{ThmLogCompl}
\label{thm:sem:log:compl}
Let $\varphi, \psi \in \Lang^\land(\PT)$.
\begin{enumerate}[(1)]
\item
\label{item:sem:log:compl:consist}
If $\I\varphi \neq \emptyset$, then $\C(\varphi)$ is derivable.

\item
\label{item:sem:log:compl:ded}
If $\I\psi \sle \I\varphi$,
then $\psi \thesis \varphi$ is derivable.
\end{enumerate}
\end{restatable}

\subsection{Semantics of the Refinement Type System}
\label{sec:sem:reft}
The interpretation $\I\RT \sle \I{\UPT\RT}$
of a refinement type $\RT$ is defined as
$\I{\reft{\PT \mid \varphi}} \deq \I\varphi$
and
$\I{\RTbis \arrow \RT} \deq \I\RTbis \realto \I\RT$
and
$\I{\RT \times \RTbis} \deq \I{\form{\pi_1}}(\I\RT) \cap \I{\form{\pi_2}}(\I\RTbis)$.

\begin{definition}[Sound Judgement]
\label{def:sound:typing}
A judgement $\Env \thesis M \colon \RT$ with
$\Env = x_1\colon\RTbis_1, \dots, x_n\colon\RTbis_n$
is \emph{sound} if $\UPT\Env \thesis M \colon \UPT\RT$ is derivable and
if moreover $\I M(u_1,\dots,u_n) \in \I\RT$
whenever $u_i \in \I{\RTbis_i}$ for all $i = 1,\dots,n$.
\end{definition}

\begin{theorem}[Soundness of Typing]
\label{thm:sem:sound:reft}
If $\Env \thesis M \colon \RT$ is derivable then it is sound.
\end{theorem}

It follows from Examples~\ref{ex:scott:stream-tree} and~\ref{ex:sem:log:fixpoints}
that $\thesis \cnt \comp \bft \colon \RT^+$
is sound, where $\RT^+$ is either type in Example~\ref{ex:reft:spec}.
This judgement is also derivable, thanks to our main result:

\begin{restatable}[Positive Completeness]{theorem}{ThmCompl}
\label{thm:compl}
If $\Env^- \thesis M\colon \RT^+$ is sound,
then it is derivable.
\end{restatable}

Write $\RT^\land$ for a refinement type which only contains formulae in $\Lang^\land$,
and similarly for $\Env^\land$.
We obtain Theorem~\ref{thm:compl} by reduction to 
judgements of the form $\Env^\land \thesis M \colon \RT^\land$,
for which completeness essentially amounts to~\cite{abramsky91apal}
(through formally our system differs).

\begin{restatable}{theorem}{ThmComplLand}
\label{thm:compl:land}
If $\Env^\land \thesis M \colon \RT^\land$ is sound,
then it is derivable.
\end{restatable}



The reduction of Theorem~\ref{thm:compl} to Theorem~\ref{thm:compl:land}
is an adaptation of~\cite{bk03ic,rk25wollic}.
Consider first the case of a sound $\thesis M \colon \RT$,
where $\RT$ is positive.
By Lemma~\ref{lem:reft} and Proposition~\ref{prop:ded:prenex},
let $\varphi \in \Lang^\pm(\UPT\RT)$
such that $\RT \eqtype \reft{\UPT\RT \mid (\exists \itvar)\varphi}$.
Since the judgement $\thesis M \colon \RT$
is sound, there is some $m \in \NN$
such that $\thesis M \colon \reft{\UPT\RT \mid \varphi[m/\itvar]}$
is also sound.
By Lemmas~\ref{lem:ded:omega} and~\ref{lem:ded:omegalorland},
the closed formula $\varphi[m/\itvar] \in \Lang^\pm(\UPT\RT)$
is equivalent to a disjunction of formulae $\gamma \in \Lang^\land(\UPT\RT)$.
But there must be at least one of these $\gamma$'s such that
the judgement $\thesis M \colon \reft{\UPT\RT \mid \gamma}$ is sound,
and we can apply Theorem~\ref{thm:compl:land}.
%
%
Since
\(
  \reft{\UPT\RT \mid \gamma}
  \subtype
  \reft{\UPT\RT \mid \varphi[m / \itvar]}
  \subtype
  \reft{\UPT\RT \mid (\exists \itvar)\varphi}
\),
we obtain a derivation of 
$\thesis M \colon \reft{\UPT\RT \mid (\exists \itvar)\varphi}$
from a derivation of $\thesis M \colon \reft{\UPT\RT \mid \gamma}$.

Consider now the case of $x \colon \RTbis \thesis M \colon \RT$,
where $\RTbis$ is negative.
Let $\psi \in \Lang^\pm(\UPT\RTbis)$
such that $\RTbis \eqtype \reft{\UPT\RTbis \mid (\forall \itvarbis)\psi}$.
Since $\I\RT$ is Scott-open and $\I M$ is Scott-continuous,
well-filteredness (Proposition~\ref{prop:sem:wf})
yields some $m \in \NN$ such that
$\I{\psi[m / \itvarbis]} \sle \I M^{-1}(\I\RT)$,
i.e.\ such that
$x \colon \reft{\UPT\RTbis \mid \psi[m / \itvarbis]} \thesis M \colon \RT$
is sound.
%
We can similarly write $\psi[m / \itvarbis]$ as a disjunction
of $\delta \in \Lang^\land(\UPT\RTbis)$,
and each disjunct $\delta$ leads to a sound judgement
$x \colon \reft{\UPT\RTbis \mid \delta} \thesis M \colon \RT$.
If $\delta \not\thesis \False$, then Proposition~\ref{prop:sem:char:fin}
gives a finite $d \in \I{\UPT\RTbis}$ such that $\I\delta = \up d$.
Since $\I\RT$ is saturated, we have
$\up d \sle \I M^{-1}(\I\RT)$ if, and only if, $\I M(d) \in \I\RT$.
Unfolding $\RT$ similarly as above,
the condition $\I M(d) \in \I\RT$ yields some $\gamma \in \Lang^\land(\UPT\RT)$
with $\reft{\UPT\RT \mid \gamma} \subtype \RT$
and such that
$x \colon \reft{\UPT\RTbis \mid \delta} \thesis M \colon \reft{\UPT\RT \mid \gamma}$
is sound.%
\opt{full}{ See Appendix~\ref{sec:proof:sem:reft} for details.}%
\opt{short}{ See \cite[\S F]{rd26full} for details.}%


\begin{remark}[Undecidability]
\label{rem:sem:undec}
Positive judgements are semi-decidable but not decidable:
by Turing-completeness,
there is a $K \colon \Nat \arrow \Nat$
such that $K \term n$ evaluates the $n$th partial recursive function on input $n$.
Hence the total function which given $n \in \NN$
decides the derivability (or soundness)
of $\thesis K \term n \colon \reft{\Nat \mid \form{\tot}}$,
is not recursive
(\cite[Theorem II.2.3]{odifreddi99book}).
\lipicsEnd
\end{remark}

\section{Conclusion}
\label{sec:conc}

We presented a finitary refinement type system which is
complete for a set of Scott-open specifications 
strictly extending~\cite{abramsky87lics,abramsky91apal}.
We moreover noted that these specifications are not decidable 
for the typed $\lambda$-calculus under consideration (Remark~\ref{rem:sem:undec} above).

An important direction of future work is to characterise
the expressiveness of our logic.
Example~\ref{ex:sem:log:fixpoints} means that our positive
(resp. negative) fragment corresponds to a 
modal $\mu$-calculus on (finitary) polynomial types,
restricted to negation-free formulae with least (resp.\ greatest) fixpoints.
However, we cannot yet be as precise for higher-order types
(with formulae involving the realizability implication $\realto$). 

An other direction for future work is to extend this system
to liveness properties (as in Remark~\ref{rem:form:fix}),
while still targeting a complete finitary type system
(in contrast with~\cite{jr21esop,rk25wollic}).

Our system has sum types (so as to have a recursive type of natural numbers),
but the category $\Scott$ does not have finite coproducts.
Working with Call-By-Push-Value (CBPV)~\cite{levy03book,levy22siglog},
for the usual adjunction between dcpos and cpos with strict functions,
may lead to a cleaner situation.
In the long run, it would be nice if this basis could extend to
enriched models of CBPV,
so as to handle further computational effects.
Print and global store are particularly relevant,
as an important trend in proving temporal properties
considers programs generating streams of events.
Major works in this line include
\cite{ssv08jfp,hc14lics,hl17lics,nukt18lics,kt14lics,ust17popl,nukt18lics,%
su23popl}.
In contrast with ours, these approaches are based on trace semantics
of syntactic expressions rather than denotational domains.

\colin{See e.g.~\cite[Theorem 4.1 (and Figure 6)]{nukt18lics}
or~\cite[Theorem 1 (and Definition 20 from the full version)]{su23popl}.}

In a different direction, we think the approach of this paper
could extend to linear types~\cite{hjk00mscs,nw03concur,winskel04llcs},
in particular for the $\lambda$-calculus $\HOPLA$,
and possibly relying on the categorical study of~\cite{bf06book}.


\opt{full}{\newpage}
\bibliography{bibliographie}

\appendix
\opt{full}{\newpage}
\opt{full}{
\section{Proofs of \S\ref{sec:log} (\nameref{sec:log})}
\label{sec:proof:log}


\begin{figure}[t!]
\begin{subfigure}{\textwidth}
\centering
\(
\begin{array}{c}

  \psi \land (\varphi \lor \theta)
  \,\thesisiff\,
  (\psi \land \varphi) \lor (\psi \land \theta)

\qquad

  \psi \lor (\varphi \land \theta)
  \,\thesisiff\,
  (\psi \lor \varphi) \land (\psi \lor \theta)

\end{array}
\)
\caption{Distributive laws.%
\label{fig:proof:log:der:distr}}
\end{subfigure}

\begin{subfigure}{\textwidth}
\centering
\(
\begin{array}{c}

  \form{\modgen}(\varphi \lor \psi)
  \,\thesisiff\,
  \form{\modgen}\varphi \lor \form{\modgen}\psi

\qquad

  \form{\modgen}(\varphi \land \psi)
  \,\thesisiff\,
  \form{\modgen}\varphi \land \form{\modgen}\psi

\end{array}
\)
\caption{Modalities $\modgen \in \{\pi_1, \pi_2, \inj_1, \inj_2, \fold\}$.%
\label{fig:proof:log:der:mod}}
\end{subfigure}

\begin{subfigure}{\textwidth}
\centering
\(
\begin{array}{c}

  (\psi_1 \lor \psi_2) \realto \varphi
  \,\thesisiff\,
  (\psi_1 \realto \varphi) \land (\psi_2 \realto \varphi)

\\\\

  \psi \realto (\varphi_1 \land \varphi_2)
  \,\thesisiff\,
  (\psi \realto \varphi_1) \land (\psi \realto \varphi_2)

\\\\

  \delta \realto (\varphi_1 \lor \varphi_2)
  \,\thesisiff\,
  (\delta \realto \varphi_1) \lor (\delta \realto \varphi_2)

\end{array}
\)
\caption{The realizability implication $\realto$,
where $\delta \in \Lang^\land$.%
\label{fig:proof:log:der:realto}}
\end{subfigure}

\begin{subfigure}{\textwidth}
\centering
\(
\begin{array}{c}

  (\finmu^0 \FP)\varphi
  \,\thesisiff\,
  \False

\qquad

  (\finmu^{t+1} \FP)\varphi
  \,\thesisiff\,
  \varphi[(\finmu^t \FP)\varphi / \FP]

\\\\

  (\finnu^0 \FP)\varphi
  \,\thesisiff\,
  \True

\qquad

  (\finnu^{t+1} \FP)\varphi
  \,\thesisiff\,
  \varphi[(\finnu^t \FP)\varphi / \FP]

\end{array}
\)
\caption{Bounded iteration.%
\label{fig:proof:log:der:iter}}
\end{subfigure}

\caption{Some derivable sequents.%
\label{fig:proof:log:der}}
\end{figure}

Figure~\ref{fig:proof:log:der} gathers some derivable sequents.

\begin{lemma}
\label{lem:proof:log:der}
The sequents in Figure~\ref{fig:proof:log:der} are all derivable.
\end{lemma}

\begin{proof}
We discuss each subfigure separately.
\begin{description}
\item[Case of Figure~\ref{fig:proof:log:der:distr}.]
The sequent
\(
  \psi \land (\varphi \lor \theta)
  \,\thesis\,
  (\psi \land \varphi) \lor (\psi \land \theta)
\)
is given by the rule $\ax{D}$ in Figure~\ref{fig:ded:prop}.
We derive its converse with
\[
\text{
\AXC{}
\UIC{$\psi \,\thesis\, \psi $}
\UIC{$\psi \land \varphi \,\thesis\, \psi $}
\AXC{}
\UIC{$\psi \,\thesis\, \psi $}
\UIC{$\psi \land \theta \,\thesis\, \psi $}
\BIC{$(\psi \land \varphi) \lor (\psi \land \theta) \,\thesis\, \psi $}
\AXC{}
\UIC{$\varphi \,\thesis\, \varphi$}
\UIC{$\psi \land \varphi \,\thesis\, \varphi$}
\UIC{$\psi \land \varphi \,\thesis\, \varphi \lor \theta$}
\AXC{}
\UIC{$\theta \,\thesis\, \theta$}
\UIC{$\psi \land \theta \,\thesis\, \theta$}
\UIC{$\psi \land \theta \,\thesis\, \varphi \lor \theta$}
\BIC{$(\psi \land \varphi) \lor (\psi \land \theta) \,\thesis\, \varphi \lor \theta$}
\BIC{$(\psi \land \varphi) \lor (\psi \land \theta)
  \,\thesis\,
  \psi \land (\varphi \lor \theta)$}
\DisplayProof}
\]

We similarly obtain
\(
  \psi \lor (\varphi \land \theta)
  \,\thesis\,
  (\psi \lor \varphi) \land (\psi \lor \theta)
\).
For the converse, note that $\ax{D}$ yields
\[
  (\psi \lor \varphi) \land (\psi \lor \theta)
  \,\thesis\,
  \big ( (\psi \lor \varphi) \land \psi \big)
  \lor
  \big ( (\psi \lor \varphi) \land \theta \big)
\]

\noindent
and thus
\[
  (\psi \lor \varphi) \land (\psi \lor \theta)
  \,\thesis\,
  \psi
  \lor
  \big ( (\psi \lor \varphi) \land \theta \big)
\]

\noindent
Now, using $\ax{D}$ again yields
\[
  (\psi \lor \varphi) \land (\psi \lor \theta)
  \,\thesis\,
  \psi
  \lor
  (\psi \land \theta)
  \lor
  (\varphi \land \theta)
\]

\noindent
Since $\psi \lor (\psi \land \theta) \,\thesis\, \psi$,
we obtain
\[
  (\psi \lor \varphi) \land (\psi \lor \theta)
  \,\thesis\,
  \psi
  \lor
  (\varphi \land \theta)
\]

\item[Case of Figure~\ref{fig:proof:log:der:mod}.]
The sequent
\(
  \form{\modgen}(\varphi \lor \psi)
  \,\thesis\,
  \form{\modgen}\varphi \lor \form{\modgen}\psi
\)
is given by Figure~\ref{fig:ded:mod}.
We derive its converse with
\[
\text{
\AXC{}
\UIC{$\varphi_1 \,\thesis\, \varphi_1$}
\UIC{$\varphi_1 \,\thesis\, \varphi_1 \lor \varphi_2$}
\UIC{$\form{\modgen}\varphi_1 \,\thesis\, \form{\modgen}(\varphi_1 \lor \varphi_2)$}
\AXC{}
\UIC{$\varphi_2 \,\thesis\, \varphi_2$}
\UIC{$\varphi_2 \,\thesis\, \varphi_1 \lor \varphi_2$}
\UIC{$\form{\modgen}\varphi_2 \,\thesis\, \form{\modgen}(\varphi_1 \lor \varphi_2)$}
\BIC{$\form{\modgen}\varphi_1 \lor \form{\modgen}\varphi_2
  \,\thesis\,
  \form{\modgen}(\varphi_1 \lor \varphi_2)$}
\DisplayProof}
\]

The law for $\land$ is derived similarly.

\item[Case of Figure~\ref{fig:proof:log:der:realto}.]
The sequent
\(
  {(\psi_1 \realto \varphi) \land (\psi_1 \realto \varphi)
  \,\thesis\,
  (\psi_1 \lor \psi_2) \realto \varphi}
\)
is given by the rule $\ax{\lor/{\realto}}$ 
in Figure~\ref{fig:ded:realto}.
We derive its converse with
\[
\text{
\AXC{}
\UIC{$\psi_1 \,\thesis\, \psi_1$}
\UIC{$\psi_1 \,\thesis\, \psi_1 \lor \psi_2$}
\AXC{}
\UIC{$\varphi \,\thesis\, \varphi$}
\BIC{$(\psi_1 \lor \psi_2) \realto \varphi \,\thesis\, \psi_1 \realto \varphi$}
\AXC{}
\UIC{$\psi_2 \,\thesis\, \psi_2$}
\UIC{$\psi_2 \,\thesis\, \psi_1 \lor \psi_2$}
\AXC{}
\UIC{$\varphi \,\thesis\, \varphi$}
\BIC{$(\psi_1 \lor \psi_2) \realto \varphi \,\thesis\, \psi_2 \realto \varphi$}
\BIC{$(\psi_1 \lor \psi_2) \realto \varphi
  \,\thesis\,
  (\psi_1 \realto \varphi) \land (\psi_2 \realto \varphi)$}
\DisplayProof}
\]

The law
\(
  \psi \realto (\varphi_1 \land \varphi_2)
  \,\thesisiff\,
  (\psi \realto \varphi_1) \land (\delta \land \varphi_2)
\)
is derived similarly.

The sequent
\(
  \delta \realto (\varphi_1 \lor \varphi_2)
  \,\thesis\,
  (\delta \realto \varphi_1) \lor (\delta \realto \varphi_2)
\)
with $\delta \in \Lang^\land$ is given by the rule
$\ax{{\realto}/\lor}$
in Figure~\ref{fig:ded:realto}.
We derive its converse with
\[
\text{
\AXC{}
\UIC{$\delta \,\thesis\, \delta$}
\AXC{}
\UIC{$\varphi_1 \,\thesis\, \varphi_1$}
\UIC{$\varphi_1 \,\thesis\, \varphi_1 \lor \varphi_2$}
\BIC{$\delta \realto \varphi_1 \,\thesis\, \delta \realto (\varphi_1 \lor \varphi_2)$}
\AXC{}
\UIC{$\delta \,\thesis\, \delta$}
\AXC{}
\UIC{$\varphi_2 \,\thesis\, \varphi_2$}
\UIC{$\varphi_2 \,\thesis\, \varphi_1 \lor \varphi_2$}
\BIC{$\delta \realto \varphi_2 \,\thesis\, \delta \realto (\varphi_1 \lor \varphi_2)$}
\BIC{$(\delta \realto \varphi_1) \lor (\delta \realto \varphi_2)
  \,\thesis\,
  \delta \realto (\varphi_1 \lor \varphi_2)$}
\DisplayProof}
\]

\item[Case of Figure~\ref{fig:proof:log:der:iter}.]
We derive 
$(\finmu^0 \FP)\varphi \thesisiff \False$
using the rules
\[
\begin{array}{c}

\dfrac{}
  {(\finmu^0 \FP)\varphi \,\thesis\, \False}

\qquad

\dfrac{}
  {\False \,\thesis\, \psi}
\end{array}
\]

We moreover derive
\(
  (\finmu^{t+1} \FP)\varphi
  \,\thesisiff\,
  \varphi[(\finmu^t \FP)\varphi / \FP]
\)
from the rules
\[
\begin{array}{c}

\dfrac{(\finmu^t \FP)\varphi \,\thesis\, \psi}
  {(\finmu^{t+1} \FP)\varphi \,\thesis\, \varphi[\psi / \FP]}

\qquad

\dfrac{}
  {\varphi[(\finmu^t \FP)\varphi / \FP] \,\thesis\, (\finmu^{t+1} \FP)\varphi}

\end{array}
\]

The laws for $(\finnu^t \FP)$ are obtained similarly.
\qedhere
\end{description}
\end{proof}

We now turn to Lemma~\ref{lem:ded:omegalorland}.
It uses the following fact.

\begin{lemma}
\label{lem:proof:ded:omegalorland:prelim}
\begin{enumerate}[(1)]
\item
\label{item:proof:ded:omegalorland:mod}
For each $\varphi \in \Lang^{\lor\land}$,
there is some $\psi \in \Lang^{\lor\land}$
such that $\form{\modgen}\varphi \thesisiff \psi$.

\item
\label{item:proof:ded:omegalorland:realto:base}
For each $\delta \in \Lang^{\land}$
and each $\varphi \in \Lang^{\lor\land}$,
there is some $\psi \in \Lang^{\lor\land}$
such that $(\delta \realto \varphi) \thesisiff \psi$

\item
\label{item:proof:ded:omegalorland:realto}
For each $\varphi_1, \varphi_2 \in \Lang^{\lor\land}$,
there is some $\psi \in \Lang^{\lor\land}$
such that $(\varphi_1 \realto \varphi_2) \thesisiff \psi$
\end{enumerate}
\end{lemma}

\begin{proof}
\begin{enumerate}[(1)]
\item 
By induction on $\varphi \in \Lang^{\lor\land}$,
using Figure~\ref{fig:proof:log:der:mod}.

\item
By induction on $\varphi \in \Lang^{\lor\land}$,
using the last law in Figure~\ref{fig:proof:log:der:realto}.

\item 
By induction on $\varphi_1 \in \Lang^{\lor\land}$,
using item~(\ref{item:proof:ded:omegalorland:realto:base})
in the base case of $\varphi_1 \in \Lang^\land$,
and using the first law in Figure~\ref{fig:proof:log:der:realto}
together with distributive laws
(Figure~\ref{fig:proof:log:der:distr})
in the induction step ($\varphi_1$ of the form $\varphi \lor \varphi'$).
\qedhere
\end{enumerate}
\end{proof}

We can now prove Lemma~\ref{lem:ded:omegalorland}.

\LemOmegaLorLand*

\begin{proof}
We reason by induction on $\varphi \in \Lang^\omega(\PT)$.
\begin{itemize}
\item
The cases of $\form{\pair{}}$, $\True$ and $\False$ are trivial.

\item
The case of $\varphi \lor \psi$ follows directly from the induction hypothesis.
As for $\varphi \land \psi$, we use the induction hypothesis
and conclude with distributive laws (Figure~\ref{fig:proof:log:der:distr}).

\item
The cases of $\form{\modgen}\varphi$ follow from
the induction hypothesis and
Lemma~\ref{lem:proof:ded:omegalorland:prelim}(\ref{item:proof:ded:omegalorland:mod}).

\item
The case of $\psi \realto \varphi$ follows from
the induction hypothesis and
Lemma~\ref{lem:proof:ded:omegalorland:prelim}(\ref{item:proof:ded:omegalorland:realto}).
\qedhere
\end{itemize}
\end{proof}

As for Lemma~\ref{lem:ded:omega},
we use the following.

\begin{lemma}
\label{lem:proof:ded:neutralomega:prelim}
Let the closed iteration term $t$ be the $n$th successor of $0$.
Then:
\begin{enumerate}[(1)]
\item
\label{item:proof:ded:neutralomega:mu}
\(
  (\finmu^t \FP)\varphi
  \thesisiff
  (\FP. \varphi)^n(\False)
\)

\item
\label{item:proof:ded:neutralomega:nu}
\(
  (\finnu^t \FP)\varphi
  \thesisiff
  (\FP. \varphi)^n(\True)
\)
\end{enumerate}
\end{lemma}

\begin{proof}
\hfill
\begin{enumerate}[(1)]
\item 
By induction on $n \in \NN$,
using Figure~\ref{fig:proof:log:der:iter}
and Equation~\eqref{eq:form:syntfun} in~\S\ref{sec:formulae}.

\item
Similar to item~(\ref{item:proof:ded:neutralomega:mu}).
\qedhere
\end{enumerate}
\end{proof}

We now prove Lemma~\ref{lem:ded:omega}.

\LemNeutralOmega*

\begin{proof}
We prove a more general statement.
Let $\varphi \in \Lang^\pm(\FPEnv;\PT)$,
without iteration variables,
but with $\Env = \FP_1\colon\PT_1,\dots,\FP_n\colon\PT_n$.
Then for all $\theta_1,\dots,\theta_n$
with $\theta_i \in \Lang^\omega(\PT_i)$,
there is some $\psi \in \Lang^\omega(\PT)$
such that
\[
\begin{array}{l l l}
  \varphi[\theta_1 / \FP_1, \dots, \theta_n / \FP_n] 
& \thesisiff
& \psi
\end{array}
\]

We reason by induction on $\varphi \in \Lang^\pm(\FPEnv;\PT)$.
\begin{itemize}
\item In the case of $\FP_j \in \Lang^\pm(\FPEnv;\PT)$,
we take $\psi$ to be $\theta_j$.

\item In the cases of $\form{\pair{}}$, $\True$ and $\False$,
we can take $\psi$ to be $\varphi$.

\item
The cases of $\varphi_1 \lor \varphi_2$, $\varphi_1 \land \varphi_2$,
$\form{\modgen}\varphi$ and $\varphi_1 \realto \varphi_2$
all follow from the induction hypothesis.

\item
In the case of $(\finmu^t \FP)\varphi$,
let $\theta_1,\dots,\theta_n$ be as in the statement.
Note that $\varphi \in \Lang^{\pm}(\FPEnv, \FP\colon\PT; \PT)$.

We show that for each $m \in \NN$,
there is some $\psi_m \in \Lang^\omega(\PT)$ such that
\[
\begin{array}{l l l}
  (\FP.\varphi[\theta_1 / \FP_1, \dots, \theta_n / \FP_n])^m(\False)
& \thesisiff
& \psi_m
\end{array}
\]

We reason by induction on $m \in \NN$.
In the base case $m = 0$, we can take $\psi_0$ to be $\False$.
As for the induction step, we can take
\[
\begin{array}{l l l}
  \psi_{m+1}
& \thesisiff
& \varphi[\theta_1 / \FP_1, \dots, \theta_n / \FP_n, \psi_m / \FP]
\end{array}
\]

\noindent
obtained by applying the induction hypothesis on $\varphi$
with $\theta_1,\dots,\theta_n$ and $\psi_m$.

We can then conclude using
Lemma~\ref{lem:proof:ded:neutralomega:prelim}(\ref{item:proof:ded:neutralomega:mu}),
taking $\psi$ to be $\psi_m$ with
$m \in \NN$ such that the \emph{closed} iteration term $t$
is the $m$th successor of $0$.

\item
The case of $(\finnu^t \FP)\varphi$
is similar, using item~(\ref{item:proof:ded:neutralomega:nu})
of Lemma~\ref{lem:proof:ded:neutralomega:prelim}
instead of item~(\ref{item:proof:ded:neutralomega:mu}).
\qedhere
\end{itemize}
\end{proof}

\begin{figure}
\begin{subfigure}{\textwidth}
\centering
\(
\begin{array}{c}

  (\exists\itvar)\varphi \lor \psi
  \,\thesisiff\,
  (\exists\itvar)(\varphi \lor \psi) 

\qquad

  (\exists\itvar)\varphi \land \psi
  \,\thesisiff\,
  (\exists\itvar)(\varphi \land \psi) 

\\\\

  (\forall\itvar)\varphi \land \psi
  \,\thesisiff\,
  (\forall\itvar)(\varphi \land \psi) 

\qquad

  (\forall\itvar)\varphi \lor \psi
  \,\thesisiff\,
  (\forall\itvar)(\varphi \lor \psi) 

\end{array}
\)
\caption{Commutation over propositional connectives,
where $\itvar \notin \FV(\psi)$.%
\label{fig:proof:log:der:quant:prop}}
\end{subfigure}

\begin{subfigure}{\textwidth}
\centering
\(
\begin{array}{c}

  \form{\modgen}(\exists \itvar)\varphi
  \,\thesisiff\,
  (\exists \itvar)\form{\modgen}\varphi

\qquad

  \form{\modgen}(\forall \itvar)\varphi
  \,\thesisiff\,
  (\forall \itvar)\form{\modgen}\varphi

\end{array}
\)
\caption{Commutation over modalities.%
\label{fig:proof:log:der:quant:mod}}
\end{subfigure}

\begin{subfigure}{\textwidth}
\centering
\(
\begin{array}{c}

  (\forall \itvarbis)\psi \realto \varphi
  \,\thesisiff\,
  (\exists \itvarbis)(\psi \realto \varphi)

\qquad

  \theta \realto (\exists \itvar)\varphi
  \,\thesisiff\,
  (\exists \itvar)(\theta \realto \varphi)

\\\\

  (\exists \itvarbis)\psi \realto \varphi
  \,\thesisiff\,
  (\forall \itvarbis)(\psi \realto \varphi)

\qquad

  \psi \realto (\forall \itvar)\varphi
  \,\thesisiff\,
  (\forall \itvar)(\psi \realto \varphi)

\end{array}
\)
\caption{Commutation over the realizability implication,
where $\itvarbis \notin \FV(\varphi)$,
$\itvar \notin \FV(\psi,\theta)$ and $\theta \in \Lang^\pm$.%
\label{fig:proof:log:der:quant:realto}}
\end{subfigure}

\begin{subfigure}{\textwidth}
\centering
\(
\begin{array}{c}

 (\fingen^t \FP)(\exists \itvar)\varphi 
 \,\thesisiff\,
 (\exists \itvar)(\fingen^t \FP)\varphi

\qquad

 (\fingen^t \FP)(\forall \itvar)\varphi 
 \,\thesisiff\,
 (\forall \itvar)(\fingen^t \FP)\varphi
\end{array}
\)
\caption{Commutation over bounded iteration,
where $\fingen \in \{\finmu, \finnu\}$
and $\itvar \notin \FV(t)$.%
\label{fig:proof:log:der:quant:iter}}
\end{subfigure}

\begin{subfigure}{\textwidth}
\centering
\(
\begin{array}{c}

 (\exists \itvar)\varphi[\itvar / \itvarbis]
 \,\thesisiff\,
 (\exists \itvar)(\exists \itvarbis)\varphi

\qquad

 (\forall \itvar)(\forall \itvarbis)\varphi 
 \,\thesisiff\,
 (\forall \itvar)\varphi[\itvar / \itvarbis]
\end{array}
\)
\caption{Merge laws.%
\label{fig:proof:log:der:quant:merge}}
\end{subfigure}

\caption{Quantifier laws.%
\label{fig:proof:log:der:quant}}
\end{figure}

We finally turn to the crucial Proposition~\ref{prop:ded:prenex}.
Figure~\ref{fig:proof:log:der:quant} gathers some derivable
quantifier laws.

\begin{lemma}
\label{lem:proof:log:der:quand}
The sequents in Figure~\ref{fig:proof:log:der:quant} are all derivable.
\end{lemma}

\begin{proof}
We discuss each subfigure separately.
\begin{description}
\item[Case of Figure~\ref{fig:proof:log:der:quant:prop}.]
We derive
\(
  (\exists\itvar)\varphi \lor \psi
  \,\thesisiff\,
  (\exists\itvar)(\varphi \lor \psi) 
\)
as
\[
\text{
\AXC{}
\UIC{$\varphi \,\thesis\, \varphi$}
\UIC{$\varphi \,\thesis\, \varphi \lor \psi$}
\UIC{$\varphi \,\thesis\, (\exists\itvar)(\varphi \lor \psi)$}
\UIC{$(\exists\itvar)\varphi \,\thesis\, (\exists\itvar)(\varphi \lor \psi)$}
\AXC{}
\UIC{$\psi \,\thesis\, \psi$}
\UIC{$\psi \,\thesis\, \varphi \lor \psi$}
\UIC{$\psi \,\thesis\, (\exists\itvar)(\varphi \lor \psi)$}
\BIC{$(\exists\itvar)\varphi \lor \psi
  \,\thesis\,
  (\exists\itvar)(\varphi \lor \psi)$}
\DisplayProof}
\]

\noindent
and
\[
\text{
\AXC{}
\UIC{$\varphi \,\thesis\, \varphi$}
\UIC{$\varphi \,\thesis\, (\exists\itvar)\varphi$}
\UIC{$\varphi \,\thesis\, (\exists\itvar)\varphi \lor \psi$}
\AXC{}
\UIC{$\psi \,\thesis\, \psi$}
\UIC{$\psi \,\thesis\, (\exists\itvar)\varphi \lor \psi$}
\BIC{$\varphi \lor \psi
  \,\thesis\,
  (\exists\itvar)\varphi \lor \psi$}
\UIC{$(\exists\itvar)(\varphi \lor \psi)
  \,\thesis\,
  (\exists\itvar)\varphi \lor \psi$}
\DisplayProof}
\]

The rule $\ax{\exists L}$ in Figure~\ref{fig:ded:quant}
yields
\(
  (\exists\itvar)\varphi \land \psi
  \,\thesis\,
  (\exists\itvar)(\varphi \land \psi) 
\)
and we derive
\(
  (\exists\itvar)(\varphi \land \psi) 
  \,\thesis\,
  (\exists\itvar)\varphi \land \psi
\)
as
\[
\text{
\AXC{}
\UIC{$\varphi \,\thesis\, \varphi$}
\UIC{$\varphi \land \psi \,\thesis\, \varphi$}
\UIC{$\varphi \land \psi \,\thesis\, (\exists\itvar)\varphi$}
\UIC{$(\exists\itvar)(\varphi \land \psi) \,\thesis\, (\exists\itvar)\varphi$}
\AXC{}
\UIC{$\psi \,\thesis\, \psi$}
\UIC{$\varphi \land \psi \,\thesis\, \psi$}
\UIC{$(\exists\itvar)(\varphi \land \psi) \,\thesis\, \psi$}
\BIC{$(\exists\itvar)(\varphi \land \psi) 
  \,\thesis\,
  (\exists\itvar)\varphi \land \psi$}
\DisplayProof}
\]

The laws for $(\forall \itvar)$ are obtained similarly.

\item[Case of Figure~\ref{fig:proof:log:der:quant:mod}.]

We have
\(
  \form{\modgen}(\exists \itvar)\varphi
  \,\thesis\,
  (\exists \itvar)\form{\modgen}\varphi
\)
from Figure~\ref{fig:ded:mod},
and we derive the converse as
\[
\text{
\AXC{}
\UIC{$\varphi \,\thesis\, \varphi$}
\UIC{$\varphi \,\thesis\, (\exists \itvar)\varphi$}
\UIC{$\form{\modgen}\varphi \,\thesis\, \form{\modgen}(\exists \itvar)\varphi$}
\UIC{$(\exists \itvar)\form{\modgen}\varphi
  \,\thesis\,
  \form{\modgen}(\exists \itvar)\varphi$}
\DisplayProof}
\]

The laws for $(\forall \itvar)$ are obtained similarly.

\item[Case of Figure~\ref{fig:proof:log:der:quant:realto}.]
We begin with
\(
  (\forall \itvarbis)\psi \realto \varphi
  \,\thesisiff\,
  (\exists \itvarbis)(\psi \realto \varphi)
\).
We have
\(
  (\forall \itvarbis)\psi \realto \varphi
  \,\thesis\,
  (\exists \itvarbis)(\psi \realto \varphi)
\)
by rule $\ax{WF}$ in Figure~\ref{fig:ded:realto},
and we derive its converse as
\[
\text{
\AXC{}
\UIC{$\psi \,\thesis\, \psi$}
\UIC{$(\forall \itvarbis)\psi \,\thesis\, \psi$}
\AXC{}
\UIC{$\varphi \,\thesis\, \varphi$}
\BIC{$\psi \realto \varphi
  \,\thesis\,
  (\forall \itvarbis)\psi \realto \varphi$}
\UIC{$(\exists \itvarbis)(\psi \realto \varphi)
  \,\thesis\,
  (\forall \itvarbis)\psi \realto \varphi$}
\DisplayProof}
\]

The other laws
are obtained similarly.

\item[Case of Figure~\ref{fig:proof:log:der:quant:iter}.]
We have
\(
  (\fingen^t \FP)(\exists \itvar)\varphi
  \,\thesis\,
  (\exists \itvar)(\fingen^t \FP)\varphi
\)
from Figure~\ref{fig:ded:quant}
and we derive its converse as
\[
\text{
\AXC{}
\UIC{$\varphi \,\thesis\, \varphi$}
\UIC{$\varphi \,\thesis\, (\exists \itvar)\varphi$}
\UIC{$(\fingen^t \FP)\varphi \,\thesis\, (\fingen^t \FP)(\exists \itvar)\varphi$}
\UIC{$(\exists \itvar)(\fingen^t \FP)\varphi
  \,\thesis\,
  (\fingen^t \FP)(\exists \itvar)\varphi$}
\DisplayProof}
\]

The laws for $(\forall \itvar)$ are handled similarly.

\item[Case of Figure~\ref{fig:proof:log:der:quant:merge}.]
We derive
\[
\begin{array}{c !{\qquad\qquad} c}

\text{
\AXC{}
\UIC{$\varphi[\itvar / \itvarbis]
  \,\thesis\,
  \varphi[\itvar / \itvarbis]$}
\UIC{$\varphi[\itvar / \itvarbis]
  \,\thesis\,
  (\exists \itvarbis)\varphi$}
\UIC{$\varphi[\itvar / \itvarbis]
  \,\thesis\,
  (\exists \itvar)(\exists \itvarbis)\varphi$}
\UIC{$(\exists \itvar)\varphi[\itvar / \itvarbis]
  \,\thesis\,
  (\exists \itvar)(\exists \itvarbis)\varphi$}
\DisplayProof}

&

\text{
\AXC{}
\UIC{$\varphi[\itvar / \itvarbis]
  \,\thesis\,
  \varphi[\itvar / \itvarbis]$}
\UIC{$(\forall \itvarbis)\varphi
  \,\thesis\,
  \varphi[\itvar / \itvarbis]$}
\UIC{$(\forall \itvar)(\forall \itvarbis)\varphi
  \,\thesis\,
  \varphi[\itvar / \itvarbis]$}
\UIC{$(\forall \itvar)(\forall \itvarbis)\varphi
  \,\thesis\,
  (\forall \itvar)\varphi[\itvar / \itvarbis]$}
\DisplayProof}

\end{array}
\]

\noindent
and we obtain the converses from the rules
$\ax{\exists M}$ and $\ax{\forall M}$
in Figure~\ref{fig:ded:quant}.
\qedhere
\end{description}
\end{proof}

We can finally prove Proposition~\ref{prop:ded:prenex}.

\PropFormPrenex*

\begin{proof}
We reason by induction on $\varphi \in \Lang^\sign(\FPEnv; \PT)$.
\begin{itemize}
\item
In the case of $\FP$ with $(\FP \colon \PTbis) \in \FPEnv$,
we take $\psi$ to be $\FP$.

\item
In the cases of $\form{\pair{}}$, $\True$ and $\False$,
we can take $\psi$ to be $\varphi$.

\item
The cases of $(\exists \itvar)\varphi$ and $(\forall \itvar)\varphi$ 
follow from the induction hypothesis
and Figure~\ref{fig:proof:log:der:quant:merge}.

\item
The cases of $\varphi_1 \lor \varphi_2$ and $\varphi_1 \land \varphi_2$
are dealt-with using the induction hypothesis and
Figures~\ref{fig:proof:log:der:quant:prop}
and~\ref{fig:proof:log:der:quant:merge}.

\item
The cases of $\form{\modgen}\varphi$ are dealt-with 
using the induction hypothesis and
Figure~\ref{fig:proof:log:der:quant:mod}.

\item
In the case of a \emph{negative} realizability implication
$\varphi_1 \realto \varphi_2$,
we conclude by induction hypothesis
and Figure~\ref{fig:proof:log:der:quant:realto}.

Consider now the case of a \emph{positive} realizability implication
$\varphi_1 \realto \varphi_2$.
Using the induction hypothesis
and the top left law in Figure~\ref{fig:proof:log:der:quant:realto},
we can reduce
to the case of
$(\exists \itvar)(\theta \realto \varphi_2)$
with $\theta \in \Lang^\pm$.
We can the conclude using the top right law
in Figure~\ref{fig:proof:log:der:quant:realto}.

\item
The cases of $(\finmu^t \FP)\varphi$ and $(\finnu^t \FP)\varphi$
are dealt-with using the induction hypothesis and
Figure~\ref{fig:proof:log:der:quant:iter}.
\qedhere
\end{itemize}
\end{proof}

}
\opt{full}{
\section{Proofs of \S\ref{sec:reft} (\nameref{sec:reft})}
\label{sec:proof:reft}

We prove Lemma~\ref{lem:reft}.

\LemReft*

\begin{proof}
The proof is by induction on $\RT^\sign$.
The base case of $\reft{\PT \mid \varphi}$
with $\varphi \in \Lang^\sign(\PT)$ is trivial.
In the base case of $\PT$, one can take $\varphi \deq \True \in \Lang^\sign(\PT)$.
In the cases of $\RT^\sign \times \RTbis^\sign$
and $\RTter^{-\sign} \arrow \RT^\sign$,
by induction hypotheses we get
$\varphi \in \Lang^\sign(\UPT\RT)$,
$\psi \in \Lang^\sign(\UPT\RTbis)$,
and $\theta \in \Lang^{-\sign}(\UPT\RTter)$
such that
$\RT^\sign \eqtype \reft{\UPT\RT \mid \varphi}$,
$\RTbis^\sign \eqtype \reft{\UPT\RTbis \mid \varphi}$
and
$\RTter^{-\sign} \eqtype \reft{\UPT\RTter \mid \varphi}$.
We then conclude with
\[
\begin{array}{r c l}
  \RT^\sign \times \RTbis^\sign
& \eqtype
& \reft{\UPT\RT \times \UPT\RTbis \mid \form{\pi_1}\varphi \land \form{\pi_2}\psi}
\\

  \RTter^{-\sign} \arrow \RT^\sign
& \eqtype
& \reft{\UPT\RTter \arrow \UPT\RT \mid \theta \realto \varphi}
\end{array}
\]
\end{proof}

}
\opt{full}{

\section{Proofs of \S\ref{sec:sem:pure} (\nameref{sec:sem:pure})}
\label{sec:proof:sem:pure}

\colin{REFERENCES:
\begin{itemize}

\item Steps for solving domain equations,
besides~\cite[\S 7.1]{ac98book},
on may look at~\cite{aj95chapter}, namely:
\begin{itemize}
\item Embedding projection pairs
are defined in~\cite[Definition 3.1.15]{aj95chapter}.

\item Bilimits (in the sense of limits-colimits coincidence
for e-p pairs) are discussed in~\cite[\S 3.3.2]{aj95chapter}.

\item The case of algebraic domains is handled in~\cite[\S 3.3.3]{aj95chapter}.

\item The case of bounded-complete domains is handled
in~\cite[Proposition 4.1.3]{aj95chapter}.

\item Equations with parameters handled in~\cite[\S 5.2.3]{aj95chapter}
\end{itemize}
\end{itemize}}

\subsection{Solutions of Recursive Domain Equations}
We review the usual solution of recursive domain equations.
We refer to~\cite{ac98book,aj95chapter,streicher06book}.

\subsubsection{Categories of Domains}
In the following, $\DCPO$ is the category of those
posets with all directed suprema, and with Scott-continuous
functions as morphisms.
$\CPO$ is the full subcategory of $\DCPO$
on posets with a least element.
Note that $\Scott$ is a full subcategory of $\CPO$.

Recall that $\DCPO$, $\CPO$ and $\Scott$ are Cartesian-closed categories.
Finite products are given by Cartesian products
of sets equipped with the component-wise order.
The closed structures are
given by equipping each homset with the pointwise order.
See \cite[Theorem~3.3.3, Theorem~3.3.5 and Corollary~4.1.6]{aj95chapter}.
Hence for each $n \in \NN$, the categories
$\Scott^n$, $\CPO^n$ and $\DCPO^n$
are (not full) subcategories of $\Scott$, $\CPO$ and $\DCPO$
respectively.

\begin{lemma}
\label{lem:proof:scott:enrich}
If $n \in \NN$ then
$\DCPO^n$,
$\CPO^n$, $\Scott^n$ are enriched in $\DCPO$.
\end{lemma}

\begin{proof}
The result for $n=1$ follows from the
Cartesian-closure of $\DCPO$, $\CPO$ and $\Scott$.
In the cases of $n \neq 1$, the result follows from the fact that
in $\DCPO, \CPO, \Scott$, finite products are Cartesian products
of sets equipped with the component-wise order.
\end{proof}

Let $\cat C$ be a category enriched over $\DCPO$.
Given objects $X,Y \in \cat C$,
an \emph{embedding-projection} pair $X \to Y$
is a pair of morphisms $\ladj f: X\rightleftarrows Y:\radj f$
where $\radj f \comp \ladj f = \id_X$ and $\ladj f \comp \radj f \leq \id_Y$.
The morphism $\ladj f$ is an \emph{embedding}
(it reflects (as well as preserves) the order),
while $\radj f$ is called a \emph{projection}.
Note that if $X$ (resp.\ $Y$) has a least element,
then so does $Y$ (resp.\ $X$) and $\ladj f$ (resp $\radj f$)
is strict.%
\footnote{A Scott-continuous function is \emph{strict} if it preserves
least elements.}
It is well-known that $\ladj f$ completely determines $\radj f$
and reciprocally, see~\cite[\S 7.1]{ac98book}
(cf.\ also~\cite[\S 3.1.4]{aj95chapter} and~\cite[\S 9]{streicher06book}).
Given an embedding $e$ (resp.\ a projection $p$),
we write $\radj e$ (resp.\ $\ladj p$)
for the corresponding projection (resp.\ embedding).

We write $\cat C^\ep$ for the category with the same objects as $\cat C$,
and with embedding-projection pairs as morphisms.
Note that we have faithful functors
$\ladj{(\pl)} \colon \cat C^\ep \to \cat C$
and
$\radj{(\pl)} \colon \cat C^\ep \to \cat C^\op$,
taking $(\ladj f,\radj f)$ to $\ladj f$ and to $\radj f$,
respectively.
Given a functor $H$ of codomain $\cat C^\ep$,
we write $\radj H$ for $\radj{(\pl)} \comp H$,
and similarly for $\ladj H$.

\subsubsection{The Limit-Colimit Coincidence}
The following (crucial and) well-known fact
is \cite[Theorem 7.1.10]{ac98book}
(see also \cite[Theorem 3.3.7]{aj95chapter}).

\colin{In Theorem~\ref{thm:proof:scott:limcolim}
we need that composition in $\cat C$ is continuous
(similarly as in \cite[Theorem 7.1.10]{ac98book}).}

\begin{theorem}
\label{thm:proof:scott:limcolim}
Let $K \colon \omega \to \cat C^\ep$ be a functor,
where $\cat C$ is enriched over $\DCPO$.
Each limiting cone $\varpi \colon \Lim \radj K \to \radj K$
for $\radj K \colon \omega^\op \to \cat C$ consists of projections,
and the $(\ladj{(\varpi_n)},\varpi_n)_n$
form a colimiting cocone
$K \to \Colim K$ in $\cat C^\ep$.
\end{theorem}

\begin{proof}
Let
$K \colon \omega \to \cat C^\ep$
and consider a limiting cone in $\cat C$:
\begin{equation}
\label{diag:proof:scott:lim}
\begin{array}{c}
\begin{tikzcd}[column sep=2em] 
&
& \Lim \radj K
  \arrow{dll}[above]{\varpi_0}
  \arrow{dl}{\varpi_1}
  \arrow{d}{\varpi_2}
  \arrow{dr}[below]{\varpi_n}
  \arrow{drr}{\varpi_{n+1}}

\\

  \radj K(0)
& \radj K(1)
  \arrow{l}[below]{\radj{(k_0)}}
& \radj K(2)
  \arrow{l}[below]{\radj{(k_1)}}
& \radj K(n)
  \arrow[dashed]{l}
& \radj K(n+1)
  \arrow{l}[below]{\radj{(k_n)}}
& \phantom{F}
  \arrow[dashed]{l}
\end{tikzcd}
\end{array}
\end{equation}

\noindent
The components of the colimiting cocone
in $\cat C^\ep$
\begin{equation}
\label{diag:proof:scott:colim}
\begin{array}{c}
\begin{tikzcd}[column sep=2em]
&
& \Colim K

\\

  K(0)
  \arrow{r}[below]{k_0}
  \arrow{urr}[above]{\gamma_0}
& K(1)
  \arrow{r}[below]{k_1}
  \arrow{ur}[below]{\gamma_1}
& K(2)
  \arrow[dashed]{r}
  \arrow{u}{\gamma_2}
& K(n)
  \arrow{r}[below]{k_n}
  \arrow{ul}[below]{\gamma_n}
& K(n+1)
  \arrow[dashed]{r}
  \arrow{ull}[above]{\gamma_{n+1}}
& \phantom{F}
\end{tikzcd}
\end{array}
\end{equation}

\noindent
are given by $\radj{(\gamma_n)} = \varpi_n$ for projections.

Concerning embeddings,
for each $n \in \NN$ we build a cone with vertex $K(n) = \radj K(n)$.
Given $m \in \NN$, we have a morphism $h_{n,m} \colon K(n) \to K(m)$
obtained by composing $\radj{(k_i)}$'s or $\ladj{(k_i)}$'s
according to whether $m \leq n$ or $n \leq m$.
The $h_{n,m}$'s can be made so that $h_{n,m} = \radj{(k_m)} \comp h_{n,m+1}$.
The universal property of limits in $\cat C$ then yields
a unique morphism $c_n$ from $K(n) = \radj K(n)$ to $\Lim \radj K(n)$
such that $\varpi_m \comp c_n = h_{n,m}$ for all $m \in \NN$.

We are going to show that $c_n = \ladj{(\varpi_n)}$.
Note that $\varpi_n \comp c_n$ is the identity by definition of $c_n$.
It remains to show that $c_n \comp \varpi_n \leq \id_{\Lim \radj K}$.
We first show that
$(c_n \comp \varpi_n)_n$ forms an increasing sequence
in $\cat C(\Lim \radj K, \Lim \radj K)$.
To this end, note that $\varpi_n = \radj{(k_n)} \comp \varpi_{n+1}$ 
(since $\varpi$ is a cone).
We moreover have
$c_n = c_{n+1} \comp \ladj{(k_n)}$
since
\(
  \varpi_m \comp c_{n+1} \comp \ladj{(k_n)}
  =
  h_{n+1,m} \comp \ladj{(k_n)}
  =
  h_{n,m}
\)
for all $m \in \NN$.
We compute
\[
\begin{array}{*{5}{l}}
  c_n \comp \varpi_n
& =
& c_{n+1} \comp \ladj{(k_n)} \comp \radj{(k_n)} \comp \varpi_{n+1}
& \leq
& c_{n+1} \comp \varpi_{n+1}
\end{array}
\]

Let $\ell = \bigvee_{n}(c_n \comp \varpi_n)$.
We now claim that $\ell$ is the identity.
This will yield that $c_n \comp \varpi_n \leq \id_{\Lim \radj K}$.
In order to show that $\ell = \id_{\Lim \radj K}$,
we show that $\varpi_m \comp \ell = \varpi_m$ for all $m \in \NN$,
and use the universal property of limits in $\cat C$.
We have
\[
\begin{array}{l l l}
  \varpi_m \comp \ell
& =
& \varpi_m \comp \bigvee_{n}(c_n \comp \varpi_n)
\\

& =
& \varpi_m \comp \bigvee_{n \geq m}(c_n \comp \varpi_n)
\\

& =
& \bigvee_{n \geq m} \left(
  \varpi_m \comp c_n \comp \varpi_n
  \right)
\\

& =
& \bigvee_{n\geq m} \left(
  h_{n,m} \comp \varpi_n
  \right)
\end{array}
\]

\noindent
But by definition of $h_{n,m}$, we have
$h_{n,m} \comp \varpi_n = \varpi_m$ when $m \leq n$.

We can thus set $\gamma_n = (c_n, \varpi_n)$.
Moreover, $\gamma = (\gamma_n)_n$ is indeed a cocone since
$c_n = c_{n+1} \comp \ladj{(k_n)}$ (see above).

We now claim that $\gamma \colon K \to \Lim \radj K$ is colimiting.
To this end, consider a cocone $\tau \colon K \to C$.
We thus get a cone $\radj\tau \colon \radj C \to \radj K$ in $\cat C$,
and the universal property of limits yields a unique
$p \colon \radj C \to \Lim \radj K$ such that
$\varpi_n \comp p = \radj{(\tau_n)}$ for all $n \in \NN$.
We show that $p$ is a projection.
We define a morphism $e \colon \Lim \radj K \to C$
as $e = \bigvee_{n}(\ladj{(\tau_n)} \comp \varpi_n)$.
We have
\[
\begin{array}{l l l}
  e \comp p
& =
& \left( \bigvee_{n} \ladj{(\tau_n)} \comp \varpi_n \right) \comp p
\\

& =
& \bigvee_n \ladj{(\tau_n)} \comp \radj{(\tau_n)}
\\

& \leq
& \id_{C}
\end{array}
\]

\noindent
On the other hand, given $m \in \NN$ we have
\[
\begin{array}{l l l}
  \varpi_m \comp p \comp e
& =
& \bigvee_{n} \radj{(\tau_m)} \comp \ladj{(\tau_n)} \comp \varpi_n
\\

& =
& \bigvee_{n \geq m} \radj{(\tau_m)} \comp \ladj{(\tau_n)} \comp \varpi_n
\\

& =
& \bigvee_{n \geq m} h_{n,m} \comp \radj{(\tau_n)} \comp \ladj{(\tau_n)} \comp \varpi_n
\\

& =
& \bigvee_{n \geq m} h_{n,m} \comp \varpi_n
\\

& =
& \bigvee_{n \geq m} \varpi_m
\\

& =
& \varpi_m
\end{array}
\]

\noindent
so that $p \comp e = \id_{\Lim \radj K}$
by the universal property of limits in $\cat C$.

Moreover, for all $n \in \NN$ we have
\[
\begin{array}{l l l}
  e \comp c_n
& =
& \bigvee_m \ladj{(\tau_m)} \comp \varpi_m \comp c_n
\\

& =
& \bigvee_m \ladj{(\tau_m)} \comp h_{n,m}
\\

& =
& \bigvee_{m\geq n} \ladj{(\tau_m)} \comp h_{n,m}
\\

& =
& \bigvee_{m\geq n} \ladj{(\tau_n)}
\\

& =
& \ladj{(\tau_n)}
\end{array}
\]

Consider now a morphism $\ell \colon \Lim \radj K \to C$
in $\cat C^\ep$
such that 
$\varpi_n \comp \radj\ell = \radj{(\tau_n)}$
and
$\ladj\ell \comp c_n = \ladj{(\tau_n)}$
for all $n \in \NN$.
The universal property of limits in $\cat C$
yields $\radj\ell = p$, so that $\ladj\ell = e$
since $e$ is uniquely determined from $p$.
\end{proof}

\subsubsection{Solutions of Domain Equations}
We shall use Theorem~\ref{thm:proof:scott:limcolim}
in the following situation.
Consider a functor
\[
\begin{array}{*{4}{l}}
  G \colon
& \cat D^\ep \times \cat C^\ep
& \longto
& \cat C^\ep
\end{array}
\]

\noindent
where $\cat C$ and $\cat D$ are enriched over $\DCPO$.
We moreover assume that $\cat C$ has a terminal object $\one$
which is initial in $\cat C^\ep$.
We are going to define a functor
\[
\begin{array}{*{4}{l}}
  K \colon
& \cat D^\ep \times \omega
& \longto
& \cat C^\ep
\end{array}
\]

\noindent
Given an object $B$ of $\cat D^\ep$,
$K(B,\pl)$ is the $\omega$-chain in $\cat C^\ep$
obtained by iterating $G_B = G(B,\pl)$ from the initial object $\one$ of $\cat C^\ep$:
\begin{equation}
\label{diag:proof:scott:chain}
\begin{tikzcd} 
  \one
  \arrow{r}[above]{\one}
& G_B(\one)
  \arrow{r}[above]{G_B(\one)}
& G^2_B(\one)
  \arrow[dashed]{r}
& G^n_B(\one)
  \arrow{r}[above]{G_B^n(\one)}
& G^{n+1}_B(\one)
  \arrow[dashed]{r}
& \phantom{F}
\end{tikzcd}
\end{equation}

Given a morphism $f \colon B \to B'$ in $\cat D^\ep$,
$K(f,\pl)$ is obtained by commutativity of the following.
\begin{equation}
\label{diag:proof:scott:natdiag}
\begin{array}{c}
\begin{tikzcd} 

  \one
  \arrow{r}[above]{\one}
  \arrow{d}{\one}
& G_B(\one)
  \arrow{r}[above]{G_B(\one)}
  \arrow{d}{G_f(\one)}
& G^2_B(\one)
  \arrow{d}{G_f^2(\one)}
  \arrow[dashed]{r}
& G^n_B(\one)
  \arrow{r}[above]{G_B^n(\one)}
  \arrow{d}{G_f^n(\one)}
& G^{n+1}_B(\one)
  \arrow{d}{G_f^{n+1}(\one)}
  \arrow[dashed]{r}
& \phantom{F}

\\

  \one
  \arrow{r}[below]{\one}
& G_{B'}(\one)
  \arrow{r}[below]{G_{B'}(\one)}
& G^2_{B'}(\one)
  \arrow[dashed]{r}
& G^n_{B'}(\one)
  \arrow{r}[below]{G_{B'}^n(\one)}
& G^{n+1}_{B'}(\one)
  \arrow[dashed]{r}
& \phantom{F}
\end{tikzcd}
\end{array}
\end{equation}

Assume now that $\cat C$ has limits of $\omega^\op$-chains
of projections.
Then Theorem~\ref{thm:proof:scott:limcolim}
yields that each $K(B,\pl)$ has a colimit in $\cat C^\ep$.
Since $K$ is a functor $\cat D^\ep \times \omega \to \cat C^\ep$,
it follows from \cite[Theorem V.3.1]{maclane98book}
that these colimits assemble into a functor
\[
\begin{array}{l r c l}
  \Fix G \colon
& \cat D^\ep
& \longto
& \cat C^\ep
\\

& B
& \longmapsto
& \Colim_{n \in \omega} K(B,n)
\end{array}
\]

If $G(B,\pl)$ preserves colimits of $\omega$-chains,
then the universal property of colimits gives an isomorphism
\(
  \fold^\ep
  :
  G(B,\Fix G(B))
  \rightleftarrows
  \Fix G(B)
  :
  \unfold^\ep
\)
in $\cat C^\ep$.

We are going to prove the following.

\begin{proposition}
\label{prop:proof:scott:contfunct}
If $G \colon \cat D^\ep \times \cat C^\ep \to \cat C^\ep$
preserves colimits of $\omega$-chains,
then so does $\Fix G \colon \cat D^\ep \to \cat C^\ep$.
\end{proposition}

The proof of Proposition~\ref{prop:proof:scott:contfunct}
is split into the following lemmas.
Fix a functor $G \colon \cat D^\ep \times \cat C^\ep \to \cat C^\ep$
which preserves colimits of $\omega$-chains.

\begin{lemma}
\label{lem:proof:scott:contdiag}
The diagonal functor $\modgen \colon \cat D^\ep \to \cat D^\ep \times \cat D^\ep$
preserves colimits of $\omega$-chains.
\end{lemma}

\begin{proof}
Since colimits are pointwise in functor categories
(\cite[Corollary V.3]{maclane98book}).%
\footnote{Note that \cite[Corollary V.3]{maclane98book} only gives the result for limits.
But recall that the opposite of a functor category $[\cat C,\cat D]$
is the functor category $[\cat C^\op, \cat D^\op]$.}
\end{proof}

Lemma~\ref{lem:proof:scott:contdiag}
entails in particular that each functor
$G^n_{(\pl)}(\one) \colon \cat D^\ep \to \cat C^\ep$
preserves colimits of $\omega$-chains
($G^{n+1}_{(\pl)}(\one)$ is
$G(\pl,G^{n}_{(\pl)}(\one)) \comp \modgen$).

Proposition~\ref{prop:proof:scott:contfunct}
relies on the fact that the functor
$K \colon \cat D^\ep \to \funct{\omega,\cat C^\ep}$
preserves colimits of $\omega$-chains.
This involves some notation.

Let $W \colon \omega \to \cat D^\ep$ be an $\omega$-chain,
with colimiting cocone $\gamma \colon W \to \Colim W$.
In the following, we write $w_m \colon W(m) \to W(m+1)$
for the connecting morphisms of $W$.
The cocone $K \gamma \colon K(W) \to K(\Colim W)$
has component at $m \in \NN$
the commutative diagram in \eqref{diag:proof:scott:natdiag}
where one takes $\gamma_m \colon W(m) \to \Colim W$
for $f \colon B \to B'$.

\begin{lemma}
\label{lem:proof:scott:colimiting}
The cocone $K\gamma \colon K(W) \to K(\Colim W)$
is colimiting.
\end{lemma}

\begin{proof}
First, it follows from the above that each
$G_\gamma^n(\one) \colon G_W^n(\one) \to G_{\Colim W}^n(\one)$
is colimiting.

Consider now a cocone
$\tau \colon K(W) \to H$ in $\funct{\omega,\cat C^\ep}$.
For each $m \in \NN$, we have
$\tau_m = \tau_{m+1} \comp K(w_m)$,
that is
\begin{equation}
\label{diag:proof:scott:taucocone}
\begin{array}{c}
\begin{tikzcd}[column sep=2.14em, row sep=large]

  \one
  \arrow{r}[above]{\one}
  \arrow{d}{\one}
& G_{B}(\one)
  \arrow{r}[above]{G_{B}(\one)}
  \arrow{d}{G_{w_m}(\one)}
& G^2_{B}(\one)
  \arrow{d}{G_{w_m}^2(\one)}
  \arrow[dashed]{r}
& G^n_{B}(\one)
  \arrow{r}[above]{G_{B}^n(\one)}
  \arrow{d}{G_{w_m}^n(\one)}
& G^{n+1}_{B}(\one)
  \arrow{d}{G_{w_m}^{n+1}(\one)}
  \arrow[dashed]{r}
& \phantom{F}

\\

  \one
  \arrow{r}[above]{\one}
  \arrow{d}{(\tau_{m+1})_0}
& G_{B'}(\one)
  \arrow{r}[above]{G_{B'}(\one)}
  \arrow{d}{(\tau_{m+1})_1}
& G^2_{B'}(\one)
  \arrow{d}{(\tau_{m+1})_2}
  \arrow[dashed]{r}
& G^n_{B'}(\one)
  \arrow{r}[above]{G_{B'}^n(\one)}
  \arrow{d}{(\tau_{m+1})_n}
& G^{n+1}_{B'}(\one)
  \arrow{d}{(\tau_{m+1})_{n+1}}
  \arrow[dashed]{r}
& \phantom{F}

\\

  H(0)
  \arrow{r}[below]{h(0)}
& H(1)
  \arrow{r}[below]{h(1)}
& H(2)
  \arrow[dashed]{r}
& H(n)
  \arrow{r}[below]{h(n)}
& H(n+1)
  \arrow[dashed]{r}
& \phantom{F}
\end{tikzcd}
\end{array}
\end{equation}

\noindent
where $B$ is $W(m)$, $B'$ is $W(m+1)$
and the $h(n) \colon H(n) \to H(n+1)$ are the connecting morphisms of $H$.
In particular, for each $m \in \NN$ and each $n \in \NN$,
we have $(\tau_m)_n = (\tau_{m+1})_n \comp G_{w_m}^n(\one)$.
Hence, for each $n \in \NN$ we have a cocone
$((\tau_m)_n)_m \colon G_{W(-)}^n(\one) \to H(n)$,
and the universal property of $G_\gamma^n(\one)$
gives a unique morphism $\ell_n \colon G_{\Colim W}^n(\one) \to H(n)$
such that $(\tau_m)_n = \ell_n \comp G_{\gamma_m}^n(\one)$ for all $m \in \NN$.

We show that the $\ell_n$'s assemble into a morphism
$\ell \colon K(\Colim W) \to H$ in $\funct{\omega,\cat C^\ep}$.
We thus have to show that the following commutes,
where $B'$ is $\Colim W$:
\begin{equation*}
\begin{array}{c}
\begin{tikzcd} 

  \one
  \arrow{r}[above]{\one}
  \arrow{d}{\ell_0}
& G_{B'}(\one)
  \arrow{r}[above]{G_{B'}(\one)}
  \arrow{d}{\ell_1}
& G^2_{B'}(\one)
  \arrow{d}{\ell_2}
  \arrow[dashed]{r}
& G^n_{B'}(\one)
  \arrow{r}[above]{G_{B'}^n(\one)}
  \arrow{d}{\ell_n}
& G^{n+1}_{B'}(\one)
  \arrow{d}{\ell_{n+1}}
  \arrow[dashed]{r}
& \phantom{F}

\\

  H(0)
  \arrow{r}[below]{h(0)}
& H(1)
  \arrow{r}[below]{h(1)}
& H(2)
  \arrow[dashed]{r}
& H(n)
  \arrow{r}[below]{h(n)}
& H(n+1)
  \arrow[dashed]{r}
& \phantom{F}

\end{tikzcd}
\end{array}
\end{equation*}

We show that $\ell_{n+1} \comp G_{B'}^n(\one) = h(n) \comp \ell_n$
for all $n \in \NN$.
For each $m \in \NN$, 
by commutativity of \eqref{diag:proof:scott:natdiag}
and \eqref{diag:proof:scott:taucocone}
we have the following, where $B$ is $W(m)$:
\[
\begin{array}{l l l}
  \ell_{n+1} \comp G_{B'}^{n}(\one) \comp G_{\gamma_m}^{n}(\one)
& =
& \ell_{n+1} \comp G_{\gamma_m}^{n+1}(\one) \comp G_B^{n}(\one)
\\

& =
& (\tau_m)_{n+1} \comp G_B^{n}(\one)
\\

& =
& h(n) \comp (\tau_m)_n
\\

& =
& h(n) \comp \ell_n \comp G_{\gamma_m}^n(\one)
\end{array}
\]

\noindent
Then we are done by the universal property of $G_{\gamma_m}^n(\one)$.

Consider finally a morphism $f \colon K(\Colim W) \to H$ in
$\funct{\omega,\cat C^\ep}$
such that $f \comp K(\gamma) = \tau$.
Then for all $m \in \NN$ we have
$f \comp K(\gamma_m) = \tau_m$,
and thus
$f_n \comp G_{\gamma_m}^n(\one) = (\tau_m)_n$
for all $n \in \NN$.
It follows that $f_n = \ell_n$, so that $f = \ell$.
\end{proof}

We can now prove Proposition~\ref{prop:proof:scott:contfunct}.

\begin{proof}[Proof of Proposition~\ref{prop:proof:scott:contfunct}]
Let $W \colon \omega \to \cat D^\ep$ be an $\omega$-chain.
By Lemma~\ref{lem:proof:scott:colimiting},
and since colimits always commute over colimits,
we have
\[
\begin{array}{l l l}
  \Fix G(\Colim W)
& =
& \Colim_{n \in \omega} K(\Colim W,n)
\\

& \cong
& \Colim_{n \in \omega} \Colim_{m \in \omega} K(W(m),n)
\\

& \cong
& \Colim_{m \in \omega} \Colim_{n \in \omega} K(W(m),n)
\\

& \cong
& \Colim_{m \in \omega} \Fix G(W(m))
\end{array}
\]
\end{proof}

\colin{TODO: Remark on strictness}

\subsubsection{Local Continuity}
Functors $G \colon \cat D^\ep \times \cat C^\ep \to \cat E^\ep$
will be obtained from ``mixed-variance'' functors
\[
\begin{array}{*{4}{l}}
  F \colon
& \cat D^\op \times \cat C
& \longto
& \cat E
\end{array}
\]

\noindent
where
$\cat D,\cat C, \cat E$ are enriched over $\DCPO$.

\begin{definition}
\label{def:proof:scott:loc}
We say that $F$
is \emph{locally} \emph{monotone} (resp.\ \emph{continuous})
if each hom-function
\[
\begin{array}{r c l}
  \cat D(B',B) \times \cat C(A,A')
& \longto
& \cat C(F(B,A), F(B',A'))
\\

  (g,f)
& \longmapsto
& F(g,f)
\end{array}
\]

\noindent
is monotone (resp.\ Scott-continuous).
\end{definition}

\noindent
We refer to \cite[Definition 5.2.5]{aj95chapter},
\cite[Definition 7.1.15]{ac98book}
and \cite[Definition 9.1]{streicher06book}.
The following is a straightforward adaptation of \cite[Proposition 7.1.19]{ac98book}
(see also \cite[Proposition 5.2.6]{aj95chapter}).

\begin{lemma}
\label{lem:proof:scott:lift}
Let
$F \colon \cat D^\op \times \cat C \to \cat E$
be locally monotone.
Then $F$ lifts to a covariant functor
\[
\begin{array}{*{4}{l}}
  F^\ep \colon
& \cat D^\ep \times \cat C^\ep
& \longto
& \cat E^\ep
\end{array}
\]

\noindent
with $F^\ep(B,A) = F(B,A)$ on objects and
$F^\ep(g,f) = (F(\radj g,\ladj f) \,,\, F(\ladj g,\radj f))$
on morphisms.

If moreover $F$ is locally continuous, then $F^\ep$
preserves colimits of $\omega$-chains.
\end{lemma}

\subsection{Interpretation of Pure Types}

A \emph{pure type expression} is a possibly open production of the
grammar of pure types (\S\ref{sec:pure}), namely
\[
\begin{array}{r @{\ \ }c@{\ \ } l}
    \PT
&   \bnf
&   \Unit
\gs \PT \times \PT
\gs \PT \arrow \PT
\gs \PT + \PT
\gs \TV
\gs \rec \TV.\PT 
\end{array}
\]

\noindent
where $\TV$ is a type variable,
and where $\rec\TV.\PT$ binds $\TV$ in $\PT$.

Consider a pure type expression $\PT$ with free
type variables $\vec \TV = \TV_1,\dots,\TV_n$.
We are going to interpret $\PT$ as a functor
\[
\begin{array}{*{4}{l}}
  \I\PT \colon
& \left( \Scott^\ep \right)^{n}
& \longto
& \Scott^\ep
\end{array}
\]

\noindent
which preserves colimits of $\omega$-chains.

Since $\Scott^\ep$ has the same objects as $\Scott$,
this thus provides the interpretation $\I\PT \in \Scott$
of each (pure) \emph{type} $\PT$
(i.e.\ each type expression $\PT$ without free type variables $\TV$).

\subsubsection{Preliminaries}
Recall that the category $\Scott$ is Cartesian-closed
(products and homsets are equipped with pointwise orders),
see \cite[Corollary 4.1.6]{aj95chapter} or \cite[\S 1.4]{ac98book}.
This yields functors
\[
\begin{array}{r r c l}
  \Scott(\pl,\pl) \colon
& \Scott^\op \times \Scott
& \longto
& \Scott
\\

  (\pl) \times (\pl) \colon
& \Scott \times \Scott
& \longto
& \Scott
\end{array}
\]

\noindent
These functors are locally continuous
(\cite[Example 7.1.16]{ac98book}).
By Lemma~\ref{lem:proof:scott:lift}
we obtain functors
\[
\begin{array}{*{4}{l}}
  \left(\Scott(\pl,\pl) \right)^\ep
  ,\,
  \left( (\pl) \times (\pl) \right)^\ep
  \colon
& \Scott^\ep \times \Scott^\ep
& \longto
& \Scott^\ep
\end{array}
\]

\noindent
which preserve colimits of $\omega$-chains.

Moreover, $\Scott$ has limits of $\omega^\op$-chains
of projections (in the embedding-projection sense),
see \cite[Theorem 3.3.7, Theorem 3.3.11 and Proposition 4.1.3]{aj95chapter}.
More precisely, the full inclusion $\Scott \emb \DCPO$
creates limits for $\omega^\op$-chains of projections.%
\footnote{The notion of creation of limits has to be understood in the usual
sense of \cite[Definition V.1]{maclane98book}.}
In particular, $\Scott$
is closed in $\DCPO$ under limits of $\omega^\op$-chains of projections.
Note that the category $\DCPO$ has all limits,
and that they are created by the forgetful functor to the category
of posets (and monotone functions),
see \cite[Theorem 3.3.1]{aj95chapter}.
It follows that given $K \colon \omega \to \cat \Scott^\ep$,
the limit of $\radj K \colon \omega^\op \to \Scott$
is
\[
  \left\{
  (x_i)_i \in \prod_{i \in \NN} K(i)
  \mathrel{\Big|}
  \radj K(i \leq j)(x_j) = x_i
  \right\}
\]

\noindent
equipped with the pointwise order.
Moreover, the limiting
cone $\Lim \radj K \to \radj K$ 
consists in set-theoretic projections.%
\footnote{These are also projections in the embedding-projection sense
by Theorem~\ref{thm:proof:scott:limcolim}.}
In view of Theorem~\ref{thm:proof:scott:limcolim},
we also get that $\Scott$ is closed in $\DCPO$
under colimits of $\omega$-chains of embeddings.

The terminal object $\one = \{\bot\}$ of $\Scott$ is initial in $\Scott^\ep$
(\cite[Proposition 7.1.9]{ac98book}).

\subsubsection{Interpretation of (weak) sum types}
\label{sec:proof:sem:pure:sum}

We now discuss the interpretation of $\PT + \PTbis$.
It is well-known that the categories $\Scott$ and $\CPO$,
as opposed to $\DCPO$, do not have coproducts~\cite{hp90tcs}.
We shall thus interpret $\PT + \PTbis$ as a \emph{weak} coproduct,
satisfying a weakening of the
universal property of coproducts:
we do not require the ``universal'' morphisms to be unique.

Consider first the composite functor
\begin{equation}
\label{eq:proof:sem:weaksum}
\begin{tikzcd}[column sep=large]
  \CPO \times \CPO
  \arrow{r}{\incl \times \incl}
& \DCPO \times \DCPO
  \arrow{r}{(\pl) \amalg (\pl)}
& \DCPO
  \arrow{r}{(\pl)_\bot}
& \CPO
\end{tikzcd}
\end{equation}

\noindent
Here, $\incl \colon \CPO \emb \DCPO$ is the full inclusion,
while $(\pl) \amalg (\pl)$ is the coproduct in $\DCPO$
(disjoint union with the pointwise order),
and $(\pl)_\bot \colon \DCPO \to \CPO$
adds a new least element.

Note that the inclusion
$\iota_\bot \colon X \emb X_\bot$
is Scott-continuous (since directed sets are non-empty)
and an order embedding.
Also, given $f \colon X \to Y$ and $g \colon X' \to Y'$
we have
\[
\begin{array}{l r c l}
  (f \amalg g)_\bot \colon
& (X \amalg Y)_\bot
& \longto
& (X' \amalg Y')_\bot
\\

& \bot
& \longmapsto
& \bot
\\

& (\iota_\bot \comp \kappa_1)(x)
& \longmapsto
& (\iota_\bot \comp \kappa_1)(f(x))
\\

& (\iota_\bot \comp \kappa_2)(y)
& \longmapsto
& (\iota_\bot \comp \kappa_2)(g(y))
\end{array}
\]

\noindent
where $\kappa_1$ and $\kappa_2$ are the coprojection for
the coproduct $(\pl) \amalg (\pl)$ in $\DCPO$.

\begin{lemma}
\label{lem:proof:sem:weaksum:scott}
Given $X, Y \in \Scott$, we have $(X \amalg Y)_\bot \in \Scott$.
\end{lemma}

\begin{proof}
We have seen that $(X \amalg Y)_\bot$ is a cpo.
It is algebraic since the only finite approximation of $\bot \in (X \amalg Y)_\bot$
is $\bot$,
while given $x \in X$,
the finite approximations of $(\iota_\bot \comp \kappa_1)(x)$ are $\bot$
and the $(\iota_\bot \comp \kappa_1)(d)$'s
where $d$ is a finite approximation of $x$ in $X$
(and similarly with $(\iota_\bot \comp \kappa_2)(y)$ for $y \in Y$).

Moreover, $(X \amalg Y)_\bot$ is bounded-complete since if $a,b$
have an upper bound, then either both $a$ and $b$ are in the same
components, or at least one of them is $\bot$.
\end{proof}

\begin{definition}
\label{def:proof:sem:weaksum:scott}
Given $X, Y \in \Scott$,
we let $X + Y \deq (X \amalg Y)_\bot \in \Scott$.
\end{definition}

Let $X, Y \in \Scott$.
The ``coprojections''
$\inj_1 \colon X \to X + Y$
and
$\inj_2 \colon Y \to X + Y$
are the order embeddings defined as the composites
\[
\begin{array}{c !{\qquad\text{and}\qquad} c}

\begin{tikzcd}
  X
  \arrow{r}{\kappa_1}
& X \amalg Y
  \arrow[hookrightarrow]{r}{\iota_\bot}
& X + Y
\end{tikzcd}

&

\begin{tikzcd}
  Y
  \arrow{r}{\kappa_2}
& X \amalg Y
  \arrow[hookrightarrow]{r}{\iota_\bot}
& X + Y
\end{tikzcd}

\end{array}
\]

Further,
consider some $Z \in \Scott$,
together with $f \colon X \to Z$
and $g \colon Y \to Z$.
Let the dashed arrow $h$ in the diagram below be given by the
universal property of coproducts in $\DCPO$:
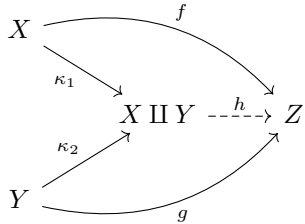
\[
\begin{tikzcd}
  X
  \arrow{dr}[swap]{\kappa_1}
  \arrow[bend left]{drr}{f}
\\
& X \amalg Y
  \arrow[dashed]{r}{h}
& Z
\\
  Y
  \arrow{ur}{\kappa_2}
  \arrow[bend right]{urr}[swap]{g}
\end{tikzcd}
\]

\noindent
We define $\copair{f,g} \colon X+Y \to Z$
to be the extension of $h$ with $\copair{f,g}(\bot) = \bot_Z$.

The function $\copair{f,g}$ is Scott-continuous.
Indeed, if $a \leq b$ in $X+Y$,
then either $a = \bot$
and $\copair{f,g}(a) = \bot_Z \leq \copair{f,g}(b)$,
or both $a$ and $b$ are in the same component.
Moreover, given a directed $D \sle X + Y$,
if $D \setminus \{\bot\}$ is non-empty,
then it is a directed set either in the component $X$ or in the component $Y$,
and
\(
  h(\bigvee D)
  =
  h(\bigvee (D\setminus \{\bot\}))
  =
  \bigvee h(D \setminus \{\bot\})
  =
  \bigvee \copair{f,g}(D)
\)
since $\copair{f,g}(\bot) = \bot_Z$.

Let now $f \colon X \to Y$
and $g \colon X' \to Y'$ in $\Scott$,
and consider 
$\copair{{\inj_1} \comp f, {\inj_2} \comp g}$
in
\[
\begin{tikzcd}[column sep=huge]
  X
  \arrow{dr}[swap]{\inj_1}
  \arrow[bend left]{drr}{{\inj_1} \comp f}
\\
& X + Y
  \arrow{r}{\copair{{\inj_1} \comp f,\, {\inj_2} \comp g}}
& X' + Y'
\\
  Y
  \arrow{ur}{\inj_2}
  \arrow[bend right]{urr}[swap]{{\inj_2} \comp g}
\end{tikzcd}
\]

Then $\copair{{\inj_1} \comp f, {\inj_2} \comp g} = f+g$,
where $f+g = (f \amalg g)_\bot$
(cf Equation~\eqref{eq:proof:sem:weaksum}).
Indeed,
we have
$\copair{{\inj_1} \comp f, {\inj_2} \comp g}(\bot) = \bot = (f+g)(\bot)$,
while
\(
  \copair{{\inj_1} \comp f, {\inj_2} \comp g}(\inj_1(x))
  =
  \inj_1(f(x))
  =
  (\iota_\bot \comp \kappa_1)(f(x))
  =
  (f \amalg g)_\bot((\iota_\bot \comp \kappa_1) (x))
  =
  (f+g)(\inj_1(x))
\),
and similarly for $\inj_2(y)$.

\begin{lemma}
\label{lem:proof:sem:weaksum:local}
The functor $(\pl) + (\pl) \colon \Scott \times \Scott \to \Scott$
is locally continuous.
\end{lemma}

\begin{proof}
Let $X,Y,X',Y' \in \Scott$,
and consider the hom-function
\[
\begin{array}{r c l}
  \Scott(X,Y) \times \Scott(X',Y')
& \longto
& \Scott(X+Y,\, X' + Y')
\\

  (f,g)
& \longmapsto
& f + g
\end{array}
\]

This function is monotone: given $f \leq f'$ and $g \leq g'$
we have
$f+g \leq f' + g'$
since
\(
  (f + g)(\bot)
  =
  \bot
  =
  (f' + g')(\bot)
\),
while 
\(
  (f + g)(\inj_1(x))
  \leq
  (f' + g')(\inj_1(x))
\)
since
\(
  (f + g)(\inj_1(x))
  =
  \copair{{\inj_1} \comp f, {\inj_2} \comp g}(\inj_1(x))
  =
  \inj_1(f(x))
\)
for $x \in X$,
and similarly for $y \in Y$.

This function is moreover Scott-continuous.
Recall (from e.g.~\cite[Proposition 1.4.3]{ac98book})
that Scott-continuity in several arguments is equivalent
to Scott-continuity in each argument separately.
Let $F \sle \CPO(X,Y)$ be directed, and let $g \in \CPO(X',Y')$.
Since directed sups of Scott-continuous functions are pointwise,
we have $((\bigvee F) + g)(\bot) = \bot = (\bigvee_{f \in F}(f+g))(\bot)$
and
$((\bigvee F) + g)(\inj_2(y)) = \inj_2(g(y)) = (\bigvee_{f \in F}(f+g))(\inj_2(y))$,
while
\[
\begin{array}{l l l}
  ((\bigvee F) + g)(\inj_1(x))
& =
& \inj_1((\bigvee F)(x))
\\

& =
& \inj_1(\bigvee_{f \in F}f(x))
\\

& =
& \bigvee_{f \in F}(\inj_1(f(x)))
\\

& =
& \bigvee_{f \in F}((f+g)(\inj_1(x)))
\\

& =
& (\bigvee_{f \in F}(f+g))(\inj_1(x))
\end{array}
\]

The case of a directed sup in the second argument of $(\pl) + (\pl)$
is similar.
\end{proof}

By combining Lemma~\ref{lem:proof:sem:weaksum:local}
with Lemma~\ref{lem:proof:scott:lift},
we obtain a functor
\[
\begin{array}{*{4}{l}}
  \left( (\pl) + (\pl) \right)^\ep
  \colon
& \Scott^\ep \times \Scott^\ep
& \longto
& \Scott^\ep
\end{array}
\]

\noindent
which preserves colimits of $\omega$-chains.

\subsubsection{Definition of the Interpretation}
\label{sec:proof:sem:pure:intdef}

Let $\PT$ be a (pure) type expression with free
type variables $\vec \TV = \TV_1,\dots,\TV_n$.
The interpretation $\I\PT \colon \left(\Scott^\ep\right)^n \to \Scott^\ep$
is defined by induction on $\PT$.
\begin{itemize}
\item
In the case of $\PT = \TV_i$,
we let $\I\PT$ take $\vec X = X_1,\dots,X_n$ to $X_i$.

\item
In the case of $\PT = \Unit$,
we let $\I\PT(\vec X)$ be $\{\bot, \top\}$ with $\bot \leq \top$.

Moreover,
we let $\I{\pair{}} \deq \top \in \I\Unit(\vec X)$.

\item
In the cases of $\PT \times \PTbis$
and $\PTbis \arrow \PT$,
the induction hypotheses give us
functors
\[
\begin{array}{*{4}{l}}
  \I\PT, \I\PTbis
  \colon
& \left( \Scott^\ep \right)^{n}
& \longto
& \Scott^\ep
\end{array}
\]

\noindent
which preserve colimits of $\omega$-chains.
We can thus set
\[
\begin{array}{r c l}
  \I{\PTbis \arrow \PT}(\vec X)
& =
& \left( \Scott \left( \I{\PTbis}(\vec X) ,\, \I{\PT}(\vec X) \right) \right)^\ep
\\

  \I{\PT \times \PTbis}(\vec X)
& =
& \left(
  \I\PT(\vec X) \times \I\PTbis(\vec X)
  \right)^\ep
\end{array}
\]

\item
In the case of $\PT + \PTbis$,
we similarly let 
\[
\begin{array}{r c l}
  \I{\PT + \PTbis}(\vec X)
& =
& \left(
  \I\PT(\vec X) + \I\PTbis(\vec X)
  \right)^\ep
\end{array}
\]

\noindent
where $(\pl) + (\pl)$ is the function in Definition~\ref{def:proof:sem:weaksum:scott}
(\S\ref{sec:proof:sem:pure:sum}).

Moreover,
we let $\I{\inj_1} \colon \I\PT(\vec X) \to \I\PT(\vec X) + \I\PTbis(\vec X)$
and $\I{\inj_2} \colon \I\PTbis(\vec X) \to \I\PT(\vec X) + \I\PTbis(\vec X)$
be the morphisms $\inj_1$ and $\inj_2$ defined in~\S\ref{sec:proof:sem:pure:sum}.

Further, the terms
$\cse\ M\ \copair{x_1 \mapsto N_1 \mid x_2 \mapsto N_2}$
are interpreted using the operation
$\copair{\pl, \pl}$, also defined in~\S\ref{sec:proof:sem:pure:sum}.

\item
In the case of $\rec \TV.\PT$,
the induction hypothesis gives a functor
\[
\begin{array}{*{4}{l}}
  \I\PT \colon
& \left( \Scott^\ep \right)^{n} \times \Scott^\ep
& \longto
& \Scott^\ep
\end{array}
\]

\noindent
which preserves colimits of $\omega$-chains.
Theorem~\ref{thm:proof:scott:limcolim}
gives a functor
\[
\begin{array}{l r c l}
  \I{\rec \TV.\PT} \colon
& \left( \Scott^\ep \right)^{n} 
& \longto
& \Scott^\ep
\\

& \vec X
& \longmapsto
& \Fix (\I{\PT}(\vec X))
\end{array}
\]

\noindent
This functor preserves colimits of $\omega$-chains by
Proposition~\ref{prop:proof:scott:contfunct}.
Moreover, since $\I\PT$ preserves
colimits of $\omega$-chains,
we obtain canonical isomorphisms
\[
\begin{array}{l @{\,:~~} l l l @{~~:\,} l}
  \I\fold
& \I{\PT[\rec \TV.\PT/\TV]}(\vec X)
& \rightleftarrows
& \I{\rec \TV.\PT}(\vec X)
& \I\unfold
\end{array}
\]

\noindent
by taking
$\I\fold = \ladj{(\fold^\ep)}$
and
$\I\unfold = \ladj{(\unfold^\ep)}$.
\end{itemize}

\begin{figure}[t!]
\centering
\(
\begin{array}{c}

\dfrac{}
  {\bot \in \Fin(\I\PT(\vec X))}  

\qquad

\dfrac{\text{$d$ finite in $X_i$}}
  {d \in \Fin(\I{\TV_i}(\vec X))}

\\\\

\dfrac{}
  {\top \in \Fin(\I\Unit(\vec X))}

\qquad

\dfrac{d \in \Fin(\I\PT(\vec X))}
  {\I{\inj_1}(d) \in \Fin(\I{\PT + \PTbis}(\vec X))}

\qquad

\dfrac{d \in \Fin(\I\PTbis(\vec X))}
  {\I{\inj_2}(d) \in \Fin(\I{\PT + \PTbis}(\vec X))}

\\\\

\dfrac{d \in \Fin(\I\PT(\vec X))
  \qquad
  e \in \Fin(\I\PTbis(\vec X))}
  {(d,e) \in \Fin(\I{\PT \times \PTbis}(\vec X))}

\qquad

\dfrac{d \in \Fin(\I{\PT[\rec\TV.\PT/\TV]}(\vec X))}
  {\I\fold(d) \in \Fin(\I{\rec\TV.\PT}(\vec X))}

\\\\

\dfrac{\begin{array}{l}
  \text{for all $i \in I$,~
  $d_i \in \Fin(\I{\PT})(\vec X)$
  ~and~
  $e_i \in \Fin(\I{\PTbis})(\vec X)$ \@;}
  \\
  \text{for all $J \sle I$,~
  $\bigvee_{j \in J} d_j$ defined in $\I{\PT}(\vec X)$
  ~$\imp$~
  $\bigvee_{j \in J} e_j$ defined in $\I{\PTbis}(\vec X)$}
  \end{array}}
  {\bigvee_{i \in I}(d_i \step e_i) \in \Fin(\I{\PT \arrow \PTbis}(\vec X))}
~(\text{$I$ finite})

\end{array}
\)
\caption{Inductive description of the finite elements of $\I\PT(\vec X)$.%
\label{fig:proof:sem:finelt}}
\end{figure}

\subsubsection{Description of the Finite Elements}
\label{sec:proof:sem:pure:finite}

For each (pure) type expression $\PT$
with free variables $\vec\TV = \TV_1,\dots,\TV_n$,
we define a set $\Fin(\I{\PT}(\vec X))$.
The definition is by induction on derivations with
the rules in Figure~\ref{fig:proof:sem:finelt}.
The set $\Fin(\I{\PT}(\vec X))$
describes the finite elements of $\I{\PT}(\vec X)$.
This relies on the following.

First, the least element $\bot$ is finite in $X$
for each $X \in \Scott$.
Moreover, $\top$ is finite in $\I\Unit = \{\bot,\top\}$.

Let $X,Y \in \Scott$.
The finite elements in the weak coproduct $X + Y$
are $\bot$,
as well as the $\I{\inj_1}(d)$'s 
and the $\I{\inj_2}(e)$'s
where $d$ is finite in $X$ and $e$ is finite in $Y$.
The finite elements in the product $X \times Y$
are exactly the pairs of finite elements.
The finite elements of $\Scott(X,Y)$ are exactly the finite sups of step functions.
Given finite $d \in X$ and $e \in Y$,
the \emph{step function} $(d \step e) \colon X \to Y$
is defined as $(d \step e)(x) = e$ if $x \geq d$ and
$(d \step e)(x)= \bot$ otherwise.
Recall that the sup $\bigvee_{i \in I}(d_i \step e_i)$ of 
a finite family of step functions exists
if, and only if,
for every $J \sle I$, the set $\{e_j \mid j \in J\}$ has an upper bound
whenever so does $\{d_j \mid j \in J\}$.
See \cite[Theorem 1.4.12]{ac98book}.

Concerning recursive types, let 
\[
\begin{array}{*{4}{l}}
  G \colon
& \Scott^\ep
& \longto
& \Scott^\ep
\end{array}
\]

\noindent
be a functor which preserves colimits of $\omega$-chains.
Recall that $\Fix G$ is the colimit in \eqref{diag:proof:scott:colim}
where $K \colon \omega \to \Scott^\ep$
takes $n$ to $G^n(\one)$
(similarly as in \eqref{diag:proof:scott:chain}).
We have seen that $\Scott$ is closed in $\DCPO$
under colimits of $\omega$-chains of embeddings.
Hence it follows from \cite[Theorem 3.3.11]{aj95chapter}
that the finite elements of $\Fix G$
are the images of the finite elements of the $G^n(\one)$'s
under the components of the colimiting cocone
$\gamma \colon K \to \Fix G$.

We thus have the following.

\begin{proposition}
\label{prop:proof:scott:fin}
$\Fin(\I\PT(\vec X))$ is the set of finite elements of $\I\PT(\vec X)$.
\end{proposition}

\subsection{Proofs of Example~\ref{ex:scott:stream-tree}}
\label{sec:proof:sem:stream-tree}

We now provide some details 
on $\I{\Stream\PTbis}$ and $\I{\Tree\PTbis}$
(Example~\ref{ex:scott:stream-tree}),
where $\PTbis$ is a pure type.
We handle streams and binary trees
uniformly by considering the covariant functor
\[
\begin{array}{l r c l}
  F \colon
& \Scott
& \longto
& \Scott
\\

& X
& \longmapsto
& \I\PTbis \times X^\Dir
\end{array}
\]

\noindent
where $\Dir$ is a finite set.
In view of Theorem~\ref{thm:proof:scott:limcolim},
$\Fix F$ is the limit of the $\omega^\op$-chain
\[
\begin{tikzcd} 
  \one
& F(\one)
  \arrow{l}[above]{\one}
& F^2(\one)
  \arrow{l}[above]{F(\one)}
& F^n(\one)
  \arrow[dashed]{l}
& F^{n+1}(\one)
  \arrow{l}[above]{F^n(\one)}
& \phantom{F}
  \arrow[dashed]{l}
\end{tikzcd}
\]

\noindent
where $\one$ is the terminal Scott domain $\{\bot\}$.
Hence, $\Fix F$ is
\[
\begin{array}{c}
  \left\{
  x \in \prod_{n \in \NN} F^n(\one)
  ~\Big|~
  x(n) = F^n(\one)(x(n+1))
  \right\}
\end{array}
\]

\noindent
equipped with the pointwise order.
In order to see that $\Fix F$ is isomorphic to $\I\PTbis^{\Dir^*}$,
define for each $n \in \NN$ an
isomorphism $\iota_n \colon \I\PTbis^{\Dir^n} \to F^n(\one)$
as $\iota_0 = \one \colon \one \to \one$
and
\[
\begin{array}{l r c l}
  \iota_{n+1} \colon
& \I\PTbis^{\Dir^{n+1}}
& \longto
& F^{n+1}(\one) = \I\PTbis \times \left( F^n(\one) \right)^\Dir
\\

& T
& \longmapsto
& \left( T(\es),~ (\iota_n(u \mapsto T(d \cdot u)))_{d \in \Dir}  \right)
\end{array}
\]

\noindent
Then observe that the following commutes
\[
\begin{tikzcd}
  \I\PTbis^{\Dir^{n+1}}
  \arrow{d}[left]{T \mapsto T\restr \Dir^n}
  \arrow{r}{\iota_{n+1}}
& F^{n+1}(\one)
  \arrow{d}{F^n(\one)}
\\
  \I\PTbis^{\Dir^{n}}
  \arrow{r}[below]{\iota_{n}}
& F^n(\one)
\end{tikzcd}
\]

The characterisation of the finite elements then follows from
Proposition~\ref{prop:proof:scott:fin}.

\subsubsection{Martin Hofmann's Breadth-First Tree Traversal}
\label{sec:proof:sem:stream-tree:bft}

We now show that the function
$\bft \colon \Tree\PTbis \arrow \Stream\PTbis$
in Table~\ref{tab:ex} (\S\ref{sec:pure})
indeed computes a tree breadth-first traversal.
We work at the denotational level.

It will be convenient to see $\I{\Tree\PTbis}$
as $(1 \cdot \{0,1\}^*) \to \I\PTbis$
rather than $\{0,1\}^* \to \I\PTbis$
(via the bijection $\{0,1\}^* \to 1 \cdot \{0,1\}^*$ which prepends a $1$).
Likewise, $\I{\Stream\PTbis}$
will be seen as $\NN_{>0} \to \I\PTbis$ rather than
$\NN \to \I\PTbis$.
With these conventions, the \emph{breadth-first traversal}
of a tree $t \in \I{\Tree\PTbis}$
is the stream $x \in \I{\Stream\PTbis}$
such that $x(i) = t(b(i))$,
where $b(i)$ is the binary representation of $i \in \NN_{>0}$
with most significant digit first
(hence $b(1) = 1$ and $b(2) = 10$).
Note that 
in a tree $t$,
the children of the node $b(i)$ (of label $t(b(i))$)
are the nodes $b(2i)$ and $b(2i+1)$
(of labels $t(b(2i))$ and $t(b(2i+1))$).

Fix a tree $t \in \I{\Tree\PTbis}$,
and let $x \in \I{\Stream\PTbis}$ be its breadth-first traversal.

Given a binary string $b \in 1 \cdot\{0,1\}^*$,
we write $t \restr b$ for the tree taking $1 \cdot b' \in 1 \cdot \{0,1\}^*$
to $t(b \cdot b')$ (thus dropping the first $1$ of $1 \cdot b'$).
Note that $t\restr 1 = t$.

For notational simplicity, we drop the brackets $\I{-}$
around interpretations of terms
(so that e.g.\ $\bftrec$ stands for $\I\bftrec$).
Given $i \in \NN_{>0}$ and $n \in \NN$,
let
\[
\begin{array}{r c l}
  a(i)
& =
& \bftrec\ (t \restr b(i))
\\

  c_n
& =
& a(2^n) \comp \dots \comp a(2^{n+1}-1)
\end{array}
\]

\noindent
Note that
\(
  a(i), c_n \colon
  \I{\Cont(\Stream \PTbis)}
  \to
  \I{\Cont(\Stream \PTbis)}
\)
and that $c_0 = a(1) = (\bftrec\ t)$.

\begin{lemma}
\label{lem:proof:sem:stream-tree}
For all $n \in \NN$, we have
\[
\begin{array}{r c l}
  \extract\ (c_n \Over)
& =
& x(2^n) \Colon \ldots \Colon x(2^{n+1}-1) \Colon
  (\extract\ (c_{n+1} \Over))
\end{array}
\]
\end{lemma}

\begin{proof}
Note that given $c \in \I{\Cont(\Stream\PTbis)}$
and $i \in \NN_{>0}$,
since
$\lft(t \restr b(i)) = t \restr b(2i)$
and
$\rght(t \restr b(i)) = t \restr b(2i+1)$,
we have
\[
\begin{array}{r c l}
  \unfold\ (a(i)\ c)
& =
& \lambda k.~
  t(b(i)) \Colon
  \unfold c\ (k \comp a(2i) \comp a(2i+1))
\\
& =
& \lambda k.~
  x(i) \Colon
  \unfold c\ (k \comp a(2i) \comp a(2i+1))
\end{array}
\]

\noindent
and thus
\[
\begin{array}{r c l}
  \extract\ (a(2^n)\ c)
& =
& x(2^n) \Colon
  \unfold c\ ({\extract} \comp {a(2^{n+1})} \comp a(2^{n+1}+1))
\end{array}
\]

\noindent
It follows that
\[
\begin{array}{l l l}
  \extract\ (c_n \Over)
& =
& \extract\ \big( (a(2^n) \comp \dots \comp a(2^{n+1}-1)) \Over \big)
\\

& =
& \extract\
  \big( a(2^n) \underbrace{((a(2^n+1) \comp \dots \comp a(2^{n+1}-1)) \Over)}_c \big)
\\

& =
& x(2^n) \Colon
  \unfold\ c\ ({\extract} \comp {a(2^{n+1})} \comp a(2^{n+1}+1))
\end{array}
\]

\noindent
But
\[
\begin{array}{l l l}
  c
& =
& (a(2^n+1) \comp \dots \comp a(2^{n+1}-1)) \Over
\\

& =
& a(2^n+1)
  \underbrace{\big( (a(2^n+2) \comp \dots \comp a(2^{n+1}-1)) \Over \big)}_{c'}
\end{array}
\]

\noindent
so that
\[
\begin{array}{l l}
& \unfold\ c\ ({\extract} \comp {a(2^{n+1})} \comp a(2^{n+1}+1))
\\

  =
& x(2^n+1) \Colon
  \unfold\ c'\
  \big(
  {\extract} \comp {a(2^{n+1})} \comp a(2^{n+1}+1)
   \comp {a(2(2^n+1))} \comp a(2(2^n+1)+1)
  \big)
\\

  =
& x(2^n+1) \Colon
  \unfold\ c'\
  \big(
  {\extract} \comp {a(2^{n+1})} \comp a(2^{n+1}+1)
   \comp {a(2^{n+1}+2)} \comp a(2^{n+1}+3)
  \big)
\end{array}
\]

\noindent
Hence,
$\extract\ (c_n \Over)$ is
\[
  x(2^n) \Colon
  x(2^n+1) \Colon
  \unfold\ c'\
  \big(
  {\extract} \comp {a(2^{n+1})} \comp a(2^{n+1}+1)
   \comp {a(2^{n+1}+2)} \comp a(2^{n+1}+3)
  \big)
\]

\noindent
By iteratively unfolding $c'$ as much as possible,
we obtain
\[
\begin{array}{r c l}
  \extract\ (c_n \Over)
& =
& x(2^n) \Colon \ldots \Colon x(2^{n+1}-1) \Colon
  \unfold\ \Over\ ({\extract} \comp {c_{n+1}})
\\

& =
& x(2^n) \Colon \ldots \Colon x(2^{n+1}-1) \Colon
  (\extract\ (c_{n+1} \Over))
\end{array}
\]
\end{proof}

Since $(\bft\ t)$ is $(\extract\ (c_0\ \Over))$,
we thus obtain from Lemma~\ref{lem:proof:sem:stream-tree}
that $(\bft t)(i) = x(i)$ for all $i \in \NN_{>0}$.

}
\opt{full}{
\section{Proofs of~\S\ref{sec:spectral} (\nameref{sec:spectral})}
\label{sec:proof:spectral}

We provide some additional
well-known material on spectral spaces (\S\ref{sec:spectral}).
Besides~\cite{dst19book,gg24book,goubault13book},
see also~\cite{lawson11mscs} on the more general \emph{stably compact} spaces.

\subparagraph*{Topological Spaces.}
A \emph{topological space} is a pair $(X,\Open(X))$ of a set $X$
and a collection $\Open(X)$ of subsets of $X$, called \emph{open sets}.
The collection $\Open(X)$
is called a \emph{topology} on $X$,
and is asked to be stable under arbitrary unions and under finite intersections.
In particular, $\emptyset$ and $X$ are open in $X$
(respectively as the empty union and the empty intersection).

A set $C \sle X$ is \emph{closed} if its complement $X \setminus C$ is open.
Closed sets are stable under finite unions and arbitrary
intersections.
Hence, any $\SP \sle X$ is contained in a least closed
set $\clos \SP  \sle X$. 
Each space $(X,\Open(X))$ is equipped with a \emph{specialisation} (pre)order
$\leq_{\Open(X)}$ on $X$, defined as
\[
\begin{array}{l l l}
  x \leq_{\Open(X)} y
& \text{iff}
& (\forall U \in \Open(X))
  \left(
  x \in U
  ~\longimp~
  y \in U
  \right)
\end{array}
\]

\noindent
Given $x \in X$, we have
$\clos{\{x\}} = \down x \deq \{ y \in X \mid y \leq_{\Open(X)} x\}$
(\cite[Lemma 4.2.7]{goubault13book}).
A topology $\Open(X)$ 
is \emph{$T_0$} when $\leq_{\Open(X)}$ is a partial order
(\cite[Proposition 4.2.3]{goubault13book}).

A subset $\SP$ of a $T_0$ space $(X,\Open(X))$ is \emph{saturated}
if $\SP$ is upward-closed w.r.t.\ the specialisation order $\leq_{\Open(X)}$,
or equivalently if
$\SP = \bigcap \{U \in \Open(X) \mid \SP \sle U \}$
(\cite[Proposition~4.2.9]{goubault13book}).

In the case of the Scott topology $\Open(X)$ on a dcpo $(X,\leq_X)$,
the specialisation (pre)order $\leq_{\Open(X)}$ 
coincides with $\leq_X$ (\cite[Lemma 1.2.3]{ac98book}).
Hence $\leq_{\Open(X)}$ is a partial order and $(X,\Open(X))$ is $T_0$.
Moreover, a set $\SP \sle X$ is saturated precisely when $\SP$
is upward-closed w.r.t.\ $\leq_X$.

\subparagraph*{Spectral Spaces.}

We take~\cite[Definition 1.1.5]{dst19book} as our official
definition of spectral space.

\begin{definition}
\label{def:proof:spectral}
A space $(X,\Open(X))$ is \emph{spectral} if all the following conditions
are satisfied.
\begin{enumerate}[(i)]
\item
\label{item:proof:spectral:T0}
$(X,\Open(X))$ is $T_0$.

\item
\label{item:proof:spectral:KO}
$\K\Open(X)$ is stable under finite intersections
and form a basis of $\Open(X)$.

\item
\label{item:proof:spectral:sober}
$(X,\Open(X))$ is sober.
\end{enumerate}
\end{definition}

Items (\ref{item:proof:spectral:T0})
and (\ref{item:proof:spectral:KO}) in Definition~\ref{def:proof:spectral}
were explained above and in \S\ref{sec:spectral}.
Note that item (\ref{item:proof:spectral:KO})
implies that $X \in \K\Open(X)$ (as the empty intersection).
Spectral spaces are therefore compact.

Regarding item (\ref{item:proof:spectral:sober}),
let us say that a closed set $F$ is \emph{irreducible}
if given a (possibly empty) finite set of closed sets $\mathcal{C}$,
if $F \sle \bigcup \mathcal{C}$,
then $F \sle C$ for some $C \in \mathcal{C}$.
Hence irreducible closed sets are non-empty.
It is clear that each $\clos{\{x\}}$ is irreducible.
A space $(X, \Open(X))$ is \emph{sober} if for each irreducible closed set $F$,
there is exactly one $x \in X$ such that $F = \clos{\{x\}}$.
Note that if $X$ is $T_0$,
then $\clos{\{x\}} = \clos{\{y\}}$ implies $x = y$.

It is well-known that if $X$ is an algebraic dcpo,
then the Scott topology on $X$ is sober.

\begin{lemma}
\label{lem:proof:spectral:sober}
If $X$ is an algebraic dcpo, then $(X,\Open(X))$ is a sober space.
\end{lemma}

\begin{proof}
Let $C$ be an irreducible closed set.
We have to provide some $x \in X$ such that $C = \clos{\{x\}}$.
Let $D$ be the set of all finite $d \in X$ such that
$d \in C$.

We show that $D$ is directed.
Since $C$ is non-empty, let $x \in C$.
Since $X$ is algebraic, there is some finite $d \leq x$,
and $d \in C$ since $C$ is Scott-closed.
Hence $D$ is non-empty.
Let now $d, d' \in D$.
Assume toward a contradiction that 
$\up d \cap \up d' \cap C = \emptyset$,
and consider the closed sets $F \deq X \setminus \up d$
and $F' \deq X \setminus \up d'$.
Then $C \sle F \cup F'$, but $C \not\sle F, F'$
since $d,d' \in C$.
Hence there is some $x \in \up d \cap \up d' \cap C$,
and $d$, $d'$ are two finite approximations of $x$.
But since $X$ is algebraic, the set of finite approximations of $x$
is directed.
Hence there is some finite $d''$ such that $d,d' \leq d'' \leq x$,
and we get $d'' \in C$ since $x \in C$ with $C$ Scott-closed.

Let now $x = \bigvee D$.
Since $D \sle C$ with $C$ Scott-closed, we have $x \in C$.

We now show that $C = \clos{\{x\}}$,
i.e.\ that for all $y \in X$, we have $y \in C$ if, and only if, $y \leq x$.
If $y \leq x$, then $y \in C$ since $x \in C$ with $C$ Scott-closed.
Assume now $y \in C$.
Given a finite $d \leq y$, we have $d \in C$ since $C$ is Scott-closed,
and thus $d \in D$ and $d \leq x$.
By algebraicity, it follows that $y \leq x$.
\end{proof}

\subparagraph*{Equivalent Definitions of Spectral Spaces.}
It follows from Lemma~\ref{lem:proof:spectral:sober}
that if $X$ is a Scott domain, then $(X, \Open(X))$
is spectral in the sense of Definition~\ref{def:proof:spectral}.
But the notion of spectral space we gave in~\S\ref{sec:spectral}
differs from Definition~\ref{def:proof:spectral},
since item~(\ref{item:proof:spectral:sober}) is replaced with the condition that $X$
is well-filtered (cf Proposition~\ref{prop:sem:wf}).
It is well-known that the two conditions are equivalent for $T_0$ spaces
which satisfy a weakening of condition~(\ref{item:proof:spectral:KO})
in Definition~\ref{def:proof:spectral}.
We present here a reformulation
of proofs in~\cite[\S 8.3]{goubault13book}.
See also~\cite[\S 6.1 and Exercise 6.3.6]{gg24book}.

A space $(X,\Open(X))$ is \emph{locally} compact
if for each $x \in X$ and each open $U \in \Open(X)$
such that $x \in U$,
there is an open $V \in \Open(X)$ and a compact $K$
such that
\[
  x \sle V \sle K \sle U
\]

It is clear that condition (\ref{item:proof:spectral:KO})
in Definition~\ref{def:proof:spectral}
implies local compactness.
The following is e.g.~\cite[Proposition~8.3.8]{goubault13book}.

\begin{proposition}
\label{prop:proof:spectral:wf-sober}
If $(X, \Open(X))$ is $T_0$, locally compact and well-filtered,
then $(X, \Open(X))$ is sober.
\end{proposition}

\begin{proof}
Let $C$ be an irreducible closed set.
We have to show that $C = \clos{\{x\}}$ for some $x \in X$.
Let $\K_C$ be the set of all compacts $K$ such that $U \sle K$
for some open $U$ with $U \cap C \neq \emptyset$.

We first show that $\K_C$ is codirected.
First, since $C$ is not empty, there is some $x \in C$.
By local compactness (applied to $x \in X$ with $X$ open),
there is an open $U$ and a compact $K$ such that $x \in U \sle K$.
But $x \in U \cap C$ implies $U \cap C \neq \emptyset$.
Hence $K \in \K_C$.
Let now $K, L \in \K_C$.
Hence, there are opens $U, V$ such that
$U \sle K$, $V \sle L$
and $U \cap C \neq \emptyset$
and $V \cap C \neq \emptyset$.
We have $(U \cap V) \cap C \neq \emptyset$,
since otherwise we would have $C \sle (X \setminus U) \cup (X \setminus V)$,
which is impossible since $C$ is irreducible
with $C \not\sle X \setminus U, X \setminus V$.
Hence, there is some $x \in (U \cap V) \cap C$.
By local compactness, there is some open $W$ and some compact $M$
such that
\[
  x \in W \sle M \sle U \cap V
\]

\noindent
Since $x \in C$, we have $W \cap C \neq \emptyset$,
and we are done with $M \in \K_C$ since $M \sle U \cap V \sle K, L$.

We are going to apply well-filteredness to a codirected
family of compact-saturated induced by $\K_C$.
Note that for each compact $K$, the set $\up K$ is also compact
(since opens are saturated, given $U \in \Open(X)$ we have
$K \sle U$ if, and only if $\up K \sle U$).

Let $\mathcal{Q} \sle \K\Sat(X)$ be the set of all $\up K$ for $K \in \K_C$.
Note that $\mathcal{Q}$ is codirected since so is $\K_C$.
Assume toward a contradiction that $\bigcap \mathcal{Q} \cap C = \emptyset$.
Hence $\bigcap \mathcal{Q} \sle X \setminus C$,
and by well-filteredness there is some $\up K \in \mathcal{Q}$
such that $\up K \sle X \setminus C$.
But then $K \cap C = \emptyset$,
which contradicts $K \in \K_C$.

Hence $\mathcal{Q} \cap C \neq \emptyset$.
Let $x \in \mathcal{Q} \cap C$.
We show that $C = \clos{\{x\}}$.
We have to show that for all $y \in X$,
we have $y \in C$ if, and only if, $y \leq_{\Open(X)} x$.
If $y \leq_{\Open(X)} x$ then
we must have $y \in C$ since $C$ 
otherwise $y$ belongs to the open $X \setminus C$ with $x \in C$.
Assume conversely that $y \in C$,
and let $U$ be an open such that $y \in U$.
By local compactness there is an open $V$ and a compact $K$ such that
\[
  y \in V \sle K \sle U
\]

\noindent
Since $y \in V \cap C$, we have $K \in \K_C$.
Hence $x \in \up K$ since $x \in \bigcap \mathcal{Q}$,
and $x \in U$ since $K \sle U$ implies $\up K \sle U$.
\end{proof}

\subparagraph*{The Hofmann-Mislove Theorem.}
Conversely, sober spaces are well-filtered.
This is a consequence of the famous Hofmann-Mislove (or Scott-open filter)
theorem.
The following is e.g.~\cite[Theorem 8.3.2]{goubault13book}
or~\cite[Theorem 7.2.9]{aj95chapter}.
A \emph{filter} in a poset is a codirected set which is also upward-closed.

\begin{theorem}[Hofmann-Mislove]
\label{thm:proof:spectral:hofmann-mislove}
Let $(X, \Open(X))$ be a sober space,
and let $\mathcal{U} \sle \Open(X)$ be a Scott-open filter.
If $\bigcap \mathcal{U} \sle V$ for some $V \in \Open(X)$,
then $V \in \mathcal{U}$.
\end{theorem}

\begin{proof}
Let $\mathcal{U} \sle \Open(X)$ be Scott-open (w.r.t.\ inclusion),
and let $V_0 \in \Open(X)$ such that $\bigcap \mathcal{U} \sle V_0$.
Assume toward a contradiction that $V_0 \notin \mathcal{U}$.

We shall apply Zorn's lemma to the poset
\[
  \left(
    \left\{ V \in \Open(X)
    \mid
    \bigcap \mathcal{U} \sle V
    ~\text{and}~
    V \notin \mathcal{U}
    \right\}
    ,\,
    \sle
  \right)
\]

\noindent
This poset is non-empty since $V_0$ belongs to it.
Moreover, it is stable by unions of non-empty chains:
if $\mathcal{V}$ is a non-empty chain of opens $V \sge \bigcap \mathcal{U}$,
then $\bigcup \mathcal{V} \in \mathcal{U}$ would imply
$V \in \mathcal{U}$ for some $V \in \mathcal{V}$ since $\mathcal{U}$
is Scott-open w.r.t.\ inclusion.

Hence by Zorn's lemma there is a maximal open $V$
such that $\bigcap \mathcal{U} \sle V$
and $V \notin \mathcal{U}$.
Let $C \deq X \setminus V$.

We claim that $C$ is irreducible.
First, $C \neq \emptyset$
since otherwise we would have $V = X$, contradicting $X \in \mathcal{U}$
($\mathcal{U}$ is a filter of opens).
Assume now $C \sle F \cup G$ with $F$, $G$ closed,
and consider the opens $A \deq X \setminus F$
and $B \deq X \setminus G$.
We have $A \cap B \sle V$.
Assume that the opens $V \cup A$ and $V \cup B$ both belong to $\mathcal{U}$.
Since $\mathcal{U}$ is a filter of opens,
we thus get
$(V \cup A) \cap (V \cup B) \in \mathcal{U}$.
But $(V \cup A) \cap (V \cup B) = V \cup (A \cap B)$,
while $V \cup (A \cap B) = V$ since $A \cap B \sle V$.
Hence $V \in \mathcal{U}$, a contradiction.
We thus have (say) $V \cup A \notin \mathcal{U}$,
which by maximality of $V$ implies $A \sle V$,
that is, $C \sle F$.

Since $X$ is sober, there is some $x \in X$
such that $C = \clos{\{x\}}$.
We of course have $x \notin V$.
Moreover, if $x \notin U$ with $U$ open,
then $U \cap C = \emptyset$
(since $y \in U \cap C$ would imply $y \leq_{\Open(X)} x$ and $x \in U$),
so that $U \sle V$
and $U \notin \mathcal{U}$ since $\mathcal{U}$ is a filter with $V \notin \mathcal{U}$.
Hence $x \in U$ for all $U \in \mathcal{U}$.
We thus have some $x \in \bigcap \mathcal{U}$ with $x \notin V$,
contradicting $\bigcap \mathcal{U} \sle V$.
\end{proof}

The Hofmann-Mislove Theorem~\ref{thm:proof:spectral:hofmann-mislove}
yields that sober spaces are well-filtered.
The following is e.g.~\cite[Proposition 8.3.5]{goubault13book}
or~\cite[Corollary 7.2.11]{aj95chapter}.

\begin{proposition}
\label{prop:proof:spectral:wf}
If $X$ is sober, then $X$ is well-filtered.
\end{proposition}

\begin{proof}
Let $\mathcal{Q}$ be a codirected set of compact-saturated sets,
and let $V \in \Open(X)$ such that $\bigcap \mathcal{Q} \sle V$.
Let $\mathcal{U}$ be the set of all opens $U$ such that
$Q \sle U$ for some $Q \in \mathcal{Q}$.

We claim that $\mathcal{U}$ is a Scott-open filter.
First, if $U \in \mathcal{U}$ and $U \sle U'$ with $U'$ open,
then we of course have $U' \in \mathcal{U}$.
Moreover, $\mathcal{U}$ is not empty since $X \in \mathcal{U}$ 
as $\mathcal{Q}$ is not empty.
Let now $U, U' \in \mathcal{U}$.
Hence there are $Q, Q' \in \mathcal{Q}$ such that $Q \sle U$ and $Q' \sle U'$.
Since $\mathcal{Q}$ is codirected,
there is $Q'' \in \mathcal{Q}$ such that $Q'' \sle Q$ and $Q'' \sle Q'$.
Hence $Q'' \sle U \cap U'$ and $U \cap U' \in \mathcal{U}$.
It follows that $\mathcal{U}$ is a filter.

It remains to show that $\mathcal{U}$ is Scott-open.
Let $\bigcup \mathcal{D} \in \mathcal{U}$
with $\mathcal{D} \sle \Open(X)$ directed.
Hence there is some $Q \in \mathcal{Q}$ such that
$Q \sle \bigcup \mathcal{D}$.
But since $Q$ is compact and $\mathcal{D}$ is a directed set of opens,
there is some $U \in \mathcal{D}$ such that $Q \sle U$.
Hence $U \in \mathcal{U}$.

For each $Q \in \mathcal{Q}$, since $Q$ is saturated
we have $Q = \bigcap \{U \in \Open(X) \mid Q \sle U\}$.
Hence $Q = \bigcap \{U \in \mathcal{U} \mid Q \sle U\}$
and $\bigcap \mathcal{Q} = \bigcap \mathcal{U}$.

We now apply the Hofmann-Mislove Theorem~\ref{thm:proof:spectral:hofmann-mislove}
to $\mathcal{U}$ and $V$,
to the effect that $V \in \mathcal{U}$.
Hence there is some $Q \in \mathcal{Q}$ such that $Q \sle V$,
as required.
\end{proof}

Combining Proposition~\ref{prop:proof:spectral:wf}
with Lemma~\ref{lem:proof:spectral:sober}
yields Proposition~\ref{prop:sem:wf}:

\PropWellFilt*

Moreover, Propositions~\ref{prop:proof:spectral:wf}
and~\ref{prop:proof:spectral:wf-sober}
imply that Definition~\ref{def:proof:spectral}
is equivalent to 
the definition of spectral space in \S\ref{sec:spectral}.
Hence, with Lemma~\ref{lem:proof:spectral:sober}
we have proved Theorem~\ref{thm:spectral}.

\ThmScottSpectral*

\subparagraph*{Intersections of Compact-Saturated Sets.}

Proposition~\ref{prop:proof:spectral:wf} also yields the following,
from which one easily obtains Corollary~\ref{cor:sem:ksatbigcap}.

\begin{lemma}
\label{lem:proof:spectral:ksat}
Let $X$ be a spectral space.
Given $\SP \sle X$,
we have $\SP \in \K\Sat(X)$
if, and only if,
$\SP = \bigcap \{ K \in \K\Open(X) \mid \SP \sle K\}$.
\end{lemma}

\begin{proof}
Fix $\SP \sle X$ with $X$ spectral.

First, note that
the set $\{ K \in \K\Open(X) \mid \SP \sle K\}$
is a codirected family of compact-saturated sets.
Indeed, since $X$ is compact we have $X \in \K\Open(X)$,
so that $\{ K \in \K\Open(X) \mid \SP \sle K\}$ is not empty.
Moreover, given $K, L \in \K\Open(X)$ such that $\SP \sle K,L$,
we have $\SP \sle K \cap L$ with $K \cap L \in \K\Open(X)$.


We first show that if $\SP \in \K\Sat(X)$,
then $\SP = \bigcap \{ K \in \K\Open(X) \mid \SP  \sle K\}$.
We have $\SP \sle \bigcap \{ K \in \K\Open(X) \mid \SP  \sle K\}$
since $\SP \sle X$ with $X \in \K\Open(X)$.
Moreover, we have seen above that if $\SP$ is saturated,
then
$\SP = \bigcap \{ U \in \Open(X) \mid \SP  \sle U\}$.
But $\K\Open(X)$ is a basis of the topology since $X$ is spectral.
Hence, if $\SP \sle X$ is compact and included in an open $U$,
then $\SP$ is included in a compact-open $K \sle U$.
It follows that
\[
\begin{array}{l l l}
  \bigcap \{ K \in \K\Open(X) \mid \SP  \sle K \}
& \sle
& \bigcap \{ U \in \Open(X) \mid \SP  \sle U \}
\end{array}
\]

\noindent
and we are done.

Well-filteredness
(Proposition~\ref{prop:proof:spectral:wf}) yields the converse,
namely that $\SP \in \K\Sat(X)$ whenever
$\SP = \bigcap \{ K \in \K\Open(X) \mid \SP \sle K\}$.
Indeed, if $\SP = \bigcap \{ K \in \K\Open(X) \mid \SP \sle K\}$,
then $\SP$ is saturated (as an intersection of saturated sets).
The set $\SP$ is moreover compact.
Let $\mathcal{D} \sle \Open(X)$ be a directed family of opens.
Since $\{ K \in \K\Open(X) \mid \SP \sle K\}$
is a codirected family of compact-saturated sets,
we get from Proposition~\ref{prop:proof:spectral:wf}
that if
$\SP = \bigcap \{ K \in \K\Open(X) \mid \SP \sle K\} \sle \bigcup \mathcal{D}$,
then $K \sle \bigcup \mathcal{D}$ for some $K \in \K\Open(X)$ such that $\SP \sle K$.
But since $K$ is compact, we have $K \sle U$ for some $U \in \mathcal{D}$.
Hence $\SP \sle U$.
\end{proof}

\begin{lemma}
\label{lem:proof:spectral:coh}
If $X$ is a spectral space, then $\K\Sat(X)$ is stable under finite
intersections.
\end{lemma}

\begin{proof}
We of course have $X \in \K\Sat(X)$.

Let $Q_1, Q_2 \in \K\Sat(X)$.
We have
$Q_i = \bigcap \left\{ K_i \in \K\Open(X) \mid Q_i \sle K_i \right\}$
by Lemma~\ref{lem:proof:spectral:ksat}.
Hence
\[
\begin{array}{l l l}
  Q_1 \cap Q_2
& =
& \bigcap \left\{
  K_1 \cap K_2
  \mid
  \text{$Q_1 \sle K_1$ and $Q_2 \sle K_2$ with $K_1,K_2 \in \K\Open(X)$}
  \right\}
\end{array}
\]

Since $\K\Open(X)$ is stable under binary intersections
(Definition~\ref{def:proof:spectral}(\ref{item:proof:spectral:KO})),
it follows that
$Q_1 \cap Q_2 = \bigcap \left\{ K \in \K\Open(X) \mid Q_1 \cap Q_2 \sle K \right\}$,
and we get $Q_1 \cap Q_2 \in \K\Sat(X)$
by Lemma~\ref{lem:proof:spectral:ksat}.
\end{proof}

Corollary~\ref{cor:sem:ksatbigcap} is given by Theorem~\ref{thm:spectral}
and the following.

\begin{corollary}
\label{cor:proof:spectral:ksatbigcap}
If $X$ is a spectral space, then $\K\Sat(X)$ is stable under all intersections.
\end{corollary}

An other important consequence of well-filteredness is
Lemma~\ref{lem:sem:log:realto} (\S\ref{sec:sem:log}).

}
\opt{full}{
\section{Proofs of \S\ref{sec:sem:log} (\nameref{sec:sem:log})}
\label{sec:proof:sem:log}

\subsection{Properties of Semantic Modalities}
\label{sec:proof:sem:log:mod}

We begin with Lemma~\ref{lem:sem:log:realto}.

\LemWfArrow*

\begin{proof}
For the first part of the statement,
since $f \colon \I\PT \to \I\PTbis$ is Scott-continuous,
we have that $f^{-1}(V)$ is open in $\I{\PT}$.
Since $\bigcap\mathcal{Q} \sle f^{-1}(V)$
with $\mathcal{Q}$ a codirected family of compact-saturated sets,
it follows from Proposition~\ref{prop:sem:wf} (well-filteredness)
that $Q \sle f^{-1}(V)$ for some $Q \in \mathcal{Q}$.
Hence $f \in Q \realto V$.

We now show that $Q \realto V$ is open in $\I\PT \to \I\PTbis$
when $V$ is open in $\I\PTbis$ and $Q$ is compact-saturated in $\I\PT$.
We have seen Lemma~\ref{lem:proof:spectral:ksat} (Appendix~\ref{sec:proof:spectral})
that $Q$ is the codirected intersection of the set of all compact-opens
$K$ such that $Q \sle K$.
It thus follows from the first part of the statement that
\[
\begin{array}{l l l}
  Q \realto V
& \sle
& \bigcup\left\{
  K \realto V
  \mid \text{$Q \sle K$ and $K$ compact open}
  \right\}
\end{array}
\]

\noindent
The converse inclusion follows from the fact that if $Q \sle K$,
then $K \realto V \sle Q \realto V$.

Consider now some $K \in \K\Open(\I\PT)$.
There are finite $d_1,\dots,d_n \in \I\PT$
such that $K = \up d_1 \cup \dots \cup \up d_n$.
Hence $K \realto V = \bigcap_{1 \leq i \leq n} (\up d_i \realto V)$.
Moreover, given $f \in \up d_i \realto V$,
we have $f(d_i) \in V$.
Hence there is some finite $e \in V$
such that $e \leq f(d_i)$.
We thus have $(d_i \step e) \leq f$,
where the step function $(d_i \step e)$ is finite in $\I\PT \to \I\PTbis$
(Proposition~\ref{prop:proof:scott:fin}, \S\ref{sec:proof:sem:pure:finite}).
It follows that 
\[
\begin{array}{l l l}
  K \realto V
& =
& \bigcap_{1 \leq i \leq n}
  \bigcup \left\{
  \up(d_i \step e)
  \mid
  \text{$e \in V$ finite}
  \right\}
\end{array}
\]

Hence $K \realto V$ is open as a finite intersection of unions
of basic opens $\up(d_i \step e)$,
and $Q \realto V$ is open as a union of opens.

We now show that $V \realto Q$ is compact-saturated in $\I\PTbis \to \I\PT$.
Since
\[
\begin{array}{r c l}
  V
& =
& \bigcup \{ \up e \mid \text{$e \in V$ finite} \}
\\

  Q
& =
& \bigcap \{ K \in \K\Open(\I\PT) \mid Q \sle K \}
\end{array}
\]

\noindent
we have
\[
\begin{array}{l l l}
  V \realto Q
& =
& \bigcap
  \left\{
  \up e \realto K
  \mid
  \text{$e \in V$ finite and $K \in \K\Open(\I\PT)$ with $Q \sle K$}
  \right\}
\end{array}
\]

Now, writing again $K \in \K\Open(\I\PT)$
as $K = \up d_1 \cup \dots \cup \up d_n$
with $d_1,\dots,d_n \in \I\PT$ finite,
we have that $\up e \realto K = \bigcup_{1 \leq i \leq n}\up(e \step d_i)$
is compact-open.
Hence $V \realto Q$ is compact-saturated as an intersection of compact-opens
(Corollary~\ref{cor:sem:ksatbigcap}).
\end{proof}

We now turn to simple but fundamental facts
on the semantic modalities $\I{\form{\modgen}}$.

\begin{lemma}
\label{lem:proof:sem:log:mod}
If $\modgen$ is either $\pi_1$, $\pi_2$ or $\fold$,
then $\I{\form{\modgen}}$ (is monotone and) preserves all unions and all intersections.

If $i \in\{1,2\}$, then $\I{\form{\inj_i}}$
(is monotone and) preserves all unions all \emph{non-empty} intersections.
\end{lemma}

\begin{proof}
In the case of $\modgen \in \{\pi_1, \pi_2, \fold\}$,
the semantic modality $\I{\form{\modgen}}$ acts by inverse image,
and thus preserves all unions and all intersections.

The semantic modalities $\I{\form{\inj_i}}$ act by \emph{direct} image,
and thus preserves all unions
(direct images are lower adjoints to inverse images).
It also preserves \emph{non-empty} intersections since 
the morphism $\I{\inj_i}$ is injective
(but $\I{\form{\inj_i}}$ 
does not preserve empty intersections since $\I{\inj_i}$
is not surjective).
\end{proof}

\begin{lemma}
\label{lem:proof:sem:log:mod:upset}
\hfill
\begin{enumerate}[(1)]
\item
Given $x \in \I{\PT[\rec \TV.\PT / \TV]}$,
we have
$\I{\form\fold}(\up x) = \up \I\fold(x)$.

\item
Given $x \in \I{\PT}$ and $y \in \I{\PTbis}$,
we have
$\I{\form{\pi_1}}(\up x) = \up (x, \bot)$
and
$\I{\form{\pi_2}}(\up y) = \up (\bot, y)$.

\item
Given $x \in \I{\PT_i}$,
we have
$\I{\form{\inj_i}}(\up x) = \up \I{\inj_i}(x)$.
\end{enumerate}
\end{lemma}

\begin{proof}
\hfill
\begin{enumerate}[(1)]
\item
Since $\I\fold$ is an iso with inverse $\I\unfold$, we have
\[
\begin{array}{l l l}
  \I{\form\fold}(\up x)
& =
& \left\{
  y \in \I{\rec \TV.\PT}
  \mid
  \I\unfold(y) \geq x
  \right\}
\\

& =
& \left\{
  y \in \I{\rec \TV.\PT}
  \mid
  y \geq \I\fold(x)
  \right\}
\end{array}
\]

\item
We only discuss the case of $x \in \I\PT$.
Since the order in $\I\PT \times \I\PTbis$ is pointwise,
we have
\[
\begin{array}{l l l}
  \I{\form{\pi_1}}(\up x)
& =
& \left\{
  p \in \I\PT \times \I\PTbis
  \mid
  \I{\pi_1}(p) \geq x
  \right\}
\\

& =
& \left\{
  (a,b) \in \I\PT \times \I\PTbis
  \mid
  a \geq x
  \right\}
\\

& =
& \up(x, \bot)
\end{array}
\]

\item
Since $\I{\inj_i}$ is an order embedding
(\S\ref{sec:proof:sem:pure:intdef} and \S\ref{sec:proof:sem:pure:sum}),
we have
\[
\begin{array}{l l l}
  \I{\form{\inj_i}}(\up x)
& =
& \left\{
  \I{\inj_i}(y)
  \mid
  y \geq x
  \right\}
\\

& =
& \left\{
  z
  \mid
  z \geq \I{\inj_i}(x)
  \right\}
\\

& =
& \up \I{\inj_i}(x)
\end{array}
\]
\qedhere
\end{enumerate}
\end{proof}

\begin{lemma}
\label{lem:proof:sem:log:mod:mod}
Given $\modgen \in \{\pi_1, \pi_2, \inj_1, \inj_2, \fold\}$,
the semantic modality $\I{\form\modgen}$
preserves compact-opens, opens and compact-saturated sets.
\end{lemma}

\begin{proof}
Let $\modgen \in \{\pi_1, \pi_2, \inj_1, \inj_2, \fold\}$.

Recall that (compact) opens are (finite) unions of sets of the form
$\up d$ with $d$ finite.
By combining Lemma~\ref{lem:proof:sem:log:mod:upset} with
Proposition~\ref{prop:proof:scott:fin},
we get that $\I{\form\modgen}(\up d)$ is of the form $\up e$
for some finite $e$.
Hence, Lemma~\ref{lem:proof:sem:log:mod} gives the result for (compact) opens.

Regarding compact saturated sets, recall from~\S\ref{sec:spectral}
that given a spectral space $X$,
each compact saturated set $Q \in \K\Sat(X)$
is the intersection of all compact opens $K \in \K\Open(X)$
such that $Q \sle K$.
But since spectral spaces are compact, we have $X \in \K\Open(X)$,
and $Q$ is always a \emph{non-empty} intersection of compact opens.
Hence the result by Lemma~\ref{lem:proof:sem:log:mod} (again)
and Corollary~\ref{cor:sem:ksatbigcap}.
\end{proof}

\subsection{Semantic Properties of Formulae}
\label{sec:proof:sem:log:form}

We begin with Proposition~\ref{prop:sem:log:degroot}.

\PropFormdeGroot*

\begin{proof}
We prove a more general statement.
Let $\val$ be a valuation of
$\FPEnv = \FP_1\colon \PTbis_1, \dots, \FP_n\colon \PTbis_n$.
\begin{enumerate}[(1)]
\item
If $\varphi \in \Lang^+(\FPEnv,\PT)$,
then the function $\I{\FP_1.\dots.\FP_n.\varphi}\val$
restricts to
\[
\begin{array}{l r c l}
  \I{\FP_1.\cdots.\FP_n.\varphi}\val \colon
& \Open(\I{\PTbis_1}) \times \dots \times \Open(\I{\PTbis_n})
& \longto
& \Open(\I{\PT})
\\

& (\SP_1,\dots,\SP_n)
& \longmapsto
& \I{\varphi}\val[\SP_1 / \FP_1, \dots, \SP_n / \FP_n]
\end{array}
\]

\item
If $\varphi \in \Lang^-(\FPEnv,\PT)$,
then the function $\I{\FP_1.\dots.\FP_n.\varphi}\val$
restricts to
\[
\begin{array}{l r c l}
  \I{\FP_1.\cdots.\FP_n.\varphi}\val \colon
& \K\Sat(\I{\PTbis_1}) \times \dots \times \K\Sat(\I{\PTbis_n})
& \longto
& \K\Sat(\I{\PT})
\\

& (\SP_1,\dots,\SP_n)
& \longmapsto
& \I{\varphi}\val[\SP_1 / \FP_1, \dots, \SP_n / \FP_n]
\end{array}
\]
\end{enumerate}

We reason by induction on $\varphi$.
\begin{itemize}
\item
The case of $\FP_i \in \Lang^\pm(\FPEnv;\PT)$ is trivial.

\item
The case of $\form{\pair{}} \in \Lang^\pm(\FPEnv;\one)$
follows from the description of finite elements
in Figure~\ref{fig:proof:sem:finelt}
(Proposition~\ref{prop:proof:scott:fin}, \S\ref{sec:proof:sem:pure:finite}).

\item
The cases of $\True$, $\varphi \land \psi$, $\False$ and $\varphi \lor \psi$
all follow from the fact that $\Open(\PT)$ and $\K\Sat(\PT)$
are closed under finite unions and finite intersections
(using Corollary~\ref{cor:sem:ksatbigcap}
for finite intersections of compact-saturated sets).

\item
The cases of $\form\modgen\varphi$ follow from Lemma~\ref{lem:proof:sem:log:mod:mod},
and that of the realizability implication $\realto$ 
follows from Lemma~\ref{lem:sem:log:realto}
(note that $\psi \realto \varphi$ has no free fixpoint variables).

\item
Finally, the case of $(\exists \itvar)\varphi \in \Lang^+(\FPEnv;\PT)$
is trivial (since $\Open(\I\PT)$ is stable under unions),
while that of $(\forall \itvar)\varphi \in \Lang^-(\FPEnv;\PT)$
follows from Corollary~\ref{cor:sem:ksatbigcap}
(closure of $\K\Sat(\I\PT)$ under all intersections).
\qedhere
\end{itemize}
\end{proof}

\subsubsection{Proofs of Lemma~\ref{lem:sem:log:posneg}
and Proposition~\ref{prop:sem:log:scott}}

We now discuss Lemma~\ref{lem:sem:log:posneg}
and Proposition~\ref{prop:sem:log:scott}.
Lemma~\ref{lem:sem:log:posneg} uses the following,
which is actually item (\ref{item:sem:log:scott:mon})
of Proposition~\ref{prop:sem:log:scott}.

\begin{lemma}
\label{lem:proof:sem:log:form:mon}
Let $\val$ be a valuation of
$\FPEnv = \FP_1\colon \PTbis_1, \dots, \FP_n\colon \PTbis_n$.
If $\varphi \in \Lang(\FPEnv;\PT)$,
then the function
\[
\begin{array}{l r c l}
  \I{\FP_1.\cdots.\FP_n.\varphi}\val \colon
& \Po(\I{\PTbis_1}) \times \dots \times \Po(\I{\PTbis_n})
& \longto
& \Po(\I{\PT})
\\

& (\SP_1,\dots,\SP_n)
& \longmapsto
& \I{\varphi}\val[\SP_1 / \FP_1, \dots, \SP_n / \FP_n]
\end{array}
\]

\noindent
is monotone.
\end{lemma}

\begin{proof}
We reason by induction on $\varphi \in \Lang(\FPEnv;\PT)$.
\begin{itemize}
\item
The case of $\FP_i \in \Lang^\pm(\FPEnv;\PT)$ is trivial.

\item
The case of $\form{\pair{}} \in \Lang(\FPEnv;\one)$ is trivial.

\item
The case of $\psi \realto \varphi$ is trivial since $\psi$ and $\varphi$
have no free fixpoint variable.

\item
The cases of $\form\modgen$ follow from Lemma~\ref{lem:proof:sem:log:mod}.

\item
The cases of $\True$, $\varphi \land \psi$, $\False$, $\varphi \lor \psi$,
$(\exists \itvar)\varphi$
and
$(\forall \itvar)\varphi$
all follow from the fact that unions and intersections are monotone
operations.

\item
In the cases of
$(\finmu^t \FPbis)\varphi$
and
$(\finnu^t \FPbis)\psi$,
let $m \deq \I{t}\val$
and note that
\[
\begin{array}{r c l !{\qquad\text{and}\qquad} r c l}
  \I{(\finmu^t \FPbis)\varphi}\val
& =
& (\I{\FPbis.\varphi}\val)^m(\emptyset)

& \I{(\finnu^t \FPbis)\psi}\val
& =
& (\I{\FPbis.\psi}\val)^m(\I\PT)
\end{array}
\]

The induction hypothesis yields that
the functions
$\I{\FP_1.\cdots.\FP_n.\FPbis.\varphi}\val$
and
$\I{\FP_1.\cdots.\FP_n.\FPbis.\psi}\val$
are monotone.
An induction on $i \leq m$ then shows that the functions

\[
\begin{array}{r c l}
  \Po(\I{\PTbis_1}) \times \dots \times \Po(\I{\PTbis_n})
& \longto
& \Po(\I{\PT})
\\

  (\SP_1,\dots,\SP_n)
& \longmapsto
& (\I{\FPbis.\varphi}\val[\SP_1/\FP_1,\dots,\SP_n/\FP_n])^i(\emptyset)
\\

  (\SP_1,\dots,\SP_n)
& \longmapsto
& (\I{\FPbis.\psi}\val[\SP_1/\FP_1,\dots,\SP_n/\FP_n])^i(\I\PT)

\end{array}
\]

\noindent
are both monotone.
\qedhere
\end{itemize}
\end{proof}

We can now prove Lemma~\ref{lem:sem:log:posneg}.

\LemPosNeg*

\begin{proof}
The proof is by mutual induction on the predicates
$\itvar \Pos \varphi$ and $\itvar \Neg \varphi$.
\begin{description}
\item[Cases of]
\[
\begin{array}{c}

\dfrac{\itvar \notin \FV(\varphi)}
  {\itvar \Pos \varphi}

\qquad

\dfrac{\itvar \notin \FV(\varphi)}
  {\itvar \Neg \varphi}

\end{array}
\]

Trivial.

\item[Cases of]
\[
\begin{array}{c}

\dfrac{\itvar \Pos \varphi}
  {\itvar \Pos \form\modgen \varphi}

\qquad

\dfrac{\itvar \Neg \varphi}
  {\itvar \Neg \form\modgen \varphi}

\end{array}
\]

By induction hypothesis and Lemma~\ref{lem:proof:sem:log:mod}.

\item[Cases of]
\[
\begin{array}{c}

\dfrac{\itvar \Neg \psi
  \qquad
  \itvar \Pos \varphi}
  {\itvar \Pos \psi \realto \varphi}

\qquad

\dfrac{\itvar \Pos \psi
  \qquad
  \itvar \Neg \varphi}
  {\itvar \Neg \psi \realto \varphi}

\end{array}
\]

By induction hypothesis, since $\realto$ is contravariant
in its first argument and covariant in its second argument.

\item[Cases of]
\[
\begin{array}{c}

\dfrac{\itvar \Pos \varphi, \psi}
  {\itvar \Pos \varphi \lor \psi}

\qquad

\dfrac{\itvar \Pos \varphi, \psi}
  {\itvar \Pos \varphi \land \psi}

\qquad

\dfrac{\itvar \Neg \varphi, \psi}
  {\itvar \Neg \varphi \lor \psi}

\qquad

\dfrac{\itvar \Neg \varphi, \psi}
  {\itvar \Neg \varphi \land \psi}

\\\\

\dfrac{\itvar \Pos \varphi}
  {\itvar \Pos (\exists \itvarbis)\varphi}

\qquad

\dfrac{\itvar \Neg \varphi}
  {\itvar \Neg (\forall \itvarbis)\varphi}

\end{array}
\]

By induction hypothesis,
since unions and intersections are monotone operations.

\item[Cases of]
\[
\begin{array}{c}
\dfrac{\itvar \Pos \varphi
  \qquad
  \itvar \notin \FV(t)}
  {\itvar \Pos (\finnu^t \FP)\varphi}

\qquad

\dfrac{\itvar \Neg \psi
  \qquad
  \itvar \notin \FV(t)}
  {\itvar \Neg (\finmu^t \FP)\psi}

\end{array}
\]

Let $n \leq m$ in $\NN$.
Note that $\I{t}\val[n / \itvar] = \I{t}\val[m / \itvar] = \I{t}\val$
since $\itvar \notin \FV(t)$.
Moreover,
$\I\varphi\val[n / \itvar] \sle \I\varphi\val[m / \itvar]$
and
$\I\psi\val[m/\itvar] \sle \I\psi\val[n/\itvar]$
by induction hypothesis.

Note that we have
\[
\begin{array}{r c l !{\qquad\text{and}\qquad} r c l}
  \I{(\finmu^t \FP)\varphi}\val
& =
& (\I{\FP.\varphi}\val)^{\I{t}\val}(\emptyset)

& \I{(\finnu^t \FP)\psi}\val
& =
& (\I{\FP.\psi}\val)^{\I{t}\val}(\I\PT)

\end{array}
\]

By monotonicity of $\I{\FP.\varphi}\val$ and $\I{\FP.\psi}\val$
(Lemma~\ref{lem:proof:sem:log:form:mon}),
an induction on $i \in \NN$ yields
\[
\begin{array}{l !{\quad} r c l}
& (\I{\FP.\varphi}\val[n / \itvar])^i(\emptyset)
& \sle
& (\I{\FP.\varphi}\val[m / \itvar])^i(\emptyset)
\\

  \text{and}
& (\I{\FP.\psi}\val[m / \itvar])^i(\I\PT)
& \sle
& (\I{\FP.\psi}\val[n / \itvar])^i(\I\PT)
\end{array}
\]

Hence we are done with
\[
\begin{array}{l !{\quad} r c l}
& (\I{\FP.\varphi}\val[n / \itvar])^{\I{t}}(\emptyset)
& \sle
& (\I{\FP.\varphi}\val[m / \itvar])^{\I{t}}(\emptyset)
\\

  \text{and}
& (\I{\FP.\psi}\val[m / \itvar])^{\I{t}}(\I\PT)
& \sle
& (\I{\FP.\psi}\val[n / \itvar])^{\I{t}}(\I\PT)
\end{array}
\]

\item[Cases of]
\[
\begin{array}{c}
\dfrac{\itvar \Pos \varphi}
  {\itvar \Pos (\finmu^t \FP)\varphi}

\qquad

\dfrac{\itvar \Neg \psi}
  {\itvar \Neg (\finnu^t \FP)\psi}

\end{array}
\]

Let $n \leq m$ in $\NN$.
We have $\I{t}\val[n / \itvar] \leq \I{t}\val[m / \itvar]$.
Moreover,
$\I\varphi\val[n / \itvar] \sle \I\varphi\val[m / \itvar]$
and
$\I\psi\val[m/\itvar] \sle \I\psi\val[n/\itvar]$
by induction hypothesis.

Note that we have
\[
\begin{array}{r c l !{\qquad\text{and}\qquad} r c l}
  \I{(\finmu^t \FP)\varphi}\val
& =
& (\I{\FP.\varphi}\val)^{\I{t}\val}(\emptyset)

& \I{(\finnu^t \FP)\psi}\val
& =
& (\I{\FP.\psi}\val)^{\I{t}\val}(\I\PT)

\end{array}
\]

By monotonicity of $\I{\FP.\varphi}\val$ and $\I{\FP.\psi}\val$
(Lemma~\ref{lem:proof:sem:log:form:mon}),
an induction on $i \in \NN$ yields
\[
\begin{array}{l !{\quad} r c l}
& (\I{\FP.\varphi}\val[n / \itvar])^i(\emptyset)
& \sle
& (\I{\FP.\varphi}\val[m / \itvar])^i(\emptyset)
\\

  \text{and}
& (\I{\FP.\psi}\val[m / \itvar])^i(\I\PT)
& \sle
& (\I{\FP.\psi}\val[n / \itvar])^i(\I\PT)
\end{array}
\]

\noindent
so that
\[
\begin{array}{l !{\quad} r c l}
& (\I{\FP.\varphi}\val[n / \itvar])^{\I{t}\val[n / \itvar]}(\emptyset)
& \sle
& (\I{\FP.\varphi}\val[m / \itvar])^{\I{t}\val[n / \itvar]}(\emptyset)
\\

  \text{and}
& (\I{\FP.\psi}\val[m / \itvar])^{\I{t}\val[n / \itvar]}(\I\PT)
& \sle
& (\I{\FP.\psi}\val[n / \itvar])^{\I{t}\val[n / \itvar]}(\I\PT)
\end{array}
\]

Again using the monotonicity of
$\I{\FP.\varphi}\val$ and $\I{\FP.\psi}\val$
(Lemma~\ref{lem:proof:sem:log:form:mon}),
we get
\[
\begin{array}{l !{\quad} r c l}
& (\I{\FP.\varphi}\val[m / \itvar])^{\I{t}\val[n / \itvar]}(\emptyset)
& \sle
& (\I{\FP.\varphi}\val[m / \itvar])^{\I{t}\val[m / \itvar]}(\emptyset)
\\

  \text{and}
& (\I{\FP.\psi}\val[m / \itvar])^{\I{t}\val[m / \itvar]}(\I\PT)
& \sle
& (\I{\FP.\psi}\val[m / \itvar])^{\I{t}\val[n / \itvar]}(\I\PT)
\end{array}
\]

\noindent
and we are done.
\qedhere
\end{description}
\end{proof}

We finally turn to Proposition~\ref{prop:sem:log:scott}.
This requires some preparatory material on Scott (co)continuity.
Given complete lattices $A$ and $B$,
a function $f \colon A \to B$ is \emph{Scott-cocontinuous}
if it preserves infs of codirected sets.
We use the well-known fact that Scott-(co)continuity
in several arguments is equivalent to Scott-(co)continuity
in each argument separately.

\begin{lemma}
\label{lem:proof:sem:log:scott:argumentwise}
Let $A$, $B$, $C$ be complete lattices,
and let $f \colon A \times B \to C$ be a function.

Then $f$ is Scott-(co)continuous if, and only if,
for all $a \in A$ and all $b \in B$, the functions
$y \mapsto f(a,y)$ and $x \mapsto f(x, b)$
are Scott-(co)continuous.
\end{lemma}

\begin{proof}
For Scott-continuity, this is e.g.~\cite[Proposition 1.4.3]{ac98book}.
Regarding Scott-cocontinuity, since $A$, $B$ and $C$ are assumed to be
\emph{complete} lattices, one sees that $f$ is Scott-cocontinuous
precisely when $f$ is Scott-continuous as a function $A^\op \times B^\op \to C^\op$.
Hence the result.
\end{proof}

Given a function $f \colon A \times B \to C$
the function $(x,y) \mapsto f(x,-)^n(y)$ is defined as
\[
\begin{array}{l l l !{\qquad\text{and}\qquad} l l l}
  f(x,-)^0(y)
& \deq
& y

& f(x,-)^{n+1}(y)
& \deq
& f(x,\, f(x,-)^n(y))
\end{array}
\]

\begin{lemma}
\label{lem:proof:sem:log:scott:iter}
Let $A$ and $B$ be complete lattices,
and let $f \colon A \times B \to B$ be a function.

If $f$ is Scott-(co)continuous,
then so is the function $(x,y) \mapsto f(x,-)^n(y)$
for each $n \in \NN$.
\end{lemma}

\begin{proof}
The proof is by induction on $n \in \NN$.
In the base case $n=0$,
the function $(x,y) \mapsto f(x,-)^n(y)$ is the second projection.
This function is Scott-continuous and cocontinuous by
Lemma~\ref{lem:proof:sem:log:scott:argumentwise}.

Assume now that $(x,y) \mapsto f(x,-)^n(y)$ is Scott-(co)continuous.
The function 
$(x,y) \mapsto f(x,-)^{n+1}(y)$ is the composite
\[
\begin{tikzcd}[column sep=huge]
  A \times B
  \arrow{r}{\pair{\pi_1,\id}}
& A \times (A \times B)
  \arrow{rr}{\id \times (x,y) \mapsto f(x,-)^n(y)}
&
& A \times B
  \arrow{r}{f}
& B
\end{tikzcd}
\]

\noindent
which is Scott-continuous thanks to Lemma~\ref{lem:proof:sem:log:scott:argumentwise}.
\end{proof}

\begin{lemma}
\label{lem:proof:sem:log:scott:veewedge}
Let $A$ and $B$ be complete lattices,
let ${(f_i)}_{i \in I}$ be a family of functions $f_i \colon A \to B$.
\begin{enumerate}[(1)]
\item
\label{item:proof:sem:log:scott:veewedge:contvee}
If each $f_i$ is Scott-continuous,
then
$x \mapsto \bigvee_i f_i(x)$
is Scott-continuous.

\item
\label{item:proof:sem:log:scott:veewedge:contwedge}
Assume that sups distribute over finite infs in $B$.
If each $f_i$ is Scott-continuous and $I$ is finite,
then
$x \mapsto \bigwedge_i f_i(x)$
is Scott-continuous.

\item
\label{item:proof:sem:log:scott:veewedge:cocontwedge}
If each $f_i$ is Scott-cocontinuous,
then
$x \mapsto \bigwedge_i f_i(x)$
is Scott-cocontinuous.

\item
\label{item:proof:sem:log:scott:veewedge:cocontvee}
Assume that infs distribute over finite sups in $B$.
If each $f_i$ is Scott-cocontinuous and $I$ is finite,
then
$x \mapsto \bigvee_i f_i(x)$
is Scott-cocontinuous.
\end{enumerate}
\end{lemma}

\begin{proof}
By duality (see the proof of Lemma~\ref{lem:proof:sem:log:scott:argumentwise}),
it is sufficient to consider the case of Scott-continuity.
\begin{enumerate}[(1)]
\item
Given a directed $D \sle A$, since sups commute over sups, we have
\[
\begin{array}{r c l}
  \bigvee_{i \in I} f_i(\bigvee D)
& =
& \bigvee_{i \in I} \bigvee_{d \in D} f_i(d)
\\

& =
& \bigvee_{d \in D} \bigvee_{i \in I} f_i(d)
\end{array}
\]

\noindent
and we are done.

\item
Let $D \sle A$ be directed.
Since the $f_i$'s are Scott continuous, we have
\[
\begin{array}{l l l}
  \bigwedge_{i \in I} f_i(\bigvee D)
& =
& \bigwedge_{i \in I} \bigvee_{d \in D} f_i(d)
\end{array}
\]

Now, we have
\(
\bigwedge_{i \in I} f_i(d)
\leq
\bigwedge_{i \in I} \bigvee_{d \in D} f_i(d)
\)
for each $d \in D$, 
and thus
\[
\begin{array}{l l l}
  \bigvee_{d \in D} \bigwedge_{i \in I} f_i(d)
& \leq
& \bigwedge_{i \in I} \bigvee_{d \in D} f_i(d)
\end{array}
\]

For the converse inequality, since $I$ is finite
and sups distribute over finite infs in $B$, we have
\[
\begin{array}{l l l}
  \bigwedge_{i \in I} \bigvee_{d \in D} f_i(d)
& \leq
& \bigvee_{g \in D^I} \bigwedge_{i \in I} f_i(g(i))
\end{array}
\]

Let $g \in D^I$.
Since $I$ is finite and $D$ is directed,
there is some $d \in D$ such that
$g(i) \leq d$ for all $i \in I$.
By monotonicity of the $f_i$'s we have
$f_i(g(i)) \leq f_i(d)$ for all $i \in I$,
and thus
$\bigwedge_{i \in I} f_i(g(i)) \leq \bigwedge_{i \in I} f_i(d)$.
It follows that
\[
\begin{array}{l l l}
  \bigwedge_{i \in I} f_i(g(i))
& \leq
& \bigvee_{d \in D} \bigwedge_{i \in I} f_i(d)
\end{array}
\]

\noindent
Since this holds for each $g \in D^I$, we obtain
\[
\begin{array}{l l l}
  \bigvee_{g \in D^I} \bigwedge_{i \in I} f_i(g(i))
& \leq
& \bigvee_{d \in D} \bigwedge_{i \in I} f_i(d)
\end{array}
\]
\qedhere
\end{enumerate}
\end{proof}

We can finally prove Proposition~\ref{prop:sem:log:scott}.

\PropFormScott*

\begin{proof}
Item (\ref{item:sem:log:scott:mon})
was proven in Lemma~\ref{lem:proof:sem:log:form:mon}.
For items~(\ref{item:sem:log:scott:pos}) and~(\ref{item:sem:log:scott:neg}),
we prove a more general statement.
Let $\val$ be a valuation of
$\FPEnv = \FP_1\colon \PTbis_1, \dots, \FP_n\colon \PTbis_n$,
and consider the function
\[
\begin{array}{l r c l}
  \I{\FP_1.\cdots.\FP_n.\varphi}\val \colon
& \Po(\I{\PTbis_1}) \times \dots \times \Po(\I{\PTbis_n})
& \longto
& \Po(\I{\PT})
\\

& (\SP_1,\dots,\SP_n)
& \longmapsto
& \I{\varphi}\val[\SP_1 / \FP_1, \dots, \SP_n / \FP_n]
\end{array}
\]

\noindent
We prove the following:
\begin{itemize}
\item
if $\varphi \in \Lang^+(\FPEnv;\PT)$,
then 
$\I{\FP_1.\cdots.\FP_n.\varphi}\val$
is Scott-continuous;

\item
if $\varphi \in \Lang^-(\FPEnv;\PT)$,
then 
$\I{\FP_1.\cdots.\FP_n.\varphi}\val$
is Scott-cocontinuous.
\end{itemize}

\noindent
We reason by induction on $\varphi$.
\begin{itemize}
\item
The case of $\FP_i \in \Lang^\pm(\FPEnv;\PT)$ is trivial.

\item
The cases of $\form{\pair{}}$
and $\psi \realto \varphi$ are trivial
(recall that $\psi \realto \varphi$ has no free fixpoint variable).

\item
The cases of $\form\modgen$ follow from Lemma~\ref{lem:proof:sem:log:mod}
(note that codirected sets are required to be non-empty).

\item
The cases of $\True$, $\varphi \land \psi$, $\False$ and $\varphi \lor \psi$,
all follow from Lemma~\ref{lem:proof:sem:log:scott:veewedge}.

\item
The cases of
$(\exists \itvar)\varphi$
and
$(\forall \itvar)\psi$,
where $\varphi$ is positive and $\psi$ is negative,
follow from items~(\ref{item:proof:sem:log:scott:veewedge:contvee})
and~(\ref{item:proof:sem:log:scott:veewedge:cocontwedge})
in
Lemma~\ref{lem:proof:sem:log:scott:veewedge},
respectively.

\item
In the cases of
$(\finmu^t \FPbis)\varphi$
and
$(\finnu^t \FPbis)\psi$,
let $m \deq \I{t}\val$
and note that
\[
\begin{array}{r c l !{\qquad\text{and}\qquad} r c l}
  \I{(\finmu^t \FPbis)\varphi}\val
& =
& (\I{\FPbis.\varphi}\val)^m(\emptyset)

& \I{(\finnu^t \FPbis)\psi}\val
& =
& (\I{\FPbis.\psi}\val)^m(\I\PT)
\end{array}
\]

We can then conclude using Lemma~\ref{lem:proof:sem:log:scott:iter}.
\qedhere
\end{itemize}
\end{proof}

\subsubsection{Notes on Example~\ref{ex:sem:log:fixpoints}}

We provide some details on Example~\ref{ex:sem:log:fixpoints}.

\subparagraph*{Streams vs $\omega$-Words.}
The set of $\omega$-words on a set $A$ is $A^\omega$.
The set $A^\omega$ can be seen as a subset of the Scott domain $(A_\bot)^\omega$.
The induced subspace topology is the usual \emph{product topology}
on $A^\omega$ (see e.g.~\cite[Example 8]{rs24jfla}).
This topology is Hausdorff and has a basis of clopens.
It is moreover compact if, and only if, the set $A$ is finite
(while the Scott topology on $(A_\bot)^\omega$ is always compact).

\subparagraph*{Boolean Spaces.}
We claimed that Stone (Boolean) spaces are exactly
those spectral spaces for which the compact-saturated sets
coincide with the closed sets.
Boolean spaces can be defined as those compact Hausdorff spaces
which have a basis of clopen sets~\cite[Definition 1.3.1]{dst19book}.
Boolean spaces are of course spectral
(see e.g.~\cite[Theorem 1.3.4]{dst19book}).

In presence of a basis of compact-opens,
if compact-saturated sets coincide with closed sets,
we get that each element of the basis is in fact clopen.
Together with the $T_0$ separation axiom, this yields that
the space is Hausdorff with a basis of clopens.
Hence, a space is Boolean provided it is
a spectral space in which the compact-saturated
sets coincide with the closed sets.

Conversely, in a Boolean space each closed set is an
intersection of clopens (since each open is a union of clopens).
But in a compact Hausdorff space, clopens coincide with
compact-opens.
Hence in a Boolean space
closed sets coincide with intersections of compact-opens,
and thus with compact-saturated sets
(Lemma~\ref{lem:proof:spectral:ksat}
and Corollary~\ref{cor:proof:spectral:ksatbigcap},
\S\ref{sec:proof:spectral}).

\subparagraph*{Well-Filteredness and the Formula in Equation~\eqref{eq:form:spec}.}
Finally, we made an informal claim corresponding to the following.

\begin{lemma}
\label{lem:proof:ex:sem:log:fixpoints}
Consider a Scott-continuous
$f \colon \I{\Tree \Bool} \to \I{\Stream \Nat}$.
Given $n \in \NN$, the following are equivalent:
\begin{enumerate}[(i)]
\item
\label{item:proof:ex:sem:log:fixpoints:spec}
$f$ satisfies the formula in Equation~\eqref{eq:form:spec}
of Example~\ref{ex:form:fix}.

\item
\label{item:proof:ex:sem:log:fixpoints:unif}
There is some $m \in \NN$ such that $f$ satisfies the formula
\begin{multline*}
  \forall\Box\form\lbl(\form\true \lor \form\false)
  ~\land~
  \exists\Next \bigwedge_{i \leq m} (\exists\Next)^i \form\lbl\form\true
  ~\land~
  \exists\Next \bigwedge_{i \leq m} (\exists\Next)^i \form\lbl\form\false
  \quad\realto
\\
  \bigvee_{i \leq m}\Next^i \form\hd \form{\geq \term n} \form\tot
\end{multline*}
\end{enumerate}
\end{lemma}

\begin{proof}
We show that
\(
  \text{(\ref{item:proof:ex:sem:log:fixpoints:spec})}
  \imp
  \text{(\ref{item:proof:ex:sem:log:fixpoints:unif})}
\).
First, note that
for a negative $\psi$
we have
\[
\begin{array}{l l l}
  \exists\Next \exists\Box \psi
& \thesisiff
& (\forall \itvarbis_1)
  (\forall \itvarbis_2)
  \left(
  \form{\lft}(\finnu^{\itvarbis_1} \FP)(\psi \land \exists\Next\FP)
  ~\lor~
  \form{\rght}(\finnu^{\itvarbis_2} \FP)(\psi \land \exists\Next\FP)
  \right)
\end{array}
\]

We shall now apply Lemma~\ref{lem:sem:log:realto}.
It is actually more convenient to apply it in the form of
Lemma~\ref{lem:proof:sem:log:sound:wfcc}
(used in the proof of the Soundness Theorem~\ref{thm:sem:log:sound}).
Hence the formula in Equation~\eqref{eq:form:spec}
implies the following formula
\begin{multline*}
  (\exists \itvarbis_1, \itvarbis_2, \itvarbis_3, \itvarbis_4, \itvar) 
  \bigg(
  \Big(
  \forall\Box\form\lbl(\form\true \lor \form\false)
\\
  ~\land~
  \left(
  \form{\lft}(\finnu^{\itvarbis_1} \FP)(\form\lbl\form\true \land \exists\Next\FP)
  ~\lor~
  \form{\rght}(\finnu^{\itvarbis_2} \FP)(\form\lbl\form\true \land \exists\Next\FP)
  \right)
\\
  ~\land~
  \left(
  \form{\lft}(\finnu^{\itvarbis_3} \FP)(\form\lbl\form\false \land \exists\Next\FP)
  ~\lor~
  \form{\rght}(\finnu^{\itvarbis_4} \FP)(\form\lbl\form\false \land \exists\Next\FP)
  \right)
  \Big)
\\
  \realto\quad
  (\finmu^\itvar \FP)(\form\hd \form{\geq \term n} \form\tot \lor \Next \FP)
  \bigg)
\end{multline*}

But now, by monotonicity
(Lemma~\ref{lem:sem:log:posneg})
we get that if $f$ satisfies the last formula above,
then it also satisfies
\begin{multline*}
  (\exists \itvar) 
  \bigg(
  \Big(
  \forall\Box\form\lbl(\form\true \lor \form\false)
\\
  ~\land~
  \left(
  \form{\lft}(\finnu^{\itvar} \FP)(\form\lbl\form\true \land \exists\Next\FP)
  ~\lor~
  \form{\rght}(\finnu^{\itvar} \FP)(\form\lbl\form\true \land \exists\Next\FP)
  \right)
\\
  ~\land~
  \left(
  \form{\lft}(\finnu^{\itvar} \FP)(\form\lbl\form\false \land \exists\Next\FP)
  ~\lor~
  \form{\rght}(\finnu^{\itvar} \FP)(\form\lbl\form\false \land \exists\Next\FP)
  \right)
  \Big)
\\
  \realto\quad
  (\finmu^\itvar \FP)(\form\hd \form{\geq \term n} \form\tot \lor \Next \FP)
  \bigg)
\end{multline*}

\noindent
and thus
\begin{multline*}
  (\exists \itvar) 
  \bigg(
  \Big(
  \forall\Box\form\lbl(\form\true \lor \form\false)
  ~\land~
  \exists\Next(\finnu^{\itvar} \FP)(\form\lbl\form\true \land \exists\Next\FP)
  ~\land~
  \exists\Next(\finnu^{\itvar} \FP)(\form\lbl\form\false \land \exists\Next\FP)
  \Big)
\\
  \realto\quad
  (\finmu^\itvar \FP)(\form\hd \form{\geq \term n} \form\tot \lor \Next \FP)
  \bigg)
\end{multline*}

\noindent
Hence the result.
\end{proof}

\subsubsection{Proof of Proposition~\ref{prop:sem:char:fin}}

We recall the statement of Proposition~\ref{prop:sem:char:fin}.

\PropTopCharFin*

The proof of Proposition~\ref{prop:sem:char:fin}
is split into Lemma~\ref{lem:proof:sem:fin:coprime-open}
and Lemma~\ref{lem:proof:sem:coprime-open:fin}.

\begin{lemma}
\label{lem:proof:sem:fin:coprime-open}
Given $\varphi \in \Lang^\land(\PT)$, if $\I\varphi \neq \emptyset$ then
$\I\varphi = \up d$ for some finite $d \in \I\PT$.
\end{lemma}

\begin{proof}
The proof is by induction on $\varphi \in \Lang^\land(\PT)$.
\begin{description}
\item[Case of $\True$.]
In this case, we have $\I\varphi = \up \bot$.

\item[Case of $\varphi \land \psi$.]
First, note that $\I\varphi, \I\psi$ are non-empty since so is their intersection.
By induction hypothesis, there are finite $d,e \in \I\PT$
such that $\I\varphi = \up d$ and $\I\psi = \up e$.
Since $\up d \cap \up e$ is non-empty, and since $\I\PT$ is a Scott domain,
we get that $d \vee e$ is defined, finite, and such that
$\up(d\vee e) = \up d \cap \up e$.
Hence $\I{\varphi \land \psi} = \up(d \vee e)$.

\item[Case of $\form{\pair{}}$.]
Since $\I{\form{\pair{}}} = \up\top$.

\item[Case of $\form{\modgen}\varphi$
with $\modgen \in \{\pi_1, \pi_2, \inj_1, \inj_2, \fold\}$.]
Note that $\I\varphi$ is non-empty since so is
$\I{\form\modgen\varphi} = \I{\form\modgen}(\I\varphi)$.
Hence by induction hypothesis, there is some finite $d$
such that $\I\varphi = \up d$.
We can then conclude using Lemma~\ref{lem:proof:sem:log:mod:upset}
and Proposition~\ref{prop:proof:scott:fin}.

\item[Case of $\psi \realto \varphi$.]
First, if $\I\psi = \emptyset$,
then $\I{\psi \realto \varphi} = \up \bot$.

Assume now that $\I\psi \neq \emptyset$.
In this case, we must also have $\I\varphi \neq \emptyset$.
Hence by induction hypothesis there are $d,e$ finite
such that
$\up e = \I\psi$ and $\up d = \I\varphi$.
Then we are done since
\[
\begin{array}{l l l}
  \I{\psi \realto \varphi}
& =
& \left\{ f \mid \forall x \in \I\psi,~ f(x) \in \I\varphi\right\}
\\

& =
& \left\{ f \mid \forall x \geq e,~ f(x) \geq d\right\}
\\

& =
& \up \left(e \step d \right)
\end{array}
\]
\qedhere
\end{description}
\end{proof}

\begin{lemma}
\label{lem:proof:sem:coprime-open:fin}
If $d \in \I\PT$ is finite,
then there is
$\varphi \in \Lang^\land(\PT)$
such that $\up d = \I\varphi$.
\end{lemma}

\begin{proof}
We rely on Proposition~\ref{prop:proof:scott:fin}
and on the inductive definition of $\Fin(\I\PT)$ in Figure~\ref{fig:proof:sem:finelt}.
We reason by cases on the rules 
in Figure~\ref{fig:proof:sem:finelt}.
\begin{description}
\item[Case of]
\[
\dfrac{}
  {\bot \in \Fin(\I\PT)}  
\]

\noindent
Since $\up \bot = \I{\True}$.

\item[Case of]
\[
\dfrac{}
  {\top \in \Fin(\I\Unit)}
\]

\noindent
Since $\up \top = \I{\form{\pair{}}}$.

\item[Case of]
\[
\dfrac{d \in \Fin(\I\PT)
  \qquad
  e \in \Fin(\I\PTbis)}
  {(d,e) \in \Fin(\I{\PT \times \PTbis})}
\]

\noindent
By induction hypothesis, we
have $\varphi \in \Lang^\land(\PT)$
and $\psi \in \Lang^\land(\PTbis)$
such that
$\I\varphi = \up d$
and
$\I\psi = \up e$.
Since the order in $\I{\PT \times \PTbis}$
is pointwise, we get
\[
\begin{array}{l l l}
  \up(d,e)
& =
& \up d \times \up e
\\

& =
& (\up d \times \I\PTbis)
  \cap
  (\I\PT \times \up e)
\\

& =
& \I{\form{\pi_1}\varphi \land \form{\pi_2}\psi}
\end{array}
\]

\item[Case of]
\[
\dfrac{d \in \Fin(\I\PT_i)}
  {\I{\inj_i}(d) \in \Fin(\I{\PT_1 + \PT_2})}
\]

By induction hypothesis, there is
$\varphi \in \Lang^\land(\PT_i)$
such that $\I\varphi = \up d$.
Since $\I{\inj_i}$ is an order embedding
(\S\ref{sec:proof:sem:pure:intdef} and \S\ref{sec:proof:sem:pure:sum}),
\[
\begin{array}{l l l}
  \up \I{\inj_i}(d)
& =
& \left\{
  \inj_i(x) \mid \inj_i(x) \geq \inj_i(d)
  \right\}
\\

& =
& \left\{
  \inj_i(x) \mid x \geq d
  \right\}
\\

& =
& \I{\form{\inj_i}\varphi}
\end{array}
\]

\item[Case of]
\[
\dfrac{d \in \Fin(\I{\PT[\rec\TV.\PT/\TV]})}
  {\I\fold(d) \in \Fin(\I{\rec\TV.\PT})}
\]

\noindent
By induction hypothesis, there is
$\varphi \in \Lang^\land(\PT[\rec\TV.\PT/\TV])$
such that $\I\varphi = \up d$.
We thus have $\up \I\fold(d) = \I{\form\fold \varphi}$.

\item[Case of]
\[
\dfrac{\begin{array}{l}
  \text{for all $i \in I$,~
  $d_i \in \Fin(\I{\PT})$
  ~and~
  $e_i \in \Fin(\I{\PTbis})$ \@;}
  \\
  \text{for all $J \sle I$,~
  $\bigvee_{j \in J} d_j$ defined in $\I{\PT}$
  ~$\imp$~
  $\bigvee_{j \in J} e_j$ defined in $\I{\PTbis}$}
  \end{array}}
  {\bigvee_{i \in I}(d_i \step e_i) \in \Fin(\I{\PT \arrow \PTbis})}
\]

\noindent
where $I$ is a finite set.

By induction hypothesis, for each $i \in I$
there are $\varphi_i \in \Lang^\land(\PTbis)$
and $\psi_i \in \Lang^\land(\PT)$
such that $\up e_i = \I{\varphi_i}$
and $\up d_i = \I{\psi_i}$.
Note that
\[
\begin{array}{l l l}
  \up(d_i \step e_i)
& =
& \left\{
  f \colon \I\PT \to \I\PTbis \mid
  \forall x \geq d_i,~
  f(x) \geq e_i
  \right\}
\\

& =
& \left\{
  f \colon \I\PT \to \I\PTbis \mid
  \forall x \in \I{\psi_i},~
  f(x) \in \I{\varphi_i}
  \right\}
\\

& =
& \I{\psi_i \realto \varphi_i}
\end{array}
\]

\noindent
The result then follows from the fact that
\[
\begin{array}{l l l}
  \up \left(
  \bigvee_{i \in I} d_i \step e_i
  \right)
& =
& \bigcap_{i \in I} \up(d_i \step e_i)
\end{array}
\]
\qedhere
\end{description}
\end{proof}

\subsection{Soundness of Deduction}
\label{sec:proof:sem:log:sound}

This Appendix~\ref{sec:proof:sem:log:sound}
is devoted to the proof of the Soundness Theorem~\ref{thm:sem:log:sound}.
We handle some rules separately.

\begin{lemma}
\label{lem:proof:sem:log:sound:mergequant}
The following deduction rules are sound.
\[
\begin{array}{c}

\ax{\exists M}
\dfrac{}
  {(\exists \itvar)(\exists \itvarbis)\varphi
  \,\thesis\,
  (\exists \itvar)\varphi[\itvar / \itvarbis]}

\qquad

\ax{\forall M}
\dfrac{}
  {(\forall \itvar)\psi[\itvar / \itvarbis]
  \,\thesis\,
  (\forall \itvar)(\forall \itvarbis)\psi}

\end{array}
\]
\end{lemma}

\begin{proof}
Note that $\itvar, \itvarbis \Pos \varphi$,
while $\itvar, \itvarbis \Neg \psi$.
By monotonicity (Lemma~\ref{lem:sem:log:posneg}),
if $x \in \I\varphi\val[m / \itvar, n / \itvarbis]$
for some $m, n \in \NN$,
then $x \in \I\varphi\val[\max(m,n) / \itvar, \max(m,n) / \itvarbis]$.
Hence the rule $\ax{\exists M}$.

As for the rule $\ax{\forall M}$,
by \emph{anti}monotonicity
(Lemma~\ref{lem:sem:log:posneg} again),
if
$x \notin \I\psi\val[m / \itvar, n / \itvarbis]$
for some $m, n \in \NN$,
then 
$x \notin \I\psi\val[\max(m,n) / \itvar, \max(m,n) / \itvarbis]$.
\end{proof}

\begin{lemma}
\label{lem:proof:sem:log:sound:fixquant}
The following deduction rules are sound,
where $\fingen \in \{\finmu, \finnu\}$.
\[
\begin{array}{c}

\ax{\fingen/\exists}
\dfrac{\itvar \notin \FV(t)}
  {(\fingen^t \FP)(\exists \itvar)\varphi
  \,\thesis\,
  (\exists \itvar)(\fingen^t \FP)\varphi}

\qquad

\ax{\forall/\fingen}
\dfrac{\itvarbis \notin \FV(t)}
  {(\forall \itvarbis)(\fingen^t \FP)\psi
  \,\thesis\,
  (\fingen^t \FP)(\forall \itvarbis)\psi}

\end{array}
\]
\end{lemma}

\begin{proof}
Note that $\varphi$ is positive with $\itvar \Pos \varphi$,
while $\psi$ is negative with $\itvarbis \Neg \psi$.
Hence $\ax{\forall/\fingen}$ is dual to $\ax{\fingen/\exists}$.
We only detail $\ax{\fingen/\exists}$.
The proof relies on the following.

\begin{claim}
Let $A$ be a complete lattice and consider a family
of Scott-continuous functions $(f_i \colon A \to A \mid i \in \NN)$
such that $f_i \leq f_j$ whenever $i \leq j$.
Let $f \colon A \to A$ take $a$ to $\bigvee_{i \in \NN}f_i(a)$.
Then for each $n \in \NN$, we have $f^n(a) = \bigvee_{i \in \NN} f^n_i(a)$.
\end{claim}

\begin{claimproof}
%
The proof of the Claim is by induction on $n \in \NN$.
In the base case $n = 0$,
we have $f^0(a) = a$
while $\bigvee_{i\in \NN} f_i^0(a) = \bigvee_{i \in \NN} a = a$.

For the induction step, by induction hypothesis we have
$f^{n+1}(a) = f(\bigvee_{i \in \NN} f_i^n(a))$.
Since the $f_i$'s are in particular monotone,
we get $f^n_i \leq f^n_j$ whenever $i \leq j$.
It follows that the set $\{f^n_i(a) \mid i \in \NN \}$ is directed,
so that
$f^{n+1}(a) = \bigvee_{i \in \NN} f(f^n_i(a))$
by Scott-continuity of $f$
(Lemma~\ref{lem:proof:sem:log:scott:veewedge}%
(\ref{item:proof:sem:log:scott:veewedge:contvee})).
By definition of $f$, we thus get
$f^{n+1}(a) = \bigvee_{i \in \NN} \bigvee_{j \in \NN} f_j(f_i^n(a))$.
It remains to show that
\[
\begin{array}{l l l}
  \bigvee_{i \in \NN} f_i^{n+1}(a)
& =
& \bigvee_{i \in \NN}\bigvee_{j \in \NN} f_j(f_i^n(a))
\end{array}
\]

The inequality $\leq$ is obvious
($f_i^{n+1}(a) = f_i(f_i^{n}(a))$).
For the converse inequality, let $i, j \in \NN$ and let $m = \max\{i,j\}$.
We have $f_i, f_j \leq f_m$, and thus $f_j(f_i^n(a)) \leq f_m(f_m^n(a))$.
Hence $f_j(f_i^n(a)) \leq f_m^{n+1}(a)$, which gives the result.
\end{claimproof}

Returning to the proof of the lemma,
let $f_i \deq \I{\FP.\varphi}\val[i / \itvar]$
and let $n \deq \I{t}\val$ (recall that $\itvar \notin \FV(t$).
The $f_i$'s are Scott continuous by
Proposition~\ref{prop:sem:log:scott}(\ref{item:sem:log:scott:pos}).
Moreover, $i \leq j$ implies $f_i \leq f_j$
by Lemma~\ref{lem:sem:log:posneg}.
Hence the Claim gives the result since
$f = \I{\FP.(\exists \itvar)\varphi}\val$,
so that
\[
\begin{array}{r c l c l}
  \I{(\finmu^t \FP)(\exists \itvar)\varphi}\val
& =
& (\I{\FP. (\exists \itvar)\varphi}\val)^n(\I\False)
& =
& f^n(\I\False)
\\

  \I{(\finnu^t \FP)(\exists \itvar)\varphi}\val
& =
& (\I{\FP. (\exists \itvar)\varphi}\val)^n(\I\True)
& =
& f^n(\I\True)

\end{array}
\]

\noindent
while
\[
\begin{array}{r c l}
  \I{(\exists \itvar)(\finmu^t \FP)\varphi}\val
& =
& \bigcup_{i \in \NN} f_i^n(\I\False)
\\

  \I{(\exists \itvar)(\finnu^t \FP)\varphi}\val
& =
& \bigcup_{i \in \NN} f_i^n(\I\True)

\end{array}
\]
\end{proof}

\begin{lemma}
\label{lem:proof:sem:log:sound:wfcc}
The following deduction rules are sound.
\[
\begin{array}{c}

\ax{WF}
\dfrac{\itvar \notin \FV(\varphi)}
  {(\forall \itvar)\psi \realto \varphi
  \,\thesis\,
  (\exists \itvar)(\psi \realto \varphi)}

\qquad

\ax{CC}
\dfrac{\psi \in \Lang^\pm(\PTbis)
  \qquad
  \itvar \notin \FV(\psi)}
  {\psi \realto (\exists \itvar)\varphi
  \,\thesis\,
  (\exists \itvar)(\psi \realto \varphi)}

\end{array}
\]
\end{lemma}

\begin{proof}
We begin with $\ax{WF}$.
Let $f \in \I{(\forall \itvar)\psi \realto \varphi}\val$.

Note that $(\forall \itvar)\psi$ is negative.
Hence $\itvar \Neg \psi$,
and Lemma~\ref{lem:sem:log:posneg}
yields that the function
$n \in \NN \mapsto \I\psi\val[n / \itvar]$ is antimonotone.
It then follows from Proposition~\ref{prop:sem:log:degroot}
that $\{\I\psi\val[n / \itvar] \mid n \in \NN\}$
is a codirected family of compact saturated sets.

Moreover, for each $n \in \NN$,
$\I\psi\val[n / \itvar]$ is the codirected
intersection of all compact opens $K$
such that $\I\psi\val[n / \itvar] \sle K$
(see Lemma~\ref{lem:proof:spectral:ksat} and its proof, \S\ref{sec:proof:spectral}).
Let $\mathcal{Q}$ be the set of all compact opens $K$
such that $\I\psi\val[n / \itvar] \sle K$ for some $n \in \NN$.
Hence $\mathcal{Q}$ is a codirected collection of compact opens
with $\bigcap \mathcal{Q} = \I{(\forall \itvar)\psi}\val$.

On the other hand, $\varphi$ is positive.
Hence $\I\varphi\val$ is open (Proposition~\ref{prop:sem:log:degroot} again),
and we get from Lemma~\ref{lem:sem:log:realto}
that $f \in K \realto \I\varphi\val$
for some $K \in \mathcal{Q}$.
But $K \sge \I\psi\val[n / \itvar]$ for some $n \in \NN$,
so that $f \in \I{\psi \realto \varphi}\val[n / \itvar]$
by contravariance of $\realto$ in its first argument
(recall that $\itvar$ is not free in $\varphi$).

We now turn to $\ax{CC}$.
Let $f \in \I{\psi \realto (\exists \itvar)\varphi}$.
The formula $\varphi$ is positive.
Hence $\itvar \Pos \varphi$
and by Lemma~\ref{lem:sem:log:posneg}
and Proposition~\ref{prop:sem:log:degroot}
we get that $\I{(\exists \itvar)\varphi}\val$
is the directed union of the opens 
$(\I{(\exists \itvar)\varphi}\val[n / \itvar] \mid n \in \NN)$.

On the other hand, again by Proposition~\ref{prop:sem:log:degroot},
we have that $\I\psi\val$ is compact (open).
Since $f$ is continuous, 
we get that that $\I\psi\val$ is a compact set included in the union
of the directed family of opens
$(f^{-1}(\I{(\exists \itvar)\varphi}\val[n / \itvar]) \mid n \in \NN)$.
Hence
$\I\psi \sle f^{-1}(\I{(\exists \itvar)\varphi}\val[n / \itvar])$
for some $n \in \NN$,
and
$f \in \I{\psi \realto \varphi}\val[n / \itvar]$
(recall that $\itvar$ is not free in $\varphi$).
\end{proof}

\begin{lemma}
\label{lem:proof:sem:log:sound:realtolor}
The following rule is sound.
\[
\ax{{\realto}/\lor}
\dfrac{\delta \in \Lang^\land(\PT)}
  {\delta \realto (\varphi_1 \lor \varphi_2)
  \,\thesis\,
  (\delta \realto \varphi_1) \lor (\delta \realto \varphi_2)}
\]
\end{lemma}

\begin{proof}
If $\I\delta = \emptyset$,
then $\I{\delta \realto \psi}\val = \I{\True}$ for all $\psi$,
hence the result.
Otherwise,
by Proposition~\ref{prop:sem:char:fin}
there is some finite $d \in \I\PT$ such that $\I\delta = \up d$.
Hence, 
if $f \in \I{\delta \realto (\varphi_1 \lor \varphi_2)}\val$,
then in particular $f(d) \in \I{\varphi_1}\val \cup \I{\varphi_2}\val$.
But since $f$ is monotone and
$\I{\varphi_i}\val$ is saturated (Proposition~\ref{prop:sem:log:degroot}),
we have
$f \in \I{\delta \realto \varphi_i}\val$
whenever $f(d) \in \I{\varphi_i}\val$.
\end{proof}

We now prove the Soundness Theorem~\ref{thm:sem:log:sound}.

\ThmLogSound*

\begin{proof}
We reason by mutual induction on $\C(\delta)$ (Figure~\ref{fig:consist})
and $\psi \thesis \varphi$ (Figure~\ref{fig:ded}).
We discuss each (sub)figure separately.
\begin{description}
\item[Propositional rules in Figure~\ref{fig:ded:prop}.]
These rules are all fairly standard, and we omit the proof.

\item[Modal rules in Figure~\ref{fig:ded:mod}.]
The rules $\ax{I}$ and the rule for $\realto$ were already
discussed in~\S\ref{sec:sem:log}.
The remaining rules all follow from Lemma~\ref{lem:proof:sem:log:mod}.

\item[Quantifier rules in Figure~\ref{fig:ded:quant}.]
The rules $\ax{\exists M}$ and $\ax{\forall M}$
are handled by Lemma~\ref{lem:proof:sem:log:sound:mergequant},
while the rules $(\fingen / \exists)$ and $(\forall / \fingen)$
are handled by Lemma~\ref{lem:proof:sem:log:sound:fixquant}.
The remaining rules are standard.

\item[Rules for bounded iteration in Figure~\ref{fig:ded:iter}.]
We discuss the rules
\[
\begin{array}{l l l}

\dfrac{\psi \,\thesis\, \varphi}
  {(\finmu^t \FP)\psi \,\thesis\, (\finmu^t \FP)\varphi}

\qquad

\dfrac{\psi \,\thesis\, \varphi}
  {(\finnu^t \FP)\psi \,\thesis\, (\finnu^t \FP)\varphi}
 
\end{array}
\]

\noindent
By induction hypothesis we have
$\I\psi\val \sle \I\varphi\val$.
Since moreover the functions
$\I{\FP.\psi}\val$ and $\I{\FP.\varphi}\val$
are monotone
(Proposition~\ref{prop:sem:log:scott}(\ref{item:sem:log:scott:mon})),
for each (appropriate) $\SP$ an induction on $n \in \NN$ yields
\[
\begin{array}{l l l}
  (\I{\FP.\psi}\val)^n(\SP)
& \sle
& (\I{\FP.\varphi}\val)^n(\SP)
\end{array}
\]

\noindent
Hence the result.

The other rules all directly follow from the interpretation
of bounded iteration in Figure~\ref{fig:sem:form}.

\item[Rules for the realizability implication in Figure~\ref{fig:ded:realto}.]
The rules $\ax{WF}$ and $\ax{CC}$ are handled by
Lemma~\ref{lem:proof:sem:log:sound:wfcc},
while the rule $\ax{{\realto}/\lor}$ is dealt-with
in Lemma~\ref{lem:proof:sem:log:sound:realtolor}.
All the remaining rules are standard, excepted $\ax{C}$.
\[
\ax{C}
\dfrac{\C(\delta)
  \qquad
  \delta \in \Lang^\land}
  {(\delta \realto \False) \,\thesis\, \False}
\]

By induction hypothesis we have that $\I\delta \neq \emptyset$.
Hence $\I{\delta \realto \False}\val$ is empty.

\item[Rules for the consistency predicate $\C$ in Figure~\ref{fig:consist}.]
We only detail a couple of rules, the other one being trivial.
\begin{description}
\item[Case of]
\[
\dfrac{\C(\varphi) 
  \qquad
  \C(\psi)}
  {\C(\form{\pi_1}\varphi \land \form{\pi_2}\psi)}
\]

\noindent
By induction hypothesis we have $\I\varphi \neq \emptyset$
and $\I\psi \neq \emptyset$.
Hence $\I{\form{\pi_1}\varphi \land \form{\pi_2}\psi} \neq \emptyset$
since
\[
\begin{array}{l l l}
  \I{\form{\pi_1}\varphi \land \form{\pi_2}\psi}
& =
& \left\{ z \mid \pi_1(z) \in \I\varphi \right\}
  \cap
  \left\{ z \mid \pi_2(z) \in \I\psi \right\}
\\

& =
& \left\{ (x,y) \mid \text{$x \in \I\varphi$ and $y \in \I\psi$} \right\}
\end{array}
\]

\item[Case of]
\[
\dfrac{\begin{array}{l}
  \text{$I$ finite and $\forall i \in I$,}~
  \C(\psi_i) 
  ~\text{and}~
  \C(\varphi_i) ;
  \\
  \text{$\forall J \sle I$,}~
  \bigwedge_{j \in J} \psi_j \thesis \False
  ~~\text{or}~~
  \C\left( \bigwedge_{j \in J} \varphi_j \right)
  \end{array}}
  {\C\left( \bigwedge_{i \in I}(\psi_i \realto \varphi_i) \right)}
\]

\noindent
Let $\PT,\PTbis$ such that
$\varphi_i \in \Lang^\land(\PT)$ and $\psi_i \in \Lang^\land(\PTbis)$
for all $i \in I$.

First, by induction hypothesis 
we have
$\I{\psi_i} \neq \emptyset$
and
$\I{\varphi_i} \neq \emptyset$
for all $i \in I$.
Hence, it follows from Proposition~\ref{prop:sem:char:fin}
that for each $i \in I$,
there are finite $d_i \in \I\PT$
and $e_i \in \I\PTbis$
such that
$\up d_i = \I{\varphi_i}$
and
$\up e_i = \I{\psi_i}$.
We thus have
$\I{\psi_i \realto \varphi_i} = \up(e_i \step d_i)$
for each $i \in I$,
so that
\[
\begin{array}{l l l}
  \I{\bigwedge_{i \in I} \left( \psi_i \realto \varphi_i \right)}
& =
& \bigcap_{i \in I} \up(e_i \step d_i)
\end{array}
\]

Assume 
$\I{\bigwedge_{i \in I} \left( \psi_i \realto \varphi_i \right)} = \emptyset$.
As recalled in \S\ref{sec:proof:sem:pure:finite}
(see also \cite[Theorem 1.4.12]{ac98book}),
there is some $J \sle I$
such that
\[
\begin{array}{l l l !{\quad\text{and}\quad} l l l}
  \bigcap_{i \in I} \up e_i
& \neq
& \emptyset

& \bigcap_{i \in I} \up d_i
& =
& \emptyset
\end{array}
\]

But the induction hypothesis
yields either $\bigcap_{i \in I} \up e_i = \emptyset$
or $\bigcap_{i \in I} \up d_i \neq \emptyset$,
a contradiction.
\qedhere
\end{description}
\end{description}
\end{proof}

\subsection{Completeness for the Neutral Fragment}
\label{sec:proof:sem:log:compl}

This Appendix~\ref{sec:proof:sem:log:compl}
is devoted to the proof of Theorem~\ref{thm:sem:log:compl}.

\ThmLogCompl*

The proof uses the following preliminary lemma.
The \emph{size} of a formula $\delta$ is denoted $|\delta|$.

\begin{lemma}
\label{lem:proof:sem:log:compl:nf}
Let $\delta \in \Lang^\land$ such that neither $\delta \thesis \False$
nor $\True \thesis \delta$ is derivable.
\begin{enumerate}[(1)]
\item
\label{item:proof:sem:log:compl:nf:unit}
If $\delta \in \Lang^\land(\Unit)$,
then $\delta \thesisiff \form{\pair{}}$.

\item
\label{item:proof:sem:log:compl:nf:rec}
If $\delta \in \Lang^\land(\rec \TV.\PT)$,
then
$\delta \thesisiff \form{\fold}\psi$
with $|\form{\fold}\psi| \leq |\delta|$.

\item
\label{item:proof:sem:log:compl:nf:plus}
If $\delta \in \Lang^\land(\PT_1 + \PT_2)$,
then there is an $i \in \{1,2\}$
such that
$\delta \thesisiff \form{\inj_i}\psi$
with $|\form{\inj_i}\psi| \leq |\delta|$.

\item
\label{item:proof:sem:log:compl:nf:times}
If $\delta \in \Lang^\land(\PT_1 \times \PT_2)$,
then
$\delta \thesisiff \form{\pi_1}\psi_1 \land \form{\pi_2}\psi_2$
with
\(
  |\form{\pi_1}\psi_1 \land \form{\pi_2}\psi_2| \leq |\delta|
\).

\item
\label{item:proof:sem:log:compl:nf:arrow}
If $\delta \in \Lang^\land(\PTbis \arrow \PT)$,
then there is a finite non-empty set $J$ such that
\[
\begin{array}{l l l}
  \delta
& \thesisiff
& \bigwedge_{j \in J} \left(\varphi_j \realto \psi_j\right)
\end{array}
\]

\noindent
with
\(
  |\bigwedge_{j \in J} \left(\varphi_j \realto \psi_j\right)|
  \leq
  |\delta|
\).
\end{enumerate}
\end{lemma}

\begin{proof}
The proof is a mutual induction on the construction of $\delta \in \Lang^\land$
with the (appropriate) rules in Figure~\ref{fig:form}.
The cases of the rules in Figure~\ref{fig:form:mod}
and of the rule
\[
\dfrac{\varphi_1 \in \Lang^\land(\PT_1)
  \qquad
  \varphi_2 \in \Lang^\land(\PT_2)}
  {\varphi_1 \realto \varphi_2 \in \Lang^\land(\PT_1 \arrow \PT_2)}
\]

\noindent
are trivial.
The only possible rule of
Figure~\ref{fig:form:prop}
is
\[
\dfrac{\varphi_1 \in \Lang^\land
  \qquad
  \varphi_2 \in \Lang^\land}
  {\varphi_1 \land \varphi_2 \in \Lang^\land}
\]

\noindent
We handle the different items separately.
In each case, if $\True \thesis \varphi_i$,
then $\delta \thesisiff \varphi_{3-i}$,
and the result follows from the induction hypothesis.
Also, we can assume $\varphi_i \not\thesis \False$
since $\delta \not\thesis \False$.
\begin{description}
\item[Item~(\ref{item:proof:sem:log:compl:nf:unit}).]
By induction hypothesis
we have $\varphi_i \thesisiff \form{\pair{}}$,
and the result follows.

\item[Item~(\ref{item:proof:sem:log:compl:nf:rec}).]
By induction hypothesis,
there are $\psi_i$ smaller than $\varphi_i$
such that
$\varphi_i \thesisiff \form{\fold}\psi_i$.
We have
$\delta \thesisiff \form{\fold}(\psi_1 \land \psi_2)$.
Moreover, since $|\varphi_i| \geq |\psi_i|+1$,
we have 
\[
\begin{array}{*{9}{l}}
  |\delta|
& \geq
& 1 + |\varphi_1| + |\varphi_2|

& \geq
& 3 + |\psi_1| + |\psi_2|

& \geq
& 2 + |\psi_1 \land \psi_2|

& \geq
& 1 + |\form{\fold}(\psi_1 \land \psi_2)|
\end{array}
\]

\item[Item~(\ref{item:proof:sem:log:compl:nf:plus}).]
By induction hypothesis,
there are $j_i \in \{1,2\}$
and $\psi_i$ smaller than $\varphi_i$
such that
$\varphi_i \thesisiff \form{\inj_{j_i}}\psi_i$.
Since $\delta \not\thesis \False$
we must have $j_1 = j_2$,
and
we can conclude as for item~(\ref{item:proof:sem:log:compl:nf:rec}).

\item[Item~(\ref{item:proof:sem:log:compl:nf:times}).]
By induction hypothesis,
we have
$\varphi_i \thesisiff \form{\pi_1}\psi_{i,1} \land \form{\pi_2}\psi_{i,2}$
with
$|\form{\pi_1}\psi_{i,1} \land \form{\pi_2}\psi_{i,2}| \leq |\varphi_i|$.
Note that
\[
\begin{array}{l l l}
  |\form{\pi_1}(\psi_{1,1} \land \psi_{2,1})
  \land \form{\pi_2}(\psi_{1,2} \land \psi_{2,2})|
& \leq
& 
  |(\form{\pi_1}\psi_{1,1} \land \form{\pi_2}\psi_{1,2})
   \land
   (\form{\pi_1}\psi_{2,1} \land \form{\pi_2}\psi_{2,2})|
\end{array}
\]

Hence we are done
with $\delta \thesisiff \form{\pi_1}\psi_1 \land \form{\pi_2}\psi_2$
where $\psi_i$ is $\psi_{1,i} \land \psi_{2,i}$.

\item[Item~(\ref{item:proof:sem:log:compl:nf:arrow}).]
By induction hypothesis.
\qedhere
\end{description}
\end{proof}

Theorem~\ref{thm:sem:log:compl}(\ref{item:sem:log:compl:consist})
is a consequence of the following.

\begin{lemma}
\label{lem:proof:sem:log:compl:c-false}
Let $\delta \in \Lang^\land(\PT)$.
\begin{enumerate}[(1)]
\item
\label{item:proof:sem:log:compl:c-false:c}
If $\I\delta \neq \emptyset$,
then $\C(\delta)$ is derivable.

\item
\label{item:proof:sem:log:compl:c-false:ded}
If $\I\delta = \emptyset$,
then $\delta \thesis \False$ is derivable.
\end{enumerate}
\end{lemma}

\begin{proof}
Both statements are proved by a simultaneous induction on the 
size of $\delta \in \Lang^\land$,
using Lemma~\ref{lem:proof:sem:log:compl:nf}.

Note that we may always assume $\True \not\thesis \delta$
and $\delta \not\thesis \False$.
Indeed, if $\True \thesis \delta$ then
item~(\ref{item:proof:sem:log:compl:c-false:c}) is trivial,
while 
item~(\ref{item:proof:sem:log:compl:c-false:ded})
is impossible by the Soundness Theorem~\ref{thm:sem:log:sound}.
Similarly, if $\delta \thesis \False$,
then
item~(\ref{item:proof:sem:log:compl:c-false:c})
is impossible by the Soundness Theorem~\ref{thm:sem:log:sound},
while 
item~(\ref{item:proof:sem:log:compl:c-false:ded})
is trivial.

We reason by cases on $\PT$.
\begin{description}
\item[Case of $\Unit$.]
By Lemma~\ref{lem:proof:sem:log:compl:nf},
we have $\delta \thesisiff \form{\pair{}}$.
Hence $\I\delta = \{\top\}$,
only item~(\ref{item:proof:sem:log:compl:c-false:c})
is possible
and we get $\C(\delta)$ since $\C(\form{\pair{}})$ and $\form{\pair{}} \thesis \delta$.

\item[Case of $\rec\TV.\PT$.]
By Lemma~\ref{lem:proof:sem:log:compl:nf},
we have $\delta \thesisiff \form\fold\psi$
with $|\psi| < |\delta|$.

We first consider item~(\ref{item:proof:sem:log:compl:c-false:c}).
If $\I\delta \neq \emptyset$,
then $\I{\psi} \neq \emptyset$.
Since
$\psi$ is smaller than $\delta$,
the induction hypothesis
yields $\C(\psi)$.
Hence $\C(\form\fold \psi)$
and we get $\C(\delta)$ since $\form\fold \psi \thesis \delta$.

We now turn to item~(\ref{item:proof:sem:log:compl:c-false:ded}).
If $\I\delta = \emptyset$ then
$\I{\psi} = \emptyset$.
Since $\psi$ is smaller than $\delta$,
the induction hypothesis yields
$\psi \thesis_{\PT[\rec\TV.\PT/\TV]} \False$,
so that
$\form\fold\psi \thesis_{\rec\TV.\PT} \form\fold \False$.
Then we are done since $\form\fold \False \thesis \False$.

\item[Case of $\PT_1 + \PT_2$.]
Similar.

\item[Case of $\PT_1 \times \PT_2$.]
By Lemma~\ref{lem:proof:sem:log:compl:nf},
we have
$\delta \thesisiff \form{\pi_1}\psi_1 \land \form{\pi_2}\psi_2$
with
\(
  |\form{\pi_1}\psi_1 \land \form{\pi_2}\psi_2| \leq |\delta|
\).

We first consider item~(\ref{item:proof:sem:log:compl:c-false:c}).
If $\I\delta \neq \emptyset$,
then $\I{\psi_1} \neq \emptyset$
and $\I{\psi_2} \neq \emptyset$.
Since $|\psi_1|, |\psi_2| < |\delta|$,
the induction hypothesis yields
$\C(\psi_1)$ and $\C(\psi_2)$.
Hence
$\C(\form{\pi_1}\psi_1 \land \form{\pi_2}\psi_2)$
and we are done.

We now turn to item~(\ref{item:proof:sem:log:compl:c-false:ded}).
If $\I\delta = \emptyset$ then we must have (say)
$\I{\psi_i} = \emptyset$.
Hence $\psi_i \thesis \False$ by induction hypothesis.
It follows that
$\form{\pi_i}\psi_i \thesis \form{\pi_i}\False$,
and we are done since
$\form{\pi_i}\False \thesis \False$.

\item[Case of $\PTbis \arrow \PT$.]
This is the most important and most difficult case.

By Lemma~\ref{lem:proof:sem:log:compl:nf},
there is a finite set $J$ such that
\[
\begin{array}{l l l}
  \delta
& \thesisiff
& \bigwedge_{j \in J} \left(\varphi_j \realto \psi_j\right)
\end{array}
\]

\noindent
with
\(
  |\bigwedge_{j \in J} \left(\varphi_j \realto \psi_j\right)|
  \leq
  |\delta|
\).

Assume first that for some $j \in J$,
we have $\I{\psi_j} = \emptyset$
with $\I{\varphi_j} \neq \emptyset$.
Then $\I{\varphi_j \realto \psi_j} = \emptyset$
and $\I\delta = \emptyset$.
Hence we must be in the case of item~(\ref{item:proof:sem:log:compl:c-false:ded}).
Moreover, by induction hypothesis we have $\C(\varphi_j)$
and $\psi_j \thesis \False$,
so that we can derive $\delta \thesis \False$ using the rule $\ax{C}$
in Figure~\ref{fig:ded:realto}.

Assume now that we have $\I{\psi_j} \neq \emptyset$
for all $j \in J$ such that $\I{\varphi_j} \neq \emptyset$.

Given $j \in J$ such that $\I{\varphi_j} = \emptyset$,
by induction hypothesis we have $\varphi_j \thesis \False$,
and since $\True \thesis \left( \False \realto \psi_j \right)$,
we get $\True \thesis \left(\varphi_j \realto \psi_j \right)$.
Hence
$\bigwedge_{i \neq j}(\varphi_i \realto \psi_i) \thesisiff \delta$.

We can therefore reduce the case of 
$\bigwedge_{j \in J} (\varphi_j \realto \psi_j)$
where 
$\I{\psi_j} \neq \emptyset$ and $\I{\varphi_j} \neq \emptyset$
for all $j \in J$.
In, particular, the induction hypothesis yields $\C(\psi_j)$
and $\C(\varphi_j)$ for all $j \in J$.

Regarding item~(\ref{item:proof:sem:log:compl:c-false:c}),
if $\I\delta \neq \emptyset$,
then for all $I \sle J$ we have either
$\BigI{\bigwedge_{i \in I} \varphi_i} = \emptyset$
or
$\BigI{\bigwedge_{i \in I} \psi_i} \neq \emptyset$,
so that either
$\bigwedge_{i \in I} \varphi_i \thesis \False$
or
$\C(\bigwedge_{i \in I} \psi_i)$
by induction hypothesis.
We can thus obtain $\C(\delta)$ by using the last rule in
Figure~\ref{fig:consist}.

Concerning item~(\ref{item:proof:sem:log:compl:c-false:ded}),
if $\I\delta = \emptyset$,
then there is some $I \sle J$
such that
$\BigI{\bigwedge_{i \in I} \varphi_i} \neq \emptyset$
while
$\BigI{\bigwedge_{i \in I} \psi_i} = \emptyset$.
Hence
$\C(\bigwedge_{i \in I} \varphi_i)$
and
$\bigwedge_{i \in I} \psi_i \thesis \False$
by induction hypothesis.
We have
\[
\begin{array}{l l l}
  \delta
& \thesis
& \bigwedge_{j \in J}
  \left(
  \left( \bigwedge_{i \in I} \varphi_i \right)
  \realto
  \psi_j
  \right)
\end{array}
\]

\noindent
and thus
\[
\begin{array}{l l l}
  \delta
& \thesis
& \left( \bigwedge_{i \in I} \varphi_i \right)
  \realto
  \bigwedge_{i \in I} \psi_i
\end{array}
\]

\noindent
and we obtain $\delta \thesis \False$ using the rule $\ax{C}$
in Figure~\ref{fig:ded:realto}.
\qedhere
\end{description}
\end{proof}

We now turn to Theorem~\ref{thm:sem:log:compl}(\ref{item:sem:log:compl:ded}).
We shall use the following preliminary result.

\begin{lemma}
\label{lem:proof:sem:log:compl:true}
Given $\delta \in \Lang^\land(\PT)$,
if $\I\delta = \I\True$,
then $\True \thesis \delta$ is derivable.
\end{lemma}

\begin{proof}
By the Soundness Theorem~\ref{thm:sem:log:sound},
if $\delta \thesis \False$ then $\I\delta = \emptyset$.
We can thus assume $\delta \not\thesis \False$.
As consequence, we are done if we show that $\I\delta \neq \I\True$
whenever $\delta \not\thesis \False$ and $\True \not\thesis \delta$.
The proof is by induction in the size of $\delta$,
using Lemma~\ref{lem:proof:sem:log:compl:nf}.
We reason by cases on $\PT$.
\begin{description}
\item[Case of $\Unit$.]
By Lemma~\ref{lem:proof:sem:log:compl:nf},
we have $\delta \thesisiff \form{\pair{}}$.
Hence we are done since $\I{\form{\pair{}}} = \{\top \}$
while $\I\Unit = \{\bot,\top\}$.

\item[Case of $\PT_1 + \PT_2$.]
By Lemma~\ref{lem:proof:sem:log:compl:nf},
we have $\delta \thesisiff \form{\inj_i} \psi$
for some $i \in \{1,2\}$.
Hence $\I\delta \neq \I\True$.

\item[Case of $\rec\TV.\PT$.]
By Lemma~\ref{lem:proof:sem:log:compl:nf},
we have $\delta \thesisiff \form\fold\psi$
with $|\psi| < |\delta|$.

If $\psi \thesis \False$
then $\form{\fold} \psi \thesis \form{\fold}\False$,
so that $\delta \thesis \False$
since $\form{\fold}\False \thesis \False$.
Similarly, 
if $\True \thesis \psi$
then $\True \thesis \delta$
since $\True \thesis \form{\fold}\True$.

Hence, we have $\True \not\thesis \psi$
and $\psi \not\thesis \False$.
The induction hypothesis yields $\I\psi \neq \I\True$,
from which we get $\I\delta \neq \I\True$
since $\I\unfold$ is an iso.

\item[Case of $\PT_1 \times \PT_2$.]
By Lemma~\ref{lem:proof:sem:log:compl:nf},
we have
$\delta \thesisiff \form{\pi_1}\psi_1 \land \form{\pi_2}\psi_2$
with
\(
  |\form{\pi_1}\psi_1 \land \form{\pi_2}\psi_2| \leq |\delta|
\).

Assume toward a contradiction that $\psi_i \thesis \False$.
Then $\delta \thesis \False$ since $\form{\pi_i}\False \thesis \False$.
Hence we have
$\psi_1 \not\thesis \False$ and $\psi_2 \not\thesis \False$.
Moreover, since $\True \thesis \form{\pi_i}\True$
and $\True \not\thesis \delta$,
must have either $\True \not\thesis \psi_1$
or $\True \not\thesis \psi_2$.

Assume (say) $\True \not\thesis \form{\pi_i} \psi_i$.
Then by induction hypothesis we get
$\I{\psi_i} \neq \I\True$,
which implies $\I\delta \neq \I\True$.

\item[Case of $\PTbis \arrow \PT$.]
By Lemma~\ref{lem:proof:sem:log:compl:nf},
there is a finite set $J$ such that
\[
\begin{array}{l l l}
  \delta
& \thesisiff
& \bigwedge_{j \in J} \left(\varphi_j \realto \psi_j\right)
\end{array}
\]

\noindent
with
\(
  |\bigwedge_{j \in J} \left(\varphi_j \realto \psi_j\right)|
  \leq
  |\delta|
\).

Let $i \in J$ such that $\True \thesis \psi_i$.
Hence $\True \thesis \varphi_i \realto \psi_i$
and
$\delta \thesisiff \bigwedge_{j \neq i} (\varphi_j \realto \psi_j)$.
Since $\True \not\thesis \delta$,
we can thus reduce to the case where
$\True \not\thesis \psi_j$ for all $j \in J$.

Similarly, let $i \in J$ such that $\varphi_i \thesis \False$.
Hence $\True \thesis \varphi_i \realto \psi_i$
and
$\delta \thesisiff \bigwedge_{j \neq i} (\varphi_j \realto \psi_j)$.
Since $\True \not\thesis \delta$,
we can thus reduce to the case where
$\varphi_i \not\thesis \False$ for all $j \in J$.
In particular, 
by Lemma~\ref{lem:proof:sem:log:compl:c-false}
we have $\I{\varphi_j} \neq \emptyset$ and $\C(\varphi_j)$
for all $j \in J$.

Moreover, 
$\psi_j \thesis \False$
would imply $\delta \thesis \False$
using the rule $\ax{C}$ in Figure~\ref{fig:ded:realto}.
We thus have $\psi_j \not\thesis \False$ for all $j \in J$.

Assume toward a contradiction that $\I\delta = \I\True$.
Hence in particular $(\bot \step \bot) \in \I\delta$,
so that $(\bot \step \bot) \in \I{\varphi_j \realto \psi_j}$
for all $j \in J$.
For each $j \in J$,
since $\I{\varphi_j} \neq \emptyset$
we have $\bot \in \I{\psi_j}$.
But this implies $\I{\psi_j} = \I\True$,
which by induction hypothesis contradicts that
$\True \not\thesis \psi_j$
and
$\psi_j \not\thesis \False$.
\qedhere
\end{description}
\end{proof}

\begin{lemma}
\label{lem:proof:sem:log:compl:ded}
For all $\varphi,\psi \in \Lang^\land(\PT)$,
if $\I\psi \sle \I\varphi$,
then $\psi \thesis \varphi$ is derivable.
\end{lemma}

\begin{proof}
The proof is by induction on the sum of the 
sizes of $\psi$ and $\varphi$.

First, note that the result is trivial if either
$\True \thesis \varphi$
or $\psi \thesis \False$.
%
Moreover, if $\varphi \thesis \False$
then $\I\varphi = \emptyset$ by the Soundness Theorem~\ref{thm:sem:log:sound},
and we also have $\I\psi = \emptyset$,
so that $\psi \thesis \False$
by Lemma~\ref{lem:proof:sem:log:compl:c-false}.
%
Furthermore, if $\True \thesis \psi$,
then $\I\psi = \I\True$,
so that $\I\varphi = \I\True$ as well,
and Lemma~\ref{lem:proof:sem:log:compl:true}
implies $\True \thesis \varphi$.

Altogether, we can thus assume
$\True \not\thesis \varphi$,
$\varphi \not\thesis \False$,
$\True \not\thesis \psi$
and
$\psi \not\thesis \False$.
Hence Lemma~\ref{lem:proof:sem:log:compl:nf}
applies to both $\varphi$ and $\psi$.

Moreover, note that if $\varphi = \varphi_1 \land \varphi_2$,
then $\I\psi \sle \I\varphi$ implies
$\I\psi \sle \I{\varphi_i}$ for all $i = 1,2$,
and we can obtain $\psi \thesis \varphi$ from the induction hypotheses.
We can thus assume that $\varphi$ is not a conjunction.

We now reason by cases on $\PT$.
\begin{description}
\item[Case of $\Unit$.]
By Lemma~\ref{lem:proof:sem:log:compl:nf}
we have
$\psi \thesisiff \form{\pair{}}$
and $\varphi \thesisiff \form{\pair{}}$.
Hence $\psi \thesis \varphi$.

\item[Case of $\rec \TV.\PT$.]
By Lemma~\ref{lem:proof:sem:log:compl:nf}
we have
$\psi \thesisiff \form{\fold} \psi'$
and $\varphi \thesisiff \form{\fold} \varphi'$
with $|\psi'| < |\psi|$
and $|\varphi'| < |\varphi|$.
Since moreover $\I{\psi'} \sle \I{\varphi'}$,
we can conclude by induction hypothesis.

\item[Case of $\PT_1 + \PT_2$.]
By Lemma~\ref{lem:proof:sem:log:compl:nf}
we have
$\psi \thesisiff \form{\inj_i} \psi'$
and $\varphi \thesisiff \form{\inj_j} \varphi'$
with $|\psi'| < |\psi|$
and $|\varphi'| < |\varphi|$.
Since $\I\psi \sle \I\varphi$,
we must have $i = j$,
and the result follows using the induction hypothesis.

\item[Case of $\PT_1 \times \PT_2$.]
By Lemma~\ref{lem:proof:sem:log:compl:nf}
we have
$\psi \thesisiff \form{\pi_1} \psi'_1 \land \form{\pi_2} \psi'_2$
and $\varphi \thesisiff \form{\pi_i} \varphi'$
with $|\psi'_1|, |\psi'_2| < |\psi|$
and $|\varphi'| < |\varphi|$.

Since $\I\psi \sle \I\varphi$,
we must have
$\I{\psi'_i} \sle \I{\varphi'}$,
so that
$\psi'_i \thesis \varphi'$ by induction hypothesis,
and thus
$\form{\pi_i}\psi'_i \thesis \form{\pi_i}\varphi'$.
We then obtain $\psi \thesis \varphi$, as expected.

\item[Case of $\PTbis \arrow \PT$.]
This is again the most important and most difficult case.
The proof is an adaptation to our setting
of the proof of~\cite[Proposition 10.5.2]{ac98book}.

First, 
by Lemma~\ref{lem:proof:sem:log:compl:nf}
the formula
$\varphi$ must be of the form $\varphi'' \realto \varphi'$.
If $\I{\varphi''} = \emptyset$,
then by Lemma~\ref{lem:proof:sem:log:compl:c-false}
we obtain $\varphi'' \thesis \False$,
so that $\True \thesis (\varphi'' \realto \varphi')$.
Hence $\psi \thesis \varphi$ in this case.
We can thus assume $\I{\varphi''} \neq \emptyset$.
Since $\varphi \not\thesis \False$,
by Lemma~\ref{lem:proof:sem:log:compl:c-false}
we get
$\I\varphi \neq \emptyset$,
which implies that $\I{\varphi'} \neq \emptyset$ as well.
Hence, by Proposition~\ref{prop:sem:char:fin}
there are finite $d'',d'$
such that $\I{\varphi''} = \up d''$
and $\I{\varphi'} = \up d'$.

On the other hand,
Lemma~\ref{lem:proof:sem:log:compl:nf} yields that
$\psi$ is of the form
$\bigwedge_{i \in I} (\psi''_i \realto \psi'_i)$
for some finite set $I$,
with
$|\bigwedge_{i \in I} (\psi''_i \realto \psi'_i)| \leq |\psi|$.
Furthermore, since $\True \not\thesis \psi$,
reasoning as in the proof of Lemma~\ref{lem:proof:sem:log:compl:true},
we can reduce to the case where 
$\I{\psi''_i} \neq \emptyset$ for all $i \in I$.
Lemma~\ref{lem:proof:sem:log:compl:c-false}
then yields $\C(\psi''_i)$ for all $i \in I$.
Since moreover $\psi \not\thesis \False$,
Lemma~\ref{lem:proof:sem:log:compl:c-false} (again) yields
$\I{\psi'_i} \neq \emptyset$ for all $i \in I$.
Again by Proposition~\ref{prop:sem:char:fin},
for each $i \in I$ there are finite $e''_i,e'_i$
such that $\I{\psi''_i} = \up e''_i$
and $\I{\psi'_i} = \up e'_i$.
Moreover, $\bigvee_{i \in I}(e''_i \step e'_i)$
exists since $\I\psi \neq \emptyset$
by Lemma~\ref{lem:proof:sem:log:compl:c-false}
(as $\psi \not\thesis \False$).

Hence $\I\psi \sle \I\varphi$ means
\[
\begin{array}{l l l}
  \up \bigvee_{i \in I} \left(e''_i \step e'_i \right)
& \sle
& \up \left( d'' \step d' \right)
\end{array}
\]

\noindent
which implies
\[
\begin{array}{l l l}
  d'' \step d'
& \leq
& \bigvee_{i \in I} \left(e''_i \step e'_i \right)
\end{array}
\]

\noindent
We now evaluate the two functions above at $d''$.
Since
\[
\begin{array}{l l l}
  \left( \bigvee_{i \in I} e''_i \step e'_i \right)(x)
& =
& \bigvee \{ e'_i \mid x \geq e''_i\}
\end{array}
\]

\noindent
we get
\[
\begin{array}{l l l}
  d'
& \leq
& \bigvee_{d'' \geq e''_i} e_i
\end{array}
\]

\noindent
that is
\[
\begin{array}{l l l}
  \up \bigvee_{\up d'' \sle \up e''_i} e'_i
& \sle
& \up d'
\end{array}
\]

\noindent
In other words,
\[
\begin{array}{l l l}
  \I{\bigwedge_{\I{\varphi''} \sle \I{\psi ''_i}} \psi'_i}
& \sle
& \I{\varphi'}
\end{array}
\]

\noindent
and by induction hypothesis
\[
\begin{array}{l l l}
  \bigwedge_{\varphi'' \thesis \psi ''_i} \psi'_i
& \thesis
& \varphi'
\end{array}
\]

Hence we are done since
\[
\begin{array}{l l l}
  \psi
& \thesis
& \bigwedge_{i \in I} \left(
  \left( \bigwedge_{\varphi'' \thesis \psi ''_i} \psi''_i \right)
  \realto
  \psi'_i
  \right)
\end{array}
\]

\noindent
and thus
\[
\begin{array}{l l l}
  \psi
& \thesis
& \left( \bigwedge_{\varphi'' \thesis \psi ''_i} \psi''_i \right)
  \realto
  \bigwedge_{\varphi'' \thesis \psi ''_i} \psi'_i
\end{array}
\]

\noindent
while
\(
\varphi'' \thesis \bigwedge_{\varphi'' \thesis \psi ''_i} \psi''_i
\).
\qedhere
\end{description}
\end{proof}

We note here that the (closed) neutral fragment is complete.
We do not need this result.

\begin{corollary}
Let $\varphi, \psi \in \Lang^\pm(\PT)$ be \emph{closed}.
If $\I\psi \sle \I\varphi$,
then $\psi \thesis \varphi$ is derivable.
\end{corollary}

\begin{proof}
By Lemma~\ref{lem:ded:omega}
and Lemma~\ref{lem:ded:omegalorland}
(together with Soundness Theorem~\ref{thm:sem:log:sound}),
we can assume $\varphi,\psi \in \Lang^{\lor\land}(\PT)$.
If $\psi$ is a disjunction, say $\psi_1 \lor \psi_2$,
then $\I\psi \sle \I\varphi$
implies
$\I{\psi_1} \sle \I\varphi$
and
$\I{\psi_2} \sle \I\varphi$.
We can thus reduce to the case where $\psi \in \Lang^\land(\PT)$.
If $\I\psi = \emptyset$,
then Lemma~\ref{lem:proof:sem:log:compl:c-false}
yields that $\psi \thesis \False$,
and we are done.

Otherwise, by Proposition~\ref{prop:sem:char:fin}
there is some finite $d \in \I\PT$
such that $\I\psi = \up d$.
Hence $\I\psi \sle \I\varphi$
implies $d \in \I\varphi$.
If $\varphi$ is a disjunction, say $\varphi_1 \lor \varphi_2$,
then either $d \in \I{\varphi_1}$ and $\I\psi \sle \I{\varphi_1}$,
or $d \in \I{\varphi_2}$ and $\I\psi \sle \I{\varphi_2}$.
We can thus reduce to the case were $\varphi \in \Lang^\land(\PT)$,
and conclude with Lemma~\ref{lem:proof:sem:log:compl:ded}.
\end{proof}

}
\opt{full}{
\section{Proofs of \S\ref{sec:sem:reft} (\nameref{sec:sem:reft})}
\label{sec:proof:sem:reft}

We begin with the (long) proof of Theorem~\ref{thm:compl:land}.

\ThmComplLand*

\begin{proof}
Note that Lemma~\ref{lem:reft} 
restricts to types $\RTter^\land$, in the sense that given $\RTter^\land$,
there is some $\varphi \in \Lang^\land(\UPT{\RTter^\land})$
such that $\RTter^\land \eqtype \reft{\UPT{\RTter^\land} \mid \varphi}$.

In the following, we write $\Env \thesis M \colon \RT$
for $\Env^\land \thesis M \colon \RT^\land$.

Let $\varphi \in \Lang^\land(\UPT\RT)$ such that
$\RT \eqtype \reft{\UPT\RT \mid \varphi}$.
Assume
$\Env = x_1\colon \RTbis_1,\dots,x_n \colon \RTbis_n$
and
let $\psi_i \in \Lang^\land(\UPT{\RTbis_i})$
such that
$\RTbis_i \eqtype \reft{\UPT{\RTbis_i} \mid \psi_i}$
for each $i = 1,\dots,n$.

Note that if
$\I{\RTbis_i} = \emptyset$ for some $i \in \{1,\dots,n\}$,
since $\I{\UPT{\RTbis_i}}$ is not empty (as it is a Scott domain),
we must have $\I\psi_i = \emptyset$
and thus $\psi_i \thesis \False$ by
completeness for $\Lang^\land$
(Theorem~\ref{thm:sem:log:compl}).
We can then conclude with the rule for $\False$ in Figure~\ref{fig:reft:log}
(\S\ref{sec:reft}).
We can thus reduce the case of a typing context
such that 
$\I{\RTbis_i} \neq \emptyset$
(or equivalently $\psi_i \not\thesis \False$)
for each $i = 1,\dots,n$.
We call such typing contexts \emph{consistent}.

Note also that if $\True \thesis \varphi$,
then we can derive $\Env \thesis M \colon \RT$
from $\UPT\Env \thesis M \colon \UPT\RT$ by subtyping.
We can thus assume $\True \not\thesis \varphi$.
Moreover, since $\Env$ is consistent we must have $\I\varphi \neq \emptyset$,
that is, $\varphi \not\thesis \False$.
We can thus always apply Lemma~\ref{lem:proof:sem:log:compl:nf} to $\varphi$.

We now reason by induction on the typing derivation of
$\UPT\Env \thesis M \colon \UPT\RT$.
\begin{description}
\item[Case of]
\[
\dfrac{(x\colon\UPT\RT) \in \UPT\Env}
  {\UPT\Env \thesis x\colon\UPT\RT}
\]

We have $(x \colon \RTbis) \in \Env$ for some type $\RTbis$ with $\UPT\RTbis = \UPT\RT$.
Let $\varphi,\psi \in \Lang^\land(\UPT\RT)$
such that $\RT \eqtype \reft{\UPT\RT \mid \varphi}$
and $\RTbis \eqtype \reft{\UPT\RT \mid \psi}$.
Since $\Env \thesis M :\RT$ is sound,
we have $\I\psi \sle \I\varphi$.
Hence $\psi \thesis \varphi$
by completeness for $\Lang^\land$
(Theorem~\ref{thm:sem:log:compl}).
We then conclude by subtyping.

\item[Case of]
\[
\dfrac{}
  {\UPT\Env \thesis \pair{} \colon \UPT\RT}
\]

We have $\UPT\RT = \Unit$
and $\varphi \thesisiff \form{\pair{}}$.
We conclude with the rule
\[
\dfrac{}
  {\Env \thesis \pair{} \colon \reft{\Unit \mid \form{\pair{}}}}
\]

\item[Case of]
\[
\dfrac{\UPT\Env \thesis N_1 \colon \PT_1
  \qquad
  \UPT\Env \thesis N_2 \colon \PT_2}
  {\UPT\Env \thesis \pair{N_1,N_2} \colon \PT_1 \times \PT_2}
\]

\noindent
where $\UPT\RT = \PT_1 \times \PT_2$
and $M = \pair{N_1,N_2}$.
We have
$\varphi \thesisiff \form{\pi_1}\varphi_1 \land \form{\pi_2}\varphi_2$.
The judgements
$\Env \thesis N_1 \colon \reft{\PT_1 \mid \varphi_1}$
and
$\Env \thesis N_2 \colon \reft{\PT_2 \mid \varphi_2}$
are both sound since so is
\(
  \Env
  \thesis
  \pair{N_1,N_2} \colon
  \reft{\PT_1 \times \PT_2 \mid \form{\pi_1}\varphi_1 \land \form{\pi_2}\varphi_2}
\).
We can thus conclude with the induction hypothesis and
\[
\text{
\AXC{$\vdots$}
\UIC{$\Env \thesis N_1 \colon \reft{\PT_1 \mid \varphi_1}$}
\UIC{$\Env \thesis \pair{N_1,N_2} \colon
  \reft{\PT_1 \times \PT_2 \mid \form{\pi_1} \varphi_1}$}
\AXC{$\vdots$}
\UIC{$\Env \thesis N_2 \colon \reft{\PT_2 \mid \varphi_2}$}
\UIC{$\Env \thesis \pair{N_1,N_2} \colon
  \reft{\PT_1 \times \PT_2 \mid \form{\pi_1} \varphi_2}$}
\BIC{$\Env \thesis \pair{N_1,N_2} \colon
  \reft{\PT_1 \times \PT_2 \mid \form{\pi_1} \varphi_1 \land \form{\pi_2} \varphi_2}$}
\DisplayProof}
\]

\item[Case of]
\[
\begin{array}{c}
\dfrac{\UPT\Env \thesis N \colon \PT \times \PTbis}
  {\UPT\Env \thesis \pi_i(N) \colon \UPT\RT}
\end{array}
\]

\noindent
where $i = 1,2$ and $M = \pi_1(N)$.
Assume w.l.o.g.\ $i = 1$ (so that $\UPT\RT = \PT$).
Note that
$\Env \thesis N \colon \reft{\PT \times \PTbis \mid \form{\pi_1}\varphi}$
is sound since so is
$\Env \thesis \pi_1(N) \colon \reft{\PT \mid \varphi}$.
We then conclude with the induction hypothesis and the rule
\[
\dfrac{\Env \thesis N \colon \reft{\PT \times \PTbis \mid \form{\pi_1} \varphi}}
  {\Env \thesis \pi_1(N) \colon \reft{\PT \mid \varphi}}
\]

\item[Case of]
\[
\dfrac{\UPT\Env \thesis N \colon \PT[\rec\TV.\PT / \TV]}
  {\UPT\Env \thesis \fold(N) \colon \rec\TV.\PT}
\]

\noindent
where $\UPT\RT = \rec\TV.\PT$ and $M = \fold(M)$.
We have $\varphi \thesisiff \form{\fold}\psi$.
Moreover, 
$\Env \thesis N \colon \reft{\PT[\rec\TV.\PT / \TV] \mid \psi}$
is sound
since so is
$\Env \thesis \fold(N) \colon \reft{\rec\TV.\PT \mid \form\fold\psi}$.
We can thus conclude using the induction hypothesis and the rule
\[
\dfrac{\Env \thesis N \colon \reft{\PT[\rec\TV.\PT/\TV] \mid \psi}}
  {\Env \thesis \fold(N) \colon \reft{\rec\TV.\PT \mid \form\fold \psi}}
\]

\item[Case of]
\[
\dfrac{\UPT\Env \thesis N \colon \rec\TV.\PT}
  {\UPT\Env \thesis \unfold(N) \colon \PT[\rec\TV.\PT / \TV]}
\]

\noindent
where $\UPT\RT = \PT[\rec\TV.\PT/\TV]$
and $M = \unfold(N)$.
The judgement 
$\Env \thesis N \colon \reft{\rec\TV.\PT \mid \form\fold \varphi}$
is sound
since so is
$\Env \thesis \unfold(N) \colon \reft{\PT[\rec\TV.\PT / \TV] \mid  \varphi}$
and we conclude using the induction hypothesis and the rule
\[
\dfrac{\Env \thesis N \colon \reft{\rec\TV.\PT \mid \form\fold \varphi}}
  {\Env \thesis \unfold(N) \colon \reft{\PT[\rec\TV.\PT/\TV] \mid \varphi}}
\]

\item[Case of]
\[
\dfrac{\UPT\Env,x\colon\PTbis \thesis N \colon \PT}
  {\UPT\Env \thesis \lambda x.N \colon \PTbis \arrow \PT}
\]

\noindent
where $\UPT\RT = \PTbis \arrow \PT$,
and where $M = \lambda x.N$.

We have $\varphi \thesisiff \bigwedge_{i \in I}(\varphi''_i \realto \varphi'_i)$
for some finite set $I$.
Let $i \in I$.
The judgement
\(
  \Env
  \thesis
  \lambda x.N \colon
  \reft{\PTbis \arrow \PT \mid \varphi''_i \realto \varphi'_i}
\)
is sound,
and so is
\(
  \Env, x \colon \reft{\PTbis \mid \varphi''_i}
  \thesis
  N \colon
  \reft{\PT \mid \varphi'_i}
\).
Using the induction hypothesis, we derive
\(
  \Env
  \thesis
  \lambda x.N \colon
  \reft{\PTbis \arrow \PT \mid \varphi''_i \realto \varphi'_i}
\).
We can then obtain
\(
  \Env
  \thesis
  \lambda x.N \colon
  \reft{\PTbis \arrow \PT \mid \varphi}
\).

\item[Case of]
\[
\dfrac{\UPT\Env \thesis N \colon \PTbis \arrow \PT
  \qquad
  \UPT\Env \thesis V \colon \PTbis}
  {\UPT\Env \thesis N V \colon \PT}
\]

\noindent
where $\UPT\RT = \PT$ and where $M = N V$.
For each $i \in \{1,\dots,n\}$,
since $\I{\psi_i} \neq \emptyset$,
by Proposition~\ref{prop:sem:char:fin}
there is a finite $e_i \in \I{\UPT{\RTbis_i}}$
such that $\I{\psi_i} = \up e_i$.
Similarly, since $\I\varphi \neq \emptyset$,
again by Proposition~\ref{prop:sem:char:fin}
there is a finite $d \in \I{\PT}$
such that $\I{\varphi} = \up d$.

We have
$\I M(\vec e) \in \I\varphi$.
Note that
$\I M(\vec e) = \I N(\vec e)\left(\I V(\vec e) \right)$.
Since $\I\PTbis$ is a Scott domain, it is algebraic,
and $\I V(\vec e)$ is the directed l.u.b.\ of the finite $e \leq \I V(\vec e)$.
Since $\I N(\vec e)$ is Scott-continuous, we thus get
that $\I M(\vec e)$ is the l.u.b.\ of the directed set
\[
\left\{
  \I N(\vec e)(e)
  \mid
  \text{$e$ finite and $\leq \I V(\vec e)$}
\right\}
\]

\noindent
Since $d \leq \I M(\vec e)$ and since $d$ is finite,
it follows that we have
$d \leq \I N(\vec e)(e)$ for some finite $e \leq \I V(\vec e)$.
By Proposition~\ref{prop:sem:char:fin},
there is a formula $\psi \in \Lang^\land(\PTbis)$
such that $\I\psi = \up e$.

Since $d \leq \I N(\vec e)(e)$, we have
$(e \step d) \leq \I N(\vec e)$,
so that
$\I N(\vec e) \in \I{\psi \realto \varphi}$.
Since $\I N$ is monotone, it follows that 
$\Env \thesis N \colon \reft{\PTbis \arrow \PT \mid \psi \realto \varphi}$
is sound.
Hence, this judgement is derivable by induction hypothesis.

Similarly, since $e \leq \I V(\vec e)$,
we obtain that the judgement
$\Env \thesis V \colon \reft{\PTbis \mid \psi}$
is sound and thus derivable.

We can then easily derive $\Env \thesis M \colon \reft{\PT \mid \varphi}$.

\item[Case of]
\[
\dfrac{\UPT\Env,x\colon\UPT\RT \thesis N \colon \UPT\RT}
  {\UPT\Env \thesis \fix x.N \colon \UPT\RT}
\]

\noindent
where $M = \fix x.N$.
Similarly as above, for each $i \in \{1,\dots,n\}$
there is a finite $e_i \in \I{\psi_i}$
such that $\I{\RTbis_i} = \up e_i$.
There is also
a finite $d \in \I{\UPT\RT}$ such that $\I\varphi = \up d$.

Let $f \colon \I{\UPT\RT} \to \I{\UPT\RT}$
be the Scott-continuous function which takes
$a \in \I{\UPT\RT}$ to $\I N(\vec e,a)$.
We have
\[
\begin{array}{l l l}
  \I{\fix x.N}(\vec e)
& =
& \bigvee_{k \in \NN}
  f^k(\bot)
\end{array}
\]

Since $d \leq \I{\fix x.N}(\vec e)$ with $d$ finite,
there is some $k \in \NN$
such that
$d \leq f^k(\bot)$.
Write $d_k$ for $d$.
Since $\I{\UPT\RT}$ is algebraic,
by induction, for each $j = k-1,\dots,0$
there is some finite $d_j$ such that 
$d_{j+1} \leq f(d_j)$
and
$d_j \leq f^{j}(\bot)$.
In particular, $d_0 = \bot$.
For each $j = 0,\dots,k$,
let $\varphi_j$ such that $\I{\varphi_j} = \up d_j$.
Note that $\varphi_k = \varphi$.
Moreover, since $d_0 = \bot$,
we can take $\varphi_0 = \True$.

Again reasoning similarly as above,
we obtain that 
\(
  \Env, x \colon \reft{\UPT\RT \mid \varphi_j} \thesis
  N \colon \reft{\UPT\RT \mid \varphi_{j+1}}
\)
is sound and thus derivable for each $j = 0,\dots,k-1$.
Moreover,
$\Env \thesis \fix x.N \colon \reft{\UPT\RT \mid \varphi_0}$
is derivable.
We can then derive $\Env \thesis \fix x.N \colon \reft{\UPT\RT \mid \varphi}$
by iterated applications of the rule
\[
\dfrac{\Env \thesis \fix x.N \colon \reft{\PT \mid \psi}
  \qquad
  \Env, x: \reft{\PT \mid \psi} \thesis N \colon \reft{\PT \mid \psi'}}
  {\Env \thesis \fix x.N \colon \reft{\PT \mid \psi'}}
\]

\item[Case of]
\[
\dfrac{\UPT\Env \thesis N \colon \PT_i}
  {\UPT\Env \thesis \inj_i(N) \colon \PT_1 + \PT_2}
\]

\noindent
where $\UPT\RT = \PT_1 + \PT_2$ and $M = \inj_i(N)$.
We have $\varphi \thesisiff \form{\inj_j}\psi$,
and since $\Env \thesis M\colon \RT$ is sound,
we must have $i = j$.
But then $\Env \thesis N \colon \reft{\PT_i \mid \psi}$
is also sound, and we can conclude by induction hypothesis and
the rule
\[
\dfrac{\Env \thesis N \colon \reft{\PT_i \mid \psi}}
  {\Env \thesis \inj_i(N) \colon \reft{\PT_1 + \PT_2 \mid \form{\inj_i} \psi}}
\]

\item[Case of]
\[
\dfrac{\UPT\Env \thesis N \colon \PT_1 + \PT_2
  \qquad \UPT\Env, x_1 \colon \PT_1 \thesis N_1 \colon \PTbis
  \qquad \UPT\Env, x_2 \colon \PT_2 \thesis N_2 \colon \PTbis}
  {\UPT\Env \thesis \cse\ N\ \copair{x_1 \mapsto N_1 \mid x_2 \mapsto N_2} \colon \PTbis}
\]

\noindent
where $\UPT\RT = \PTbis$
and where $M$ is the term
$(\cse\ N\ \copair{x_1 \mapsto N_1 \mid x_2 \mapsto N_2})$.

Similarly as in the cases of $\fix x.N$ and $N V$ above,
for each $i \in \{1,\dots,n\}$
there is a finite $e_i \in \I{\UPT{\RTbis_i}}$
such that $\I{\psi_i} = \up e_i$.
Also, there is a finite $d \in \I\PTbis$ such that $\I\varphi = \up d$.

Since $\I\varphi \neq \I\True$,
we must have $\bot \notin \I\varphi$,
so that $\I M(\vec e) \neq \bot$
and thus $\I N(\vec e) \neq \bot$.
Hence $\I N(\vec e) = \I{\inj_i}(z)$ for some $z \in \I{\PT_i}$.
Since $\I M(\vec e) = \I{N_i}(\vec e,z) \in \up d$,
by algebraicity there is a finite $e \leq z$
such that $d \leq \I{N_i}(\vec e,e)$.
With $\psi \in \Lang^\land(\PT_i)$ such that
$\I\psi = \up e$,
we get that the judgements
$\Env \thesis N \colon \reft{\PT_1 + \PT_2 \mid \form{\inj_i}\psi}$
and
$\Env, x\colon \reft{\PT_i \mid \psi} \thesis N_i \colon \reft{\PTbis \mid \varphi}$
are sound.
We can then conclude using the induction hypothesis and the rule
\[
\dfrac{
  \Env \thesis N \colon \reft{\PT_1 + \PT_2 \mid \form{\inj_i}\psi}
  \qquad
  \Env, x_i \colon \reft{\PT_i \mid \psi} \thesis N_i \colon \RT
  \qquad
  \UPT\Env, x_{3-i} \colon \PT_{3-i} \thesis N_{3-i} \colon \UPT\RT}
  {\Env \thesis \cse\ N\ \copair{x_1 \mapsto N_1 \mid x_2 \mapsto N_2} \colon \RT}
\]
\qedhere
\end{description}
\end{proof}

We now prove our main result, namely the Positive Completeness
Theorem~\ref{thm:compl}.

\ThmCompl*

\begin{proof}
Write $\Env \thesis M \colon \RT$ for $\Env^- \thesis M \colon \RT^+$.
Let $\Env = x_1\colon \RTbis_1,\dots,x_n \colon \RTbis_n$.
By Lemma~\ref{lem:reft} and
Proposition~\ref{prop:ded:prenex},
for each $i \in \{1,\dots,n\}$
there is a $\psi_i \in \Lang^\pm(\UPT{\RTbis_i})$
such that $\RTbis_i \eqtype \reft{\UPT{\RTbis_i} \mid (\forall \itvarbis_i)\psi_i}$.
We similarly get a $\varphi \in \Lang^\pm(\UPT\RT)$
such that $\RT \eqtype \reft{\UPT\RT \mid (\exists \itvar)\varphi}$.

Let $\mathcal{Q} = \{S_m \mid m \in \NN\}$,
where for each $m \in \NN$,
\[
\begin{array}{l l l}
  S_m
& =
& \I{\psi_1[m / \itvarbis_1]}
  \times
  \dots
  \times
  \I{\psi_n[m / \itvarbis_n]}
\end{array}
\]

\begin{claim}
$\mathcal{Q}$ is a codirected family of compact saturated sets.
Moreover
\[
\begin{array}{l l l}
  \bigcap \mathcal{Q}
& =
& \I{\RTbis_1}
  \times \dots \times
  \I{\RTbis_n}
\end{array}
\]
\end{claim}

\begin{claimproof}
It is clear that each $S_m \in \mathcal{Q}$
is compact-saturated as a (finite) product of compact-saturated sets
(see e.g.~\cite[Theorem 4.5.12]{goubault13book}).
Moreover, $\mathcal{Q}$ is codirected since (it is non-empty and)
$m \leq p$ implies
$S_p \sle S_m$ by antimonotonicity
(Lemma~\ref{lem:sem:log:posneg}).

Finally, the equality follows from the fact that
\[
\begin{array}{l l l}
  (u_1,\dots,u_n) \in \bigcap \mathcal{Q}
& \text{iff}
& \forall m \in \NN,~ (u_1,\dots,u_n) \in S_m
\\

& \text{iff}
& \forall m \in \NN,~
  \forall i \in \{1,\dots,n\},~
  u_i \in \I{\psi_i[m / \itvarbis_i]}
\\

& \text{iff}
& \forall i \in \{1,\dots,n\},~
  \forall m \in \NN,~
  u_i \in \I{\psi_i[m / \itvarbis_i]}
\\

& \text{iff}
& \forall i \in \{1,\dots,n\},~
  u_i \in \I{\RTbis_i}
\end{array}
\]
\end{claimproof}

Since $\Env \thesis M \colon \RT$ is sound, we have
$\bigcap \mathcal{Q} \sle \I M^{-1}(\I\RT)$,
where $\I\RT$ is Scott-open and $M$ is Scott-continuous.
By well-filteredness (Proposition~\ref{prop:sem:wf}) and the Claim,
there is some $S_m \in \mathcal{Q}$ such that
$S_m \sle \I M^{-1}(\I\RT)$.
Hence there is some $m \in \NN$ such that the following is sound:
\begin{equation}
\label{eq:proof:sem:reft:neutral-left}
\begin{array}{l l l}
  x_1 \colon \reft{\UPT{\RTbis_1} \mid \psi_1[m / \itvarbis_1]}
  ,~ \dots,~
  x_n \colon \reft{\UPT{\RTbis_n} \mid \psi_n[m / \itvarbis_n]}
& \thesis
& M \colon \RT
\end{array}
\end{equation}

Now, by Lemmas~\ref{lem:ded:omega} and~\ref{lem:ded:omegalorland},
each $\psi_i[m/\itvar] \in \Lang^\pm(\UPT{\RTbis_i})$
is equivalent to a disjunction $\bigvee_{j \in J_i} \delta_{i,j}$,
where $\delta_{i,j} \in \Lang^\land(\UPT{\RTbis_i})$.
For each tuple $(j_1,\dots,j_n) \in J_1 \times \dots \times J_n$,
the judgement
\begin{equation}
\label{eq:proof:sem:reft:land-left}
\begin{array}{l l l}
  x_1 \colon \reft{\UPT{\RTbis_1} \mid \delta_{1,j_1}}
  ,~ \dots,~
  x_n \colon \reft{\UPT{\RTbis_n} \mid \delta_{n,j_n}}
& \thesis
& M \colon \RT
\end{array}
\end{equation}

\noindent
is sound.
If $\delta_{i,j_i} \thesis \False$, then we can
derive~\eqref{eq:proof:sem:reft:land-left}
using the rule for $\False$ in Figure~\ref{fig:reft:log} (\S\ref{sec:reft}).

Assume $\delta_{i,j_i} \not\thesis \False$ for all $i = 1,\dots,n$.
Hence by Proposition~\ref{prop:sem:char:fin},
for each $i = 1,\dots,n$
there is some finite $d_i \in \I{\UPT{\RTbis_i}}$
such that $\I{\delta_i} = \up d_i$.
The soundness of~\eqref{eq:proof:sem:reft:land-left}
then yields
$\I M(\vec d) \in \I\RT$.
Hence there is $p \in \NN$ such that $\I M(\vec d) \in \I{\varphi[p / \itvar]}$.
By Lemmas~\ref{lem:ded:omega} and~\ref{lem:ded:omegalorland} again,
the formula $\varphi[p / \itvar]$ is equivalent to a disjunction
of $\gamma \in \Lang^\land(\UPT\RT)$.
Since $\I M(\vec d) \in \I{\varphi[p / \itvar]}$,
we have $\I M(\vec d) \in \I{\gamma}$ for at least one of these $\gamma$'s.
It follows that the judgement 
\[
\begin{array}{l l l}
  x_1 \colon \reft{\UPT{\RTbis_1} \mid \delta_{1,j_1}}
  ,~ \dots,~
  x_n \colon \reft{\UPT{\RTbis_n} \mid \delta_{n,j_n}}
& \thesis
& M \colon \reft{\UPT\RT \mid \gamma}
\end{array}
\]

\noindent
is sound, and thus derivable by Theorem~\ref{thm:compl:land}.
We then derive~\eqref{eq:proof:sem:reft:land-left}
by subtyping,
using that $\gamma \thesis \varphi[p / \itvar]$.

Since we have~\eqref{eq:proof:sem:reft:land-left}
for each tuple $(j_1,\dots,j_n) \in J_1 \times \dots \times J_n$,
we then obtain~\eqref{eq:proof:sem:reft:neutral-left}
by multiple applications
of the rule for $\lor$ in Figure~\ref{fig:reft:log} (\S\ref{sec:reft}).
Then we are done by subtyping since
$(\forall \itvarbis_i)\psi_i \thesis \psi_i[m / \itvarbis_i]$
for each $i = 1,\dots,n$.
\end{proof}

}

\opt{full}{\setcounter{tocdepth}{4}}
\opt{full}{\newpage\tableofcontents}

\end{document}